\documentclass{article}
\usepackage[sort&compress,numbers]{natbib}%colon %%%% natbib package must be before algorithm2e
\usepackage[left=1.1in, right=1.1in, top=1.2in, bottom=1.2in]{geometry}
\usepackage{amsmath,amssymb,latexsym, amsopn, amsfonts}
\usepackage[ruled,vlined,linesnumbered]{algorithm2e}
\usepackage{framed}
\usepackage[usenames, dvipsnames]{color}
\usepackage{xcolor}
\usepackage{array}
\usepackage{graphicx}
\usepackage{setspace}
\usepackage[hidelinks]{hyperref} %% Hyperref after all packages

\newcommand{\bigO}{\mathrm{O}}
\newcommand{\remove}[1]{}

\newcommand{\algorithmFontSize}{footnotesize}
\newcommand{\macroFontSize}{footnotesize}

\newtheorem{remark}{Remark}[section]
\newtheorem{theorem}{Theorem}[section]
\newtheorem{lemma}[theorem]{Lemma}

\newtheorem{definition}{Definition}[section]
\newtheorem{argument}{Lemma}[section]

\newtheorem{corollary}[theorem]{Corollary}
\newtheorem{claim}[theorem]{Claim}

\newenvironment{Argument}{Lemma}    

\newenvironment{Arguments}{Lemmas}   

\SetKwInOut{Input}{input}
\SetKwInOut{Output}{output}

\newenvironment{proof}{\noindent{\bf Proof.}}{\hfill$\blacksquare$}

\newcommand{\Paragraph}[1]{\paragraph{\bf\emph{#1}}}

\newcommand{\bl}[1]{\textcolor{black}{#1}}
\newcommand{\trnsfr}[1]{\textcolor{black}{#1}}
\newcommand{\pend}[1]{\color{black}{#1}\color{black}}

\newcommand{\process}{maintainCntrs}

\newcommand{\seemCrd}{seemCrd}
\newcommand{\valCrd}{valCrd}
\newcommand{\majority}{maj}

\makeatletter
\newcommand{\RemoveAlgoNumber}{\renewcommand{\fnum@algocf}{\AlCapSty{\AlCapFnt\algorithmcfname}}}
\newcommand{\RevertAlgoNumber}{\algocf@resetfnum}
\makeatother

\date{}

\begin{document}

%\begin{frontmatter}

%\begin{center}
%%\begin{spacing}{1.5}
%\FFF\FFF\FFF\FFF\FFF\FFF\FFF\FFF\FFF\FFF\FFF
%\textsc{\huge{Practically Stabilizing Virtual Synchrony}
%\\~~\\\large{(Submitted as a Regular Paper to PODC'15)}}
%
%\FFF\FFF\FFF\FFF\AffFont{by}
%
%\FFF\FFF\FFF\FFF\textbf{Shlomi Dolev, Thomas Petig and Elad M.\ Schiller}
%
%\FFF\FFF\FFF\FFF\AuthorFont{Department of Computer Science,
%Ben-Gurion University, Beer Sheva, Israel}
%
%\FFF\FFF\FFF\FFF\AuthorFont{Technical Report $\#$15-04}
%
%\FFF\FFF\FFF\FFF\AuthorFont{Feb. 2015}
%%\end{spacing}
%\end{center}
%\thispagestyle{empty}
%\newpage

%\begin{frontmatter}

%\renewcommand{\thefootnote}{\fnsymbol{footnote}}

%
%\title{Practically Stabilizing Virtual Synchrony
%for Automata Replication using Wait-free Non-exhausting Multi-writer Counters}

\title{Practically-Self-Stabilizing Virtual Synchrony\footnote{A preliminary version of this work has appeared in the proceedings of the 17th International Symposium on Stabilization, Safety, and Security of Distributed Systems (\emph{SSS'15}).}}

\author{Shlomi Dolev~\footnote{Department of Computer Science, Ben-Gurion University of the Negev, Beer-Sheva, Israel. Email {\tt dolev@cs.bgu.ac.il} Partially supported by Rita Altura Trust Chair in Computer Sciences, Lynne and William Frankel Center for Computer Sciences and Israel
Science Foundation (grant number 428/11).} \and Chryssis Georgiou~\footnote{ Department of Computer Science, University of Cyprus, Nicosia, Cyprus. Email {\tt {\char '173}chryssis, imarco01{\char '175} @cs.ucy.ac.cy}. The work of Ioannis Macoullis is supported by a Doctoral Scholarship program of the University of Cyprus.} \and~Ioannis Marcoullis$^\ddag$ \and Elad M.\ Schiller~\footnote{Department of Engineering and Computer Science, Chalmers University of Technology, Gothenburg, SE-412 96, Sweden, Email {\tt elad@chalmers.se}.}}

\maketitle

%
%\title{\bl{Practically-}Self-Stabilizing Virtual Synchrony\tnoteref{t1}}
%\tnotetext[t1]{An extended abstract and a technical report of this paper appeared in~\cite{DBLP:conf/sss/DolevGMS15} and~\cite{DBLP:journals/corr/DolevGMS15}, respectively.}
%
%
%\author[dolevAdd]{Shlomi~Dolev}
%\ead{dolev@cs.bgu.ac.il}
%
%\author[georgiouMarcoullisAdd]{Chryssis~Georgiou}
%\ead{chryssis@cs.ucy.ac.cy}
%
%\author[georgiouMarcoullisAdd]{Ioannis~Marcoullis\corref{cor1}}
%\ead{imarco01@cs.ucy.ac.cy}
%
%\author[schillerAdd]{Elad~M.~Schiller}
%\ead{elad@chalmers.se}

%\cortext[cor1]{Corresponding author.}
%
%\address[dolevAdd]{Department of Computer Science, Ben-Gurion University of the Negev, Beer-Sheva, Israel.} 
%
%\address[georgiouMarcoullisAdd]{Department of Computer Science, University of Cyprus, Nicosia, Cyprus.} 
%
%
%\address[schillerAdd]{Department of Computer Science and Engineering, Chalmers University of Technology, G\"oteborg, Sweden.} 

%\date{}

%\maketitle

%\thispagestyle{empty}
\begin{abstract}

The virtual synchrony abstraction was proven to be extremely useful for asynchronous, large-scale, message-passing distributed systems. Self-stabilizing systems can automatically regain consistency after the occurrence of transient faults.

We present the first \bl{practically-}self-stabilizing virtual synchrony algorithm that uses a new counter algorithm that establishes an efficient practically unbounded counter, which in turn can be directly used for emulating a self-stabilizing Multiple-Writer Multiple-Reader (MWMR). Other self-stabilizing services include membership, multicast, and replicated state machine (RSM) emulation. As we base the latter on virtual synchrony, rather than consensus, the system can progress in more extreme asynchronous executions than consensus-based RSM emulations.
\end{abstract}

%\end{frontmatter}

%\renewcommand{\thefootnote}{\arabic{footnote}}
%
%\newpage
%\pagestyle{plain}
%\setcounter{page}{1}

\section{Introduction} %\vspace{-.8em}

\bl{\emph{Virtual Synchrony} (VS) is an important property provided by several Group Communication Systems (GCSs) that has proved to be valuable in the scope of fault-tolerant distributed systems where communicating processors are organized in process groups with changing membership~\cite{DBLP:conf/replication/Birman10}. During the computation, groups change allowing an outside observer to track the history (and order) of the groups, as well as the messages exchanged within each group.
The VS property guarantees that any two processors that both participate in two consecutive such groups, should deliver the same messages in their respective group.}
Systems that support the VS abstraction are designed to operate in the presence of fail-stop failures of a minority of the participants. Such a design fits large computer clusters, data-centers and cloud computing, where at any given time some of the processing units are non-operational. 
Systems that cannot tolerate such failures degrade their functionality and availability to the degree of unuseful systems.   

Group communication systems that realize the VS abstraction provide services, such as {\em group membership} and {\em reliable group multicast}. 
The group membership service is responsible for providing the current {\em group view} of the recently live and connected group members, i.e., a processor set and a unique {\em view identifier}, which is a sequence number of the view installation. 
The reliable group multicast allows the service clients to exchange messages with the group members as if it was a single communication endpoint with a single network address and to which messages are delivered in an atomic fashion, thus any message is either delivered to all recently live and connected group members prior to the next message, or is not delivered to any member. 
The challenges related to VS consist of the need to maintain atomic message delivery in the presence of asynchrony and crash failures. 
VS facilitates the implementation of a replicated state machine~\cite{DBLP:conf/replication/Birman10} that is more efficient than classical consensus-based implementations that start every multicast round with an agreement on the set of recently live and connected processors.
It is also usually easier to implement~\cite{DBLP:conf/replication/Birman10}. 

\Paragraph{Transient faults} Transient violations of design assumptions can lead a system to an arbitrary state.
For example, the assumption that error detection ensures the arrival of correct messages and the discarding of corrupted messages, might be violated since error detection is a probabilistic mechanism that may not detect a corrupt message. 
As a result, the message can be regarded as legitimate, driving the system to an arbitrary state after which, availability and functionality may be damaged forever, requiring human intervention. 
In the presence of transient faults, large multicomputer systems providing VS-based services can prove hard to manage and control. 
One key problem, not restricted to virtually synchronous systems, is catering for counters (such as view identifiers) reaching an arbitrary value. 
How can we deal with the fact that transient faults may force counters to wrap around to the zero value and violate important system assumptions and correctness invariants, such as the ordering of events? 
A self-stabilizing algorithm~\cite{D2K} can automatically recover from such unexpected failures, possibly as part of after-disaster recovery or even after benign temporal violations of the assumptions made in the design of the system. 
To the best of our knowledge, no {\em stabilizing virtual synchrony} solution exists.
We tackle this issue in our work.

\Paragraph{Practically-self-stabilization} \sloppy A relatively new self-stabilization paradigm is \emph{practically-self-stabilization}~\cite{Alon2014,DolevKS2010ConsMeetsSS,DBLP:conf/netys/BlanchardDBD14}.
Consider an asynchronous system with bounded memory and data link capacity in which corrupt pieces of data (stale information) exist due to a transient fault. (Recall that transient faults can result in the appearance of corrupted information, which the system tends to spread and thus reach an arbitrary state.)
These corrupted data may appear \emph{unexpectedly} at any processor as they lie in communication links, or may (indefinitely) remain ``hidden'' in some processor's local memory until they are added to the communication links as a response to some other processor's input.
Whilst these pieces of corrupted data are bounded in number due to the boundedness of the links and local memory, they can eventually force the system to lose its safety guarantees. 
%
%If the system is self-stabilizing, then such corrupted information may repeatedly drive the system to an undesired state of non-functionality. 
%
Such corrupt information may repeatedly drive the system to an undesired state of non-functionality. This is true for all systems and self-stabilizing systems are required to eradicate the presence of all corrupted information.
In fact, whenever they appear, the self-stabilizing system is required to regain consistency and in some sense stabilize. 
%each time requiring re-stabilization. %period if the previous one was completed. %(note that ``period'' here does not imply any periodicity or bound on stabilization time).
One can consider this as an adversary with a limited number of chances to interrupt the system, but only itself knows \emph{when} it will do this.

In this perspective, self-stabilization, as it was proposed by Dijkstra~\cite{Dijkstra74}, is not the best design criteria for asynchronous systems for which we cannot specifically define \emph{when} stabilization is expected to finish (in some metric like asynchronous cycles, for example).
%In this perspective, traditional self-stabilization is not the best tool in that it cannot specifically define \emph{when} stabilization is expected to take place (in some metric like asynchronous rounds for example).
%The notion of \emph{pseudo-stabilization}~\cite{D2K}, on the other hand, captures this case by requiring an infinite execution in order to guarantee that once these bounded corrupt information are revealed, then there is an infinite suffix of the execution in which safety is guaranteed. 
%
The newer criterion of practically-stabilizing systems is closely related to pseudo-self-stabilizing systems~\cite{DBLP:journals/dc/BurnsGM93}, as we explain next.
%
%New text proposed by Elad -- 13 Jul, 2017
Burns, Gouda and Miller~\cite{DBLP:journals/dc/BurnsGM93} deal with the above challenge by proposing the design criteria of {\em pseudo-self-stabilization}, which merely bounds the number of possible safety violations. Namely, their approach is to abandon Dijkstra's seminal proposal~\cite{Dijkstra74} to bound the period in which such violations occur (using some metric like asynchronous cycles). We consider a variation on the design criteria for pseudo-self-stabilization systems that can address additional challenges that appear when implementing a decentralized shared counter that uses a constant number of bits. 

Self-stabilizing systems can face an additional challenge due to the fact that a single transient fault can cause the counter to store its maximum possible value and still (it is often the case that) the system needs to be able to increment the counter for an unbounded number of times. The challenge becomes greater when there is no elegant way to show that the system can always maintain an order among the different values of the counter by, say, wrapping to zero in such integer overflow events. Arora, Kulkarni and Demirbas~\cite{DBLP:journals/jpdc/AroraKD06} overcome the challenge of integer overflow by using non-blocking resets in the absence of faults described~\cite{DBLP:journals/jpdc/AroraKD06}. In case faults occur, the system recovery requires a blocking operation, which performs a distributed global reset. This work considers a design criteria for message passing systems that perform in a wait-free manner also when recovering from transient faults. 

Note that, from the theoretical point of view, systems that take an extraordinary large number of steps (that accedes the counter maximum value, or even an infinite number of steps) are bound to violate any ordering constraints. This is because of the asynchronous nature of the studied system, which could arbitrarily delay a node from taking steps or defer the arrival of a message until such violations occur after, say, a counter wraps around to zero. Having practical systems in mind, we consider systems for which the number of sequential steps that they can take throughout their lifetime is not greater than an integer that can be represented using a constant number of bits. For example, Dolev, Kat and Schiller~\cite{DolevKS2010ConsMeetsSS} assume that counting from zero to $2^{64}-1$ using sequential steps is not possible in any practical system and thus consider only a \textit{practically infinite period}, of $2^{64}$ sequential steps, that the system takes when demonstrating that safety is not violated. 
The design criteria of practically-self-stabilizing systems~\cite{SalemSchiller2017,Alon2014,DBLP:conf/netys/BlanchardDBD14} requires that there is a bounded number of possible safety violations during any practically infinite period of the system execution. For such (message passing) systems, we provide a decentralized shared counter algorithm that performs in a wait-free manner also when recovering from transient faults.

%end remove

%CONTRIBUTIONS
\Paragraph{Contributions} 
We present the first %, to our knowledge, 
\bl{practically-self-stabilizing (or \emph{practically-stabilizing})} virtual synchrony solution. 
Specifically: \vspace{-.4em}
\begin{itemize}\setlength\itemsep{-0.3em}
\item We provide a \bl{practically-self-stabilizing} counter algorithm using %only bounded amount of memory 
bounded memory 
and communication bandwidth, where many writers can increment the counter for an unbounded number of times in the presence of %(a minority number) 
processor crashes and unbounded communication delays. 
%\item 
Our counter algorithm is modular with a simple interface for increasing and reading the counter, as well as providing the identifier of the %last 
processor that has incremented it. %{\bf CG: Can we actually provide the id of the last processor that has incremented it, especially in the presence of concurrency?}
\item At the heart of our counter algorithm is the underlying labeling algorithm that extends the label scheme of Alon et al.~\cite{Alon2014} to support {\em multiple} writers, whilst the algorithm specifies how the processors exchange their label information in the asynchronous system  and how they maintain proper label bookkeeping so as to ``discover'' the greatest label and discard all obsolete ones.
\item An immediate application of our counter algorithm is a \bl{practically-}self-stabilizing MWMR register emulation. %, that is internally significantly more involved than the SWMR algorithm in \cite{Alon2014}. 
%{\bf CG: I think here we need to say that the heart of the counter algorithm is the labeling algorithm that extends the work of alon to support multiple writes in a simpler and communication efficient way.} 
%{\bf CG: Then we need to mention that a first application of the counter algorithm is the MWMR register emulation.}
%\item Our self-stabilizing [[@@ this is not clear @@ counter, using]] the self-stabilizing reliable multicast and membership services yield our self-stabilizing VS solution, which leads to a self-stabilizing VS-based State Machine Replication (SMR) implementation.

\item 
The \bl{practically-}self-stabilizing counter algorithm, together with %the proposed 
implementations of a \bl{practically-}self-stabilizing reliable multicast service and membership service that we propose, are composed to yield a \bl{practically-}self-stabilizing coordinator-based Virtual Synchrony solution. 
\item Our Virtual Synchrony solution yields a \bl{practically-self-stabilizing} State Machine Replication (SMR) implementation. As this implementation is
based on virtual synchrony rather than consensus, the system can progress in more extreme asynchronous executions than consensus-based SMR implementations. 
%S-based State Machine Replication (SMR) implementation. %[[@@ Io. (I can't seem to understand what is unclear, but what about this:) @@ ]]
%The task of composing these services is non-trivial, but it leads to a self-stabilizing state machine replication implementation (SMR). 
%{\bf CG: What techniques we need to highlight here? Which one has not been used before? The FD has been used before, as well as the SS links. So, is it the multicast service?}
\end{itemize}

% RELATED WORK
\Paragraph{Related Work} 
Leslie Lamport was the first to introduce SMR, presenting it  as an example in~\cite{DBLP:dblp_journals/cacm/Lamport78}.
Schneider~\cite{DBLP:journals/csur/Schneider90} gave a more generalized approach to the design and implementation of SMR protocols. 
Group communication services can implement SMR by providing reliable multicast that guarantees VS~\cite{DBLP:journals/jsa/Bartoli04}.
Birman et al. were the first to present VS %in~\cite{DBLP:dblp_journals/tse/BirmanJRA85}, 
and a series of improvements in the efficiency of ordering protocols%in subsequent publications
%improvements in ordering protocols and efficiency in subsequent publications
%~\cite{Birman91Lightweight,birmanbook}.
~\cite{birmanbook}.
Birman gives a concise account of the evolution of the VS model for SMR %state replication 
in~\cite{DBLP:conf/replication/Birman10}.
 
Research during the last recent decades resulted in an extensive literature on ways to implement VS 
and SMR, as well as industrial construction of such systems. A recent research line on %(practically) self-
stabilizing versions of replicated state machines~\cite{DolevKS2010ConsMeetsSS,DBLP:dblp_conf/sss/DolevLLN05,Alon2014,DBLP:conf/netys/BlanchardDBD14} obtains self-stabilizing replicated state machines in shared memory as well as in synchronous and asynchronous message passing systems.

\sloppy{The bounded labeling scheme and the use of \bl{practically} unbounded sequence numbers proposed in~\cite{Alon2014}, allow the creation of \bl{practically-}stabilizing bounded-size solutions to the never-exhausted counter problem in the restricted case of a single writer.}
In~\cite{DBLP:conf/netys/BlanchardDBD14} a \bl{practically-}self-stabilizing version of Paxos was developed, which led to a \bl{practically}-self-stabilizing consensus-based SMR implementation. 
%To this end, a labeling scheme extending the one of \cite{Alon2014} to allow multiple writers.
To this end, they extended the labeling scheme of \cite{Alon2014} to allow for multiple counter writers, since unbounded counters are required for ballot numbers.
Extracting this scheme for other uses %to use as a building block elsewhere 
does not seem intuitive. 
%The labeling scheme in~\cite{DBLP:conf/netys/BlanchardDBD14} extends these capabilities to the multi-writer case by exchanging vector of labels. 
%{\bf CG: I think here we need to say that the labeling scheme was devised to be used by Paxos to solve the consensus problem and it is not obvious how one can extract and use it as a building block. IS THIS TRUE though?}
We present a simpler and significantly more communication efficient %self-stabilizing (bounded-size never-exhausted) 
practically infinite counter that also supports many writers, 
%where a single value rather than a vector of such values is communicated to achieve the same goal.
where only a pair of labels rather than a vector of labels needs to be communicated. %to achieve the same goal.
Our solution is {\em highly modular} and can  be easily used in any similar setting requiring such counters.
\bl{We also note that with~\cite{Alon2014}'s single writer atomic register emulation, a quorum read of the value could return without a value if the writer did not before perform a write to establish a maximal tag. 
An emulation based on our multiple-writer version guarantees that reads may always terminate with a value, since our labeling algorithm continuously maintains a maximal tag.}

%Next, %in Section~\ref{sec:nutshell}, we overview our construction, describing the core techniques and the way they establish the desired properties. In 
\bl{In what follows, Section~\ref{s:sys} presents the system settings and the necessary definitions. 
Section~\ref{sec:Labels} details the practically-self-stabilizing Labeling Scheme and Increment Counter algorithms. 
In Section~\ref{sec:VS} we present the practically-self-stabilizing Virtual Synchrony algorithm and the resulting replicate state machine emulation. 
We conclude with Section~\ref{s:concl}.
}

% % % % % % % % % % % % % % % % % % % % % % % % % % % % % % % % % % % % % % % % %
% %  T H E  N U T S H E L L  W A S  H E R E ! ! ! (NOW AFTER END OF DOCUMENT) % %
% % % % % % % % % % % % % % % % % % % % % % % % % % % % % % % % % % % % % % % % %

%\section{System Settings}  %\vspace{-.6em}
\section{\bl{System Settings and Definitions}}
\label{s:sys}

We consider an asynchronous message-passing system. %as the one used in \cite{Alon2014}. 
The system includes a set $P$ of $n$ communicating processors; we refer to the processor with identifier $i$, as $p_i$. 
%We assume that at most $n/2-1$ processors may become inactive. 
\bl{At most $n/2-1$ processors may fail by crashing and these may sometimes be referred to as \emph{inactive} in contrast to \emph{active} processors that are not crashed.}
We assume that the system runs on top of a stabilizing data-link layer that provides reliable FIFO communication over unreliable bounded capacity channels as the ones of~\cite{DBLP:journals/ipl/DolevDPT11,DBLP:conf/sss/DolevHSS12}. 
The network topology is of a fully connected graph where every two processors exchange (low-level messages called) {\em packets} to 
enable a reliable delivery of (high level) messages. When no confusion is possible we use the term messages for packets. 

\Paragraph{\bl{Communication and data link implementation}} 
The communication links have bounded capacity, so that the number of messages in every given instance is bounded by a constant  \bl{$cap$, which is known to the processors}. 
When processor $p_i$ sends a packet, $\pi$, to processor $p_j$, the operation $send$ inserts a copy of $\pi$ to the 
FIFO queue that represents the communication channel from $p_i$ to $p_j$, while respecting an upper bound on the number of packets in the channel, possibly omitting the new packet or one of the already sent packets. 
When $p_j$ receives $\pi$ from $p_j$, $\pi$ is dequeued from the queue representing the channel. We assume that packets can be spontaneously  omitted (lost) from the channel, however, a packet that is sent infinitely often is received infinitely often. 

\trnsfr{
One version of a self-stabilizing FIFO data link implementation that we can use, is based on the fact that communication links have bounded capacity. Packets are retransmitted until more than the total capacity acknowledgments arrive; while acknowledgments are sent only when a packet arrives (not spontaneously)~\cite{DBLP:journals/ipl/DolevDPT11,DBLP:conf/sss/DolevHSS12}.
Over this data-link, the two connected processors can constantly exchange a ``token". Specifically, the sender (possibly the processor with the
highest identifier among the two) constantly sends  packet $\pi_1$ until it receives enough acknowledgments (more than the capacity). 
Then, it constantly sends packet $\pi_2$, and so on and so forth. This assures that the receiver has received packet $\pi_1$ before the sender starts
sending packet $\pi_2$. 
This can be viewed as a token exchange. We use the abstraction of the token carrying messages back and forth between any two communication entities. %We use this token exchange technique when implementing a reliable multicast procedure, as well as a the basis for a {\em heartbeat} for detecting  whether a processor is active or not; when a processor in no longer active, the token will not be returned back to the other processor.}
a {\em heartbeat} to (imperfectly) detect whether a processor is active or not; when a processor in no longer active, the token will not be returned back to the other processor.
}

%The code of self-stabilizing algorithms usually consists of a do forever loop that contains communication operations with the neighbors and validation that the system is in a consistent state as part of the transition decision. An {\em iteration} is said to be complete if it starts in the loop's first line and ends at the last (regardless of whether it enters branches).

\Paragraph{\bl{Definitions and complexity measures}}Every processor, $p_i$, executes a program that is a sequence of {\em (atomic) steps}, where a \emph{step} starts with local computations and ends with a single communication operation, which is either $send$ or $receive$ of a packet. For ease of description, we assume the interleaving model, where steps are executed atomically, a single step at any given time. An input event can be either the receipt of a packet or a periodic timer triggering $p_i$ to (re)send. Note that the system is asynchronous and the rate of the timer is totally unknown. 

The {\em state}, $s_i$, of a node $p_i$ consists of the value of all the variables of the node including the set of all incoming communication channels. The execution of an algorithm step can change the node's state. The term {\em (system) configuration} is used for a tuple of the form $(s_1, s_2, \cdots, s_n)$, where each $s_i$ is the state of node $p_i$ (including messages in transit to $p_i$). We define an {\em execution (or run)} $R={c_0,a_0,c_1,a_1,\ldots}$ as an alternating sequence of system configurations $c_x$ and steps $a_x$, such that each configuration $c_{x+1}$, except the initial configuration $c_0$, is obtained from the preceding configuration $c_x$ by the execution of the step $a_x$.

\pend{
An execution $R_p$ is \emph{practically infinite execution} if it contains a chain of steps ordered according to Lamport's happened-before relation~\cite{DBLP:dblp_journals/cacm/Lamport78} that are longer than $2^\tau$ ($\tau$ being, for example, $64$), namely they are practically infinite for any given system~\cite{DolevKS2010ConsMeetsSS}. 
%Throughout the paper we consider infinite executions where a processor that fails by crashing stops taking steps, and any processor that does not crash eventually takes an infinite number of steps.
%\bl{
Similar to an infinite execution, a processor that fails by crashing stops taking steps, and any processor that does not crash eventually takes a practically infinite number of steps.
The code of self-stabilizing algorithms reflects the requirement for non-termination in that it usually consists of a $\tt do - forever$ loop that contains communication operations with the neighbors and validation that the system is in a consistent state as part of the transition decision. 
An {\em iteration} of an algorithm formed as a  $\tt do - forever$ loop is a complete run of the algorithm starting in the loop's first line and ending at the last line, regardless of whether it enters branches. (Note that an iteration may contain many steps).%
%}
%Since an iteration may complete with several communications with a number of processors it contains of many execution steps.
%We also define an \emph{asynchronous cycle} (in the sequel $cycle$) of $R$, to be the shortest prefix $R'$ of $R=R'R''$ in which every processor (that is not crashed) performs at least one complete iteration.}
%
\remove{
\trnsfr{A practically infinite execution is an execution with many steps, \textcolor{red}{[[I have deleted (and iterations) from here]]} %(and iterations), 
where many is defined to be proportional to the time it takes to execute a step and the life-span time of a system. \textcolor{red}{[[Can we sharpen up this sentence? It is provoking...]]}}}
}

We define the system's task by a set of executions called {\em legal executions} ($LE$) in which the task's requirements hold,
we use the term {\em safe configuration} for any configuration in any execution in $LE$.
\bl{As defined by Dijkstra in~\cite{Dijkstra74}, an algorithm is {\em self-stabilizing} with relation to the task $LE$ when every (unbounded) execution of the algorithm reaches a safe configuration with relation to the algorithm and the task.} \pend{We define the system's abstract task ${\cal T}$ by a set of variables (of the processor states) and constraints, which we call the system requirements, in a way that implies the desired system behavior~\cite{D2K}.
	 Note that an execution $R$ can satisfy the abstract task and still not belong to $LE$, because $R$ considers only a subset of variables, whereas the configurations of executions that are in $LE$ consider every variable in the processor states and message fields.
	 An algorithm is {\em practically-self-stabilizing} (or just \emph{practically-stabilizing})  with relation to the task ${\cal T}$ if in any practically infinite execution has a bounded number of deviations ${\cal T}$~\cite{SalemSchiller2017}.
%a safe configuration is reached.

This defines a measure for complexity.
The asynchrony of the system makes it hard, if not impossible to infer anything on stabilization \emph{time}, since we cannot predict \emph{when} an element from the corrupt state of the system will reach a processor, (cf. self-stabilizing solutions that give time complexity in asynchronous rounds). 
%The notion of practically-stabilizing solutions encompasses exactly this. 
Based on the above definition of practically-stabilizing algorithms, a bounded number of corrupt elements that might force the system to \emph{deviate} from its task even if these may or may not (due to asynchrony) appear. 
Whenever a deviation happens, a number of algorithmic operations are required to satisfy ${\cal T}$ once again.
As a complexity measure, we bound the total of these operations throughout an execution. 
These operations differ by algorithm, i.e., it is label creations in the labeling scheme, counter increments for the counter increment algorithm and view creations in the virtual synchrony algorithm. %, and they are dependent on the bound on the corrupt elements.
}

\begin{figure}[t]
\begin{center}
\includegraphics[scale=0.55]{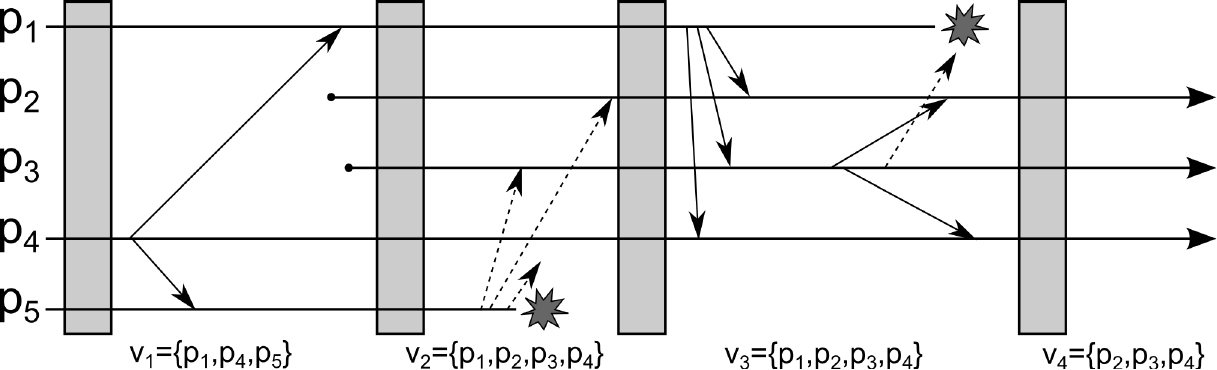}
\end{center}
\caption{\small An execution satisfying the VS property. 
The grey boxes indicate a new view installation, and the example shows four views.
View $v_1$ initially with membership $\{p_1, p_4, p_5\}$.
The reliable multicast reaches all members of the group.
Two new processors $p_2$ and $p_3$ join the group, forming view $v_2$.
In this view, $p_5$ crashes before completing its multicast which is ignored (dashed lines). 
The new view $v_3$ is formed to exclude $p_5$, and in it, $p_1$ manages a successful multicast before crashing.
The multicast of $p_3$ is reliable and guaranteed to be delivered to all non-crashed within the view, that is excluding $p_1$ which might or might not have received it (dotted line).
A new view is then formed to encapture the failure of $p_1$.}
\label{fig:VSexec}
\end{figure}
\bl{The {\em virtual synchrony task} uses the notion of a \emph{view}, a group of processors that perform multicast within the group and is uniquely identified, to ensure that any two processors that belong to two views that are consecutive according to their identifier, {\em deliver}  identical message sets in these views.} The legal execution of virtual synchrony is defined in terms of the input and output sequences of the system with the environment. When a majority of processors are continuously active every external input (and only the external inputs) should be atomically accepted and processed by the majority of the active processors. \bl{The system works in the primary component, i.e., it does not deal with partitions and requires that a view contains a majority of the system's processors, i.e., its membership size is always greater than $n/2$. Therefore,}  %Note that 
there is no delivery and processing guarantee in executions in which there is no majority, still in these executions any delivery and processing is due to a received environment input.  Figure~\ref{fig:VSexec} is an example of a virtually synchronous execution.\\
%
%Practically-stabilizing VS and self-stabilizing VS are identical when VS is defined by the behaviour of classical VS algorithms that use (bounded) counters to identify views. 
%These algorithms preserve the VS requirements as long as the counters do not reach their upper bound.
%In our setting, if a counter reaches the upper bound due to a transient fault our self-stabilizing/practically-stabilizing solution it provides a new counter initiated at 0 (technicalities are deferred to the related section). %introduces a new epoch with new sequence numbers. %therefore abiding with the classical VS definition. 
%%It, thus, converges to act exactly as the non-stabilizing VS (for the same number of steps) as an initialized non-stabilizing VS algorithm.\\
%It, thus, converges to act  exactly as the initialized non-stabilizing VS algorithm for the same number of steps.\\
\textbf{Notation.} Throughout the paper we use the following notation. Let $y$ and $y'$ be two objects that both include the field $x$. We denote  $(y =_{x}y')$ $\equiv$ $(y.x$ $=$ $y'.x)$.

\section{\bl{Practically-}Self-Stabilizing Labeling Scheme and %Increment 
Counter Algorithm} %\vspace{-.7em}
\label{sec:Labels}

\bl{\pend{Many system like the ones performing replication (e.g. GCSs requiring group identifiers, of Paxos implementations requiring ballot numbers) assume access to an infinite (unbounded) counter.
We proceed to give a practically-stabilizing, practically infinite counter based on a bounded labeling scheme.
Note that by a \emph{practically infinite (\emph{or} unbounded) counter} we imply that a $\tau$-bit counter (e.g., $64$-bit) is not truly infinite (since this is anyway not implementable on hardware), but it is large enough to provide counters for the lifetime of most conceivable systems when started at $0$.
We refer the reader to the example provided by Blanchard et al.~\cite{Blanchard2013SSPaxos}, where a $64$-bit counter initialized at $0$ and incremented per nanosecond is calculated to last for around 500 years, essentially an infinity for most of today's running systems.
%In the sequel, we will use $2^\tau$ to imply a number large enough to be considered practically infinite in the context described here.
}}

\bl{\pend{
The task of a practically-self-stabilizing labeling scheme is for every processor that takes an infinite %, yet bounded [[@@EMS@@{What is infinite yet bounded?}]] 
number of steps to reach to a label that is maximal for all active processors in the system.
The task of maintaining a practically infinite counter, is for every processor that takes an infinite yet bounded number of steps, to eventually be able to monotonically increment the counter from 0 to $2^\tau$. 
%\textcolor{Emerald}{[[@@EMS@@{There is a need to explain how to combine. That is, to say that the label is an epoch and that whenever we reach MAXINT in the counter, we select a new label to serve as an epoch, such that, that epoch is greater than all labels that any processor is using and have been seen by the creator of the label.}]]
%[[@@IO@@Does the following sentence help in combining?]]}
The latter task depends on the former to provide the maximal label in the system to be used as a sequence number epoch, so that within the same epoch, the integer sequence number is incremented as a practically infinite counter.
It is implicit that the tasks are performed in the presence of corrupt information that might exist due to transient faults.}}

\bl{\pend{
Our solutions are \emph{practically infinite}, in the following way. 
A bounded amount of stale information from the corrupt initial state, may unpredictably %force processors to adopt either counter values that are near exhaustion, or labels that [[@@ This is not clear. @@]] corrupt the current counter. 
corrupt the counter.
In such cases, processors are forced to change their labels and restart their counters.
A processor cannot predict whether a corrupt piece of information exists, or \emph{when} will it make its appearance as this is essentially the work of asynchrony.
Our solutions guarantee that only a bounded number of labels will need to change, or that only a bounded number of counter increments will need to take place before we reach to one that is eligible to last its full $2^\tau$ length, less the fact that this maximal value is practically unattainable.
}
}
%\Paragraph{\bl{Task description}}

We first present and prove the correctness of a practically-stabilizing labeling algorithm,
and then explain how this can be extended to implement practically stabilizing, practically  
unbounded counters in Section~\ref{subsec:CounterI}. 

\subsection{Labeling Algorithm for Concurrent Label Creations} %\vspace{-1em}
\label{subsec:Labels} %\vspace{-.6em}
%
%\Subsection{Extending the labeling scheme}
%
%\subsubsection{Bounded Labeling Scheme}
\subsubsection{Preliminaries}
%We extend the labeling scheme of~\cite{Alon2014} to support wait-free multi-writer systems. We do so, by extending the label with a {\em label creator's} identifier, so as to break symmetry and decide about the most recent epoch even when two or more writers concurrently attempt to create a new label.
\Paragraph{\bl{Bounded labeling scheme}}
\trnsfr{
% % Completely removed previous intro and substituted with nutshel's
The  bounded labeling scheme of Alon et al.~\cite{Alon2014} implements an SWMR register emulation in a message-passing system. The $labels$ (also called $epochs$) allow the system to stabilize, since once a label is established, the integer counter related to this label is considered to be practically infinite, as a $64$-bit integer is practically infinite and sufficient for the lifespan of any reasonable system. We extend the labeling scheme of~\cite{Alon2014} to support multiple writers, by including the epoch creator (writer) identity to break symmetry, and decide which epoch is the most recent one, even when two or more creators concurrently create a new label.
}

\bl{Formally defined,} we consider the set of integers $D$ $=$ $[1$, $k^{2}+1]$ \bl{such that $k \in \mathbb{N}$ a known constant to the processors, which we determine in Corollary~\ref{thm:kSize}}.
A {\em label} (or \emph{epoch}) is a triple $\langle lCreator, sting, Antistings\rangle$, where $lCreator$ is the identity of the processor that established (created) the label, $Antistings \subset D$ with $|Antistings|=k$, and $sting \in D$. Given two labels $\ell_{i}, \ell_{j}$,
%we define the relation  $\ell_i \prec_{id} \ell_j \equiv \ell_i.lCreator < \ell_j.lCreator$; we use %$=_{lCreator}$ to say that the labels have the same creator. Then 
we define the relation $\ell_{i}$ $\prec_{lb}$ $\ell_{j}$ $\equiv$ $(\ell_i.lCreator$ $<$ $\ell_j.lCreator)$ $\lor$ $(\ell_{i}.lCreator$ $=$ $\ell_j.lCreator$ $\land$ $((\ell_i.sting$ $\in$ $\ell_j.Antistings)$ $\land$ $(\ell_j.sting$ $\not \in \ell_i.Antistings)))$;
we use $=_{lb}$ to say that the labels are identical. Note that 
%as pointed out in~\cite{Alon2014}, 
the relation $\prec_{lb}$
does not define a total order. For example, when $\ell_i =_{lCreator} \ell_j$ and $(\ell_i.sting \not \in \ell_j.Antistings)$ and 
$(\ell_j.sting \not \in \ell_i.Antisting)$ these labels are {\em incomparable}. 

\bl{As an example, consider the situation with $k=3$, and $D=\{1,2,\ldots,10\}$. 
Assume the existence of three labels $\ell_1 = \langle i, 2, \langle 3, 5, 9\rangle \rangle $, $\ell_2 = \langle i, 1, \langle 2, 9, 10\rangle \rangle $, and $\ell_3 = \langle i+1, 1, \langle 3, 5, 9\rangle \rangle$. 
In this case, $\ell_1 \prec_{lb} \ell_3$ and $\ell_2 \prec_{lb} \ell_3$, since the creator of $\ell_3$ has a greater identity than the creator of $\ell_1$ and $\ell_2$.
We can also see that $\ell_1 \prec_{lb} \ell_2$, since the sting of $\ell_1$, namely $2$, belongs to the antistings set of $\ell_2$ (which is $\langle 2, 9, 10\rangle$) while the opposite is not true for the sting of $\ell_2$.
This makes $\ell_2$ ``immune'' to the sting of $\ell_1$. 
}

As in~\cite{Alon2014}, we demonstrate that one can still use this labeling scheme as long as it is ensured that eventually a label greater than all other labels in the system is introduced. 
We say that a label $\ell$ {\bf\em cancels} another label $\ell'$, either if they are incomparable or they have the same $lCreator$ but $\ell$ is greater than 
$\ell'$ (with respect to $sting$ and $Antistings$).   
A label with creator $p_i$ is said to belong to $p_i$'s domain.

\begin{algorithm}[t!]
   \caption{The $nextLabel()$ function; code for $p_i$}
\label{alg:NB}

\begin{small}

For any non-empty set $X\subseteq D$, function $pick(d, X)$ returns $d$ arbitrary elements of $X$\;
\Input{$S = \langle \ell_1, \ell_2 \ldots, \ell_k\rangle$ set of $k$ labels.}
\Output{$\langle i, newSting, newAntistings\rangle$}
%~\\
{
{\bf let } $newAntistings = \{\ell_j.sting : \ell_j\in S \}$\;

$newAntistings \gets newAntistings$ $\cup$ $pick(k-|newAntistings|,$  $D \setminus newAntistings)$\;

\Return{$\langle i, pick(1, D \setminus (newAntistings \cup \{ \cup_{\ell_j\in S} \ell_j.Antistings\})), newAntistings\rangle$}\;
}
\end{small}
\end{algorithm}

\Paragraph{Creating a largest label} Function $nextLabel()$, Algorithm~\ref{alg:NB}, gets a set of at most $k$ labels as input and returns a new label that is greater than all of the labels of the input\bl{, given that all the input labels have the same creator i.e., the same $lCreator$. This last condition is imposed by the labeling algorithm that calls $nextLabel()$, as we will see further down with a set of labels from the same processor.} It has the same functionality as the function called $Next_b()$ in~\cite{Alon2014}, \bl{but it additionally appends the label creator to the output.} %considers the label creator. 
The function essentially composes a new $Antistings$ set from the stings of all the labels that it receives as input, and chooses a $sting$ that is in none of the $Antistings$ of the input labels.
In this way it ensures that the new label is greater than any of the input.
Note that the function takes $k$ $Antistings$ of $k$ labels that are not necessarily distinct, implying at most $k^2$ distinct integers and thus the choice of $|D|$ $=$ $k^2+1$ allows to always obtain a greater integer as the $sting$. 
\bl{For the needs of our labeling scheme, $k=4(n^3cap+2n^2-2n)+1$ (Corollary~\ref{thm:kSize}).}
%\vspace{-.7em}

%allows always to obtain a greatest integer as the $sting$. 

\Paragraph{\bl{Scheme idea and challenges}}
\trnsfr{When all processors are active, the scheme can be viewed as a simple extension of the one of~\cite{Alon2014}. 
Informally speaking, the scheme ensures that each processor $p_i$ eventually ``cleans up'' the system from obsolete labels of which $p_i$ appears to be the creator (for example, such labels could be present in the system's initial arbitrary state). 
Specifically, $p_i$ maintains a bounded FIFO history of such labels that it has recently learned, while communicating with the other processors, and creates a label greater than all that are in its history; call this $p_i$'s {\em local maximal label}. 
In addition, each processor seeks to learn the {\em globally maximal label}, that is, the label in the system that is the greatest among the local maximal ones.}

\bl{
We note here that compared to Alon et al.~\cite{Alon2014}, which only had a single writer upon the failure of whom there would be no progress thus stabilization would not be the main concern, we have \emph{multiple} label creators. 
If these creators were not allowed to crash then the extension of the scheme would be a simple exercise.
Nevertheless, when some processors can crash the problem becomes incrementally more difficult as we now explain.
The problem lies in cleaning the system of these crashed processors' labels since they will not ``clean up'' their local labels. 
%Unfortunately, when some processors are not active, finding a global maximal becomes challenging, since these processors will not ``clean up'' their local labels. 
%So, roughly speaking, the 
%So active processors need to do this themselves, indirectly, without knowing which processors are inactive, i.e., we do not employ any form of failure detection for this algorithm.
Each active processor needs to do this itself, indirectly, without knowing which processor is inactive, i.e., we do not employ any form of failure detection for this algorithm.
To overcome this problem, each processor maintains bounded FIFO histories on labels appearing to have been created by other processors. 
These histories eventually accumulate the obsolete labels of the inactive processors. 
The reader may already see that maintaining these histories, also creates another source of possible corrupt labels.
We show that even in the presence of (a minority of) inactive processors, starting from an arbitrary state, the system eventually converges to use a global maximal label.
}

\trnsfr{
Let us explain why obsolete labels from inactive processors can create a problem when no one ever cleans (cancels) them up. 
Consider a system starting in a state that includes a cycle of labels $\ell_1 \prec \ell_2 \prec \ell_3 \prec \ell_1$, all of the same creator, say $p_x$, where $\prec$ is a relation between labels.
If $p_x$ is active, it will eventually learn about these labels and introduce a label greater than them all. 
But if $p_x$ is inactive, the system's asynchronous nature may cause a repeated cyclic label adoption, especially when $p_x$ has the greatest processor identifier, as these identifiers are used to break symmetry.
Say that an active processor learns and adopts $\ell_1$ as its global maximal label. 
Then, it learns about $\ell_2$ and hence adopts it, while forgetting about $\ell_1$. 
Then, learning of $\ell_3$ it adopts it. Lastly, it learns about $\ell_1$, and as it is greater than $\ell_3$, it adopts $\ell_1$ once more, as the greatest in the system; this can continue indefinitely. By using the bounded FIFO histories, such labels will be accumulated in the histories and hence will not be adopted again, ending this vicious cycle.
\bl{We now formally present the algorithm. }
}

\subsubsection{The Labeling Algorithm}
The labeling algorithm (Algorithm~\ref{alg:WFR}) specifies how the processors exchange their label information in the asynchronous system  and how they maintain proper label bookkeeping so as to ``discover'' their greatest label and cancel all obsolete ones. \bl{Specifically, we define the abstract task of the algorithm as one that lets every node to maintain a variable that holds the local maximal label.  We require that, after the recovery period and as long as there are no %[[@@It seems that the word `no' is missing here.@@]] 
calls to $nextLabel()$ (Algorithm~\ref{alg:NB}), these local maximal label actually refer to the same global maximal label.}

As we will be using pairs of labels with the {\em same} label creator, for the ease of presentation, we will be referring to these two variables as the {\em (label) pair}.
The first label in a pair is called $ml$. The second label  
is called $cl$ and it is either $\bot$, or equal to a label that cancels $ml$ (i.e., $cl$ indicates whether $ml$ is an obsolete label or not). 

\Paragraph{The processor state}
Each processor stores an array of label pairs, $max_i[n]$, where $max_i[i]$ refers to $p_i$'s maximal label pair and $max_i[j]$ considers the most recent value %[[@@ should it be label pair? @@]] 
that $p_i$ knows about $p_j$'s pair. 
Processor $p_i$ also stores the pairs of the most-recently-used labels in the array of queues $storedLabels_i[n]$. 
The $j$-th entry refers to the queue with pairs from $p_j$'s domain, i.e., that were created by $p_j$. 
The algorithm makes sure that $storedLabels_i[j]$ includes only label pairs with unique $ml$ from $p_j$'s domain and that at most one of them is \emph{legitimate}, i.e., not canceled. 
Queues $storedLabels_i[j]$ for $i\neq j$, have size $n+m$ whilst $storedLabels_i[i]$ has size $2(mn+2n^2-2n)$ 
%[[@@ check the size of this term. Is it correct? @@ Io. AMENDED ]] 
where $m$ is the system's total link capacity in labels.
We later show (c.f. \Arguments~\ref{th:boundedAdopt} and~\ref{th:boundedCreatorSteps}) that these queue sizes are sufficient to prevent overflows of useful labels.

\Paragraph{\bl{High level description}}
\bl{Each pair of processors periodically exchange their maximal label pairs and the maximal label pair that they know of the recipient. 
Upon receipt of such a label pair couple, the receiving processor starts by checking the integrity of its data structures and upon finding a corruption it flushes its label history queues.
It then moves to see whether the two labels that it received can cancel any of its non-canceled labels and if the received labels themselves can be canceled by labels that it has in its history.
Upon finishing this label housekeeping, it tries to find its local maximal view, first among the non-cancelled labels that other processors report as maximal, and if not such exist among its own labels. 
In latter case, if no such label exists, it generates a new one with a call to Algorithm~\ref{alg:NB} and using its own label queue as input.
At the end of the iteration the processor is guaranteed to have a maximal label, and continues to receive new label pair couples from other processors.
}

%\singlespacing
\begin{algorithm*}[t!]
  \caption{{\bl{Practically-}Self-Stabilizing Labeling Algorithm; code for $p_i$}}
%
%\begin{scriptsize}
%
\label{alg:WFR}

\begin{scriptsize}

{\bf Variables:}\\
$max[n]$ of $\langle ml$, $cl \rangle$: $max[i]$ is $p_i$'s largest label pair, $max[j]$ refers to $p_j$'s label pair (canceled when $max[j].cl \neq \bot$).\\

$storedLabels[n]$: an array of queues of the most-recently-used label pairs, where $storedLabels[j]$ holds the labels created by $p_j \in P$. For $p_j \in (P \setminus \{ p_i \})$, $storedLabels[j]$'s queue size is limited to $(n+m)$ w.r.t. label pairs, where $n=|P|$ is the number of processors in the system and $m$ is the maximum number of label pairs that can be in transit in the system. The $storedLabels[i]$'s queue size is limited to $(n(n^2+m))$ pairs. The operator $add(\ell)$ adds $lp$ to the front of the queue, and $emptyAllQueues()$ clears all $storedLabels[]$ queues. We use $lp.remove()$ for removing the record $lp \in storedLabels[]$. Note that an element is brought to the queue front every time this element is accessed in the queue.\\%~~~\\

%
%and $replace(\ell)$ updates a legitimate entry named by $\ell.name$ by canceling that entry using $\ell.cl$.
%
%The function $storedLabels[j].head()$ returns to queue head.

%$myLabels$: most-recently-used queue of labels that are associated to the creator $p_i \in P$. The total size of this queue is bounded by $x(n-1) + n^2$.\\
%~~\\

%\end{scriptsize}

{\bf Notation:} Let $y$ and $y'$ be two records that include the field $x$. We denote  $y$ $=_{x}$ $y'$ $\equiv$ $(y.x$ $=$ $y'.x)$\\

{\bf Macros:}\\
$legit(lp)$ $=$ $(lp$ $=$ $\langle \bullet, \bot \rangle)$~~~\\
%$incomparable(\ell, \ell') = ((\ell.name \not \preceq_{b} \ell'.name) \land (\ell'.name \not \preceq_{b} \ell.name))$~~~\\
$labels(lp)$ $:$ {\Return{$(storedLabels[lp.ml.lCreator])$}}\\
$double(j, lp) = (\exists lp' \in storedLabels[j] : ((lp \neq lp') \land ((lp =_{ml} lp') \lor (legit(lp) \land legit(lp')))))$~~~\\ \label{ln:double}
$staleInfo() = (\exists p_j \in P, lp \in storedLabels[j] : (lp \neq_{lCreator} j)
\lor double(j, lp))$~~~\\ \label{ln:staleInfo}
%$staleInfo() = \{ \ell \in labels(\ell) : (\exists \ell' \in labels(\ell) \land \ell \neq \ell' \land \ell.name = \ell'.name) \}$~~~\\
%
$recordDoesntExist(j) = (\langle max[j].ml, \bullet \rangle \notin labels(max[j]))$~~~\\
$notgeq(j, lp) = \mathbf{if~}(\exists lp' \in storedLabels[j]$ $:$ %$((lp =_{lCreator} lp' =_{lCreator} j)$ $\land$ 
$(lp'.ml \not \preceq_{lb} lp.ml))$ $\mathbf{then~return}(lp'.ml)$ $\mathbf{else~return}(\bot)$~~~\\
$canceled(lp) = \mathbf{if~}(\exists lp' \in labels(lp)$ $:$ $((lp' =_{ml} lp)$ $\land$ $\neg legit(lp')))$ $\mathbf{then~return}(lp')$ $\mathbf{else~return}(\langle \bot, \bot \rangle)$~~~\\
$needsUpdate(j)=(\neg legit(max[j]) \land \langle max[j].ml, \bot \rangle \in labels(max[j]))$~~\\
$legitLabels() = \{ max[j].ml : \exists p_j \in P \land legit(max[j]) \}$~\label{ln:legitLabels}~\\
$useOwnLabel()=\mathbf{if~}(\exists lp \in storedLabels[i] : legit(lp))$ $\mathbf{then~}max[i]$ $\gets$ $lp$ $\mathbf{else~}storedLabels[i].add(max[i]$ $\gets$ $\langle nextLabel(), \bot \rangle)$~\label{ln:useOwnLabelDef}
\tcp{For every $lp \in storedLabels[i]$, we pass in $nextLabel()$ both $lp.ml$ and $lp.cl$.}
%What is passed to nextLabel? $\{\ell: \ell \in (storedLabels[i].ml \cup  storedLabels[i].cl)\}$

{\bf upon} $transmitReady(p_j \in P \setminus \{ p_i \})$ {\bf do transmit}$(\langle max[i], max[j] \rangle)$~~\label{ln:transmit}\\%month~~\\

{\bf upon} $receive(\langle sentMax, lastSent \rangle)$ {\bf from} $p_j$ \label{ln:uponReceive}

\Begin{
$max[j]$ $\gets$ $sentMax$\; \label{ln:exposeStore}
\lIf{$\neg legit(lastSent)$ $\land$ $max[i] =_{ml} lastSent$}{$max[i] \gets lastSent$} \label{ln:lastSentCancel}
%$\process()$\;\label{ln:receiveProcess} 
%}%~~~\\

%{\bf do forever} $\process()$;\ %~~\\

%{\bf procedure} $\process()$

%\Begin
%{

    \lIf{$staleInfo()$}{$storedLabels.emptyAllQueues()$} \label{ln:clean}

     \lForEach{$p_j \in P : recordDoesntExist(j)$}{$labels(max[j]).add(max[j])$} \label{ln:add}

        \lForEach{$p_j \in P, lp \in storedLabels[j] : (legit(lp) \land (notgeq(j,lp)\neq \bot))$}{$lp.cl \gets notgeq(j,lp)$} \label{ln:cancelLabels}

        \lForEach{$p_j \in P, lp \in labels(max[j]) : (\neg legit(max[j]) \land (max[j] =_{ml} lp) \land legit(lp))$}{$lp \gets max[j]$} \label{ln:receivedCanceled}

        \lForEach{$p_j \in P, lp \in storedLabels[j] : double(j, lp)$}{$lp.remove()$} \label{ln:remove}

        \lForEach{$p_j \in P : (legit(max[j]) \land (canceled(max[j])\neq \langle \bot, \bot \rangle))$}{$max[j] \gets canceled(max[j])$} \label{ln:cancelMax}

    \lIf{$legitLabels() \neq \emptyset$}{$max[i] \gets \langle \max_{\prec_{lb}}(legitLabels()), \bot \rangle$} \label{ln:adopt}

    \lElse{$useOwnLabel()$}  \label{ln:useOwnLabel}

            %\lIf{$\exists! \ell \in storedLabels[i] : legit(\ell)$}{$max[i] \gets \ell$} \label{ln:recycle}

            %\lElse{$storedLabels[i].add(max[i] \gets \langle nextLabel(), \bot \rangle)$} \label{ln:create}

%    }
}

\end{scriptsize}

\end{algorithm*}

\Paragraph{Information exchange between processors}
Processor $p_i$ takes a step whenever it receives two pairs $\langle sentMax$, $lastSent \rangle$ from some other processor.
We note that in a legal execution $p_j$'s pair includes both $sentMax$, which refers to $p_j$'s
maximal label pair $max_j[j]$, and $lastSent$, which refers to a recent label pair that $p_j$ received from $p_i$ about $p_i$'s maximal label, $max_j[i]$ (line~\ref{ln:transmit}).

Whenever a processor $p_j$ sends a pair $\langle sentMax$, $lastSent \rangle$ to $p_i$, this processor stores the value of the arriving  $sentMax$ field in $max_i[j]$ (line~\ref{ln:exposeStore}). 
%Note that in a legal execution the arriving $sentMax$ is always legitimate. 
%[[@@ Make sure that the following 3 sentences are clearer to the reader. @@]]
%However, when $p_j$ acknowledges $p_i$'s label, it is possible that $p_j$ needs to inform $p_i$ of a label from $p_i$'s domain that cancels $p_i$'s  maximal label, $ml$ in $max_i[i]$. It does so by sending to $p_i$ a label that cancels $ml$ and thus it would be the case, $lastSent$ will have a $lastSent.cl$, that is not $\bot$. 
%Specifically, it contains a label that $p_j$ knows such that $lastSent.cl \not \preceq_{lb} lastSent.ml$, i.e., $lastSent.cl$ is either greater or incomparable to $lastSent.ml$. 
%[[@@ Io. Are the following 4 sentences clearer than the previously 3 when phrased this way? @@]]
However, $p_j$ may have local knowledge of a label from $p_i$'s domain that cancels $p_i$'s maximal label, $ml$, of the last received $sentMax$ from $p_i$ to $p_j$ that was stored in $max_j[i]$. 
Then $p_j$ needs to communicate this canceling label in its next communication to $p_i$.
To this end, $p_j$ assigns this canceling label to $max_j[i].cl$ which stops being $\bot$. 
Then $p_j$ transmits $max_j[i]$ to $p_i$ as a $lastSent$ label pair, and this satisfies $lastSent.cl \not \preceq_{lb} lastSent.ml$, i.e., $lastSent.cl$ is either greater or incomparable to $lastSent.ml$. 
This makes $lastSent$ illegitimate and in case this still refers to $p_i$'s current maximal label, $p_i$ must cancel $max_i[i]$ by assigning it with $lastSent$ (and thus $max_i[i].cl = lastSent.cl$) as done in line~\ref{ln:lastSentCancel}. 
Processor $p_i$ then processes the two pairs received   (lines~\ref{ln:clean} to~\ref{ln:useOwnLabel}).

\Paragraph{Label processing}
Processor $p_i$ takes a step whenever it receives a new pair message  $\langle sentMax$, $lastSent \rangle$ from processor $p_j$ (line~\ref{ln:uponReceive}). Each such step starts by removing \emph{stale} information, i.e., misplaced or doubly represented labels (line~\ref{ln:staleInfo}). 
In the case that stale information exists, the algorithm empties the entire label storage. 
Processor $p_i$ then tests whether the arriving two pairs are already included in the label storage ($storedLabels[]$), otherwise it includes them (line~\ref{ln:add}). 
%The algorithm continues to see whether $p_i$ already knows that the arriving label to be canceled or about the existence of a label $\ell$ that is not greater or equal than the arriving one. 
The algorithm continues to see whether, based on the new pairs added to the label storage, it is possible to cancel a non-canceled label pair (which may well be the newly added pair).
In this case, the algorithm updates the canceling field of any label pair $lp$ (line~\ref{ln:cancelLabels}) with the canceling label of a label pair $lp'$ such that $lp'.ml \not \preceq_{lb} lp.ml$ (line~\ref{ln:cancelLabels}). It is implied that since the two pairs belong to the same storage queue, they have the same processor as creator. 
The algorithm then checks whether any pair of the $max_i[]$ array can cause canceling to a record in the label storage (line~\ref{ln:receivedCanceled}), and also line~\ref{ln:remove} removes any canceled records that share the same creator identifier.
The test also considers the case in which the above update may cancel any arriving label in $max[j]$ and updates this entry accordingly based on stored pairs (line~\ref{ln:cancelMax}).

After this series of tests and updates, the algorithm is ready to decide upon a maximal label based on its local information. 
This is the $\preceq_{lb}$-greatest legit  label pair among all the ones in $max_i[]$ with respect to their $ml$ label (line~\ref{ln:adopt}). 
When no such legit label exists, $p_i$ requests a legit label in its own label storage, $storedLabels_i[i]$, and if one does not exist, will create a new one if needed (line~\ref{ln:useOwnLabel}). 
This is done by passing the labels in the $storedLabels_i[i]$ queue to the $nextLabel()$ function. 
Note that the returned label is coupled with a $\bot$ as the $cl$ and the resulting label pair is added to both $max_i[i]$ and $storedLabel_i[i]$. %\vspace{-.8em} %[[Clarify what we pass in nextLabel]]

\subsection{Correctness proof}
We are now ready to show the correctness of the algorithm. We begin with a proof overview.

\Paragraph{Proof overview}
The proof considers a execution $R$ of Algorithm~\ref{alg:WFR} that may initiate in an arbitrary configuration (and include a processor that takes practically infinite number of steps). It starts by showing some basic facts, such as: (1) stale information is removed, i.e., $storedLabels_i[j]$ includes only unique copies of $p_j$'s labels, and at most one legitimate such label (Corollary~\ref{th:noStaleInfo}), and (2) $p_i$ either adopts or creates the $\preceq_{lb}$-greatest legitimate local label (\Argument~\ref{th:localNondecreasing}). The proof then presents bounds on the number adoption steps (\Arguments~\ref{th:boundedAdopt} and~\ref{th:boundedCreatorSteps}), that define the required queue sizes to avoid label overflows.

The proof continues to show that active processors can eventually stop adopting or creating labels, by tackling individual cases where canceled or incomparable label pairs may cause a change of the local maximal label.
%%[[[In are particularity interested in looking into cases in which canceled label pairs and incomparable ones.]]]
We show that such labels eventually disappear from the system (\Argument~\ref{th:riskEmpty}) and thus no new labels are being adopted or created (\Argument~\ref{th:emptyRisk}), which then implies the existence of a global maximal label (\Argument~\ref{th:boundedDiffusion}). Namely, there is a legitimate label $\ell_{\max}$, such that for any processor $p_i \in P$ (that takes a practically infinite number of steps in $R$), it holds that $max_i[i]=\ell_{\max}$.
Moreover, for any processor $p_j \in P$ that is active throughout the execution, it holds that  $p_i$'s local maximal (legit) label pair $max_i[i] = \ell_{max}$ is the $\preceq_{lb}$-greatest of all the label pairs in $max_i[]$ and there is no label pair in  $storedLabels_i[j]$  that cancels $\ell_{max}$, 
i.e., $((max_i[j].ml \preceq_{lb} \ell_{\max}.ml) \land ((\forall \ell \in storedLabels_i[j]: legit(\ell)) \Rightarrow$ $(\ell.ml \preceq_{lb} \ell_{\max}.ml)))$.
%i.e., $((\forall \ell \in max_i[\bullet]: legit(\ell)) \Rightarrow (\ell.ml \preceq_{lb} \ell_{\max}.ml)) \land ((\forall \ell' \in storedLabels_i[j]:\ell_{max}.ml.lCreator = j) \Rightarrow ((\ell'.ml \preceq_{lb} \ell_{\max}.ml)\land (\ell'.cl \preceq_{lb} \ell_{\max}.ml)))$.
%[[[Moreover, for any forever active processor $p_j \in P$, it holds that $((max_i[j] \preceq_{lb} \ell_{\max}) \land ((\forall \ell \in storedLabels_i[j]: legit(\ell)) \Rightarrow (\ell \preceq_{lb} \ell_{\max})))$.]]]
%
We then demonstrate that, when starting from an initial arbitrary configuration, the system eventually reaches a configuration in which there is a global maximal label (Theorem~\ref{th:oneGreatest4All}). 
%[[@@ Can the following be moved further down or removed completely? We say a lot on this in the counter increment proof. Removing it makes things cleaner for this section .@@]]
%Note that the convergence result also holds when starting from a configuration, $c \in R$, that is obtained by taking a configuration $c'$ in which %$risk = \emptyset$ 
%there are no labels that can cause a cancellation of the local maximum of any processor,
%and then apply at least one step in which one processor $p_i \in P$ (that takes practically infinite number of steps in $R$) calls $useOwnLabel()$. 
%This will be useful in later cases when we deal with counter exhaustion. 

\noindent Before we present the proof in detail, we provide some helpful definitions and notation.

\Paragraph{Definitions} 
%\Subsection{Definitions.}
We define ${\cal H}$ to be the set of all label pairs that can be in transit in the system, with $|{\cal H}| = m$. So in an arbitrary configuration, there can be up to $m$ corrupted label pairs in the system's links. We also denote ${\cal H}_{i,j}$ as the set of label pairs that are in transit from processor $p_i$ to processor $p_j$. 
 The number of label pairs in ${\cal H}_{i,j}$ obeys the link capacity bound.
Recall that the data structures used (e.g., $max_i[]$, $storedLabels_i[]$, etc) store label pairs. For convenience of presentation and when clear from the context, 
%when we will be referring to a label rather than a label pair we mean the $ml$ part of the pair. When we refer to a \emph{legitimate} label we essentially mean that the $cl$ part of the label is $\bot$.
we may refer to the $ml$ part of the label pair as  ``the label''.
\bl{Note that in this algorithm, we consider an \emph{iteration} as the execution of lines~\ref{ln:uponReceive}--\ref{ln:useOwnLabel}, i.e., the receive action.}

\subsubsection{No stale information}
\Argument~\ref{th:onlyOneDoubleRecord} says that the predicate $staleInfo()$ (line~\ref{ln:staleInfo}) can only hold during the first execution of the $receive()$ event (line~\ref{ln:uponReceive}).

\begin{argument}
\label{th:onlyOneDoubleRecord}
Let $p_i \in P$ be a processor for which $\neg staleInfo_i()$ (line~\ref{ln:staleInfo}) does not hold during the $k$-th step in $R$ that includes the complete execution of the $receive()$ event (from line~\ref{ln:uponReceive} to~\ref{ln:useOwnLabel}). Then $k=1$.
\end{argument}

\begin{proof}
Since $R$ starts in an arbitrary configuration, there could be a queue in $storedLabels_i[]$ that holds two label records from the same creator, a label that is not stored according to its creator identifier, or more than one legitimate label. Therefore, $staleInfo_i()$ might hold during the first execution of the $receive()$ event. 
When this is the case, the $storedLabels_i[]$ structure is emptied (line~\ref{ln:clean}).
During that $receive()$ event execution (and any event execution after this), $p_i$ adds records to a queue in $storedLabels_i[]$ (according to the creator identifier) only after checking whether $recordDoesntExist()$ holds (line~\ref{ln:add}).

%However, as we show, during that event execution (and any event execution after) $p_i$ adds records to a queues in $storedLabels_i[]$ (according to the creator identifier) only after checking whether $recordDoesntExist()$ holds (line~\ref{ln:add}). 
Any other access to $storedLabels_i[]$ merely updates cancelations or removes duplicates. Namely, canceling labels that are not the $\preceq_{lb}$-greatest among the ones that share the same creating processors (line~\ref{ln:cancelLabels}) and canceling records that were canceled by other processors (line~\ref{ln:receivedCanceled}), as well as removing legitimate records that share the same $ml$ (line~\ref{ln:remove}).
It is, therefore, clear that in any subsequent iteration of $receive()$ (after the first), $staleInfo()$ cannot hold.
\end{proof}

\Argument~\ref{th:onlyOneDoubleRecord} along with the lines~\ref{ln:staleInfo} and~\ref{ln:cancelMax} of the Algorithm, imply Corollary~\ref{th:noStaleInfo}.

\begin{corollary}
\label{th:noStaleInfo}
%
%[[@@ Break the next sentence and make it clearer. @@]] After any step that includes the execution of the $receive()$ event, other than the first one, it holds that $\forall p_i, p_j \in P$, the state of $p_i$ encodes at most one legitimate label, $\ell_j =_{lCreator} j$, which appears in $storedLabels_i[j]$ rather than $storedLabels_i[k]:k\neq j$ and possibly in $max_i[]$ as well.
%[[@@ Io. Rephrased @@]] 
Consider a suffix $R'$ of execution $R$ that starts after the execution of a $receive()$ event. 
Then the following hold throughout $R'$: (i)  $\forall p_i, p_j \in P$, the state of $p_i$ encodes at most one legitimate label, $\ell_j =_{lCreator} j$ and (ii) $\ell_j$ can only appear in $storedLabels_i[j]$ and  $max_i[]$ but not in $storedLabels_i[k]:k\neq j$. 
\end{corollary}

\subsubsection{Local \texorpdfstring{$\preceq_{lb}$}{<lb}-greatest legitimate local label}

{\sloppy
\Argument~\ref{th:localNondecreasing} considers processors for which $staleInfo()$ (line~\ref{ln:staleInfo}) does not hold. Note that $\neg staleInfo()$ holds at any time after the first step that includes the $receive()$ event (\Argument~\ref{th:onlyOneDoubleRecord}). \Argument~\ref{th:localNondecreasing} shows that $p_i$ either adopts or creates the $\preceq_{lb}$-greatest legitimate local label pair and stores it in $max_i[i]$.

}

\begin{argument}
\label{th:localNondecreasing}
{\sloppy
Let $p_i \in P$ be a processor such that $\neg staleInfo_i()$ (line~\ref{ln:staleInfo}), and $L_{pre}(i) = \{ max_i[j].ml :  \exists p_j \in P \land legit(max_i[j]) \land (\exists \langle max_i[j].ml, x \rangle \in (labels(max_i[j]) \setminus \{ max_i[j] \}) \Rightarrow (x = \bot)) \}$  be the set of $max_i[]$'s labels that, before $p_i$ executes lines~\ref{ln:clean} to~\ref{ln:useOwnLabel}, are legitimate both in $max_i[]$ and in $storedLabels_i[]$'s queues. Let $L_{post}(i)=\{ max_i[j].ml : \exists p_j \in P \land legit(max_i[j]) \}$ and $\langle \ell, \bot \rangle$ be the value of $max_i[i]$ immediately after $p_i$ executes lines~\ref{ln:clean} to~\ref{ln:useOwnLabel}. The label $\langle \ell, \bot \rangle$ is the $\preceq_{lb}$-greatest legitimate label in $L_{post}(i)$. Moreover, suppose that $L_{pre}(i)$ has a $\preceq_{lb}$-greatest legitimate label pair, then that label pair is $\langle \ell, \bot \rangle$.

}
\end{argument}

%$(rec_{lgt} =(\langle \ell', \bullet \rangle \in storedLabels[j]) : rec_{lgt} \neq \langle \ell', \bot \rangle)$
\begin{proof}
\noindent {\bf $\mathbf{\langle \ell, \bot \rangle}$ is the $\mathbf{\preceq_{lb}}$-greatest legitimate label pair in $\mathbf{L_{post}(i)}$.}~~~~
Suppose that immediately before line~\ref{ln:adopt}, we have that $legitLabels_i() \neq \emptyset$, where $legitLabels_i() = \{ max_i[j].ml : \exists p_j \in P \land legit(max_i[j]) \}$ (line~\ref{ln:legitLabels}). Note that in this case $L_{post}(i)=legitLabels_i()$. By the definition of $\preceq_{lb}$-greatest legitimate label pair and line~\ref{ln:adopt}, $max_i[i] = \langle \ell, \bot \rangle$ is the $\preceq_{lb}$-greatest legitimate label pair in $L_{post}(i)$. Suppose that $legitLabels_i() = \emptyset$ immediately before line~\ref{ln:adopt}, i.e., there are no legitimate labels in $\{ max_i[j] : \exists p_j \in P  \}$. By the definition of $\preceq_{lb}$-greatest legitimate label pair and line~\ref{ln:useOwnLabelDef}, $max_i[i] = \langle \ell, \bot \rangle$ is the $\preceq_{lb}$-greatest legitimate label pair in $L_{post}(i)$.

\noindent \textbf{Suppose that $\mathbf{rec=\langle \ell', \bot \rangle}$ is a $\mathbf{\preceq_{lb}}$-greatest legitimate label pair in $\mathbf{L_{pre}(i)}$, then $\mathbf{\ell = \ell'}$.}~~~~
We show that the record $rec$ is not modified in $max_i[]$ until the end of the execution of lines~\ref{ln:clean} to~\ref{ln:useOwnLabel}. Moreover, the records that are modified in $max_i[]$, are not included in $L_{pre}(i)$ (it is canceled in $storedLabels_i[]$) and no records in $max_i[]$ become legitimate. Therefore, $rec$ is also the $\preceq_{lb}$-greatest legitimate label pair in $L_{post}(i)$, and thus, $\ell = \ell'$.

Since we assume that $staleInfo_i()$ does not hold, line~\ref{ln:clean} does not modify $rec$. Lines~\ref{ln:add},~\ref{ln:cancelLabels} and~\ref{ln:remove} might add, modify, and respectively, remove $storedLabels_i$'s records, but it does not modify $max_i[]$. Since $rec$ is not canceled in $storedLabels_i[]$ and the $\preceq_{lb}$-greatest legitimate label pair in $max_i[]$, the predicate $(legit(max[j]) \land notgeq(j))$ does not hold and line~\ref{ln:cancelLabels} does not modify $rec$. Moreover, the records in $max_i[]$, for which that predicate holds, become illegitimate.~\end{proof}

\subsubsection{Bounding the number of labels}
\noindent \Arguments~\ref{th:boundedAdopt} and~\ref{th:boundedCreatorSteps} present bounds on the number of adoption steps. 
%[[@@ Elad, say how it used later in the proof. @@]]
These are $n+m$ for labels by labels that become inactive in any point in $R$ and $(mn+2n^2-2n)$ 
%[[@@ check the size of this term. Is it correct? @@ Io. AMENDED]] 
for any active processor.
Following the above, choosing the queue sizes as $n+m$ for $storedLabels_i[j]$ if $i\neq j$, and $2(nm+2n^2-2n)+1$ %[[@@ check the size of this term. Is it correct? @@ Io. addressed]] 
for $storedLabels_i[i]$ is sufficient to prevent overflows given that $m$ is the system's total link capacity in labels.

\Paragraph{Maximum number of label adoptions in the absence of creations}
Suppose that there exists a processor, $p_j$, that has stopped adding labels to the system (the else part of line~\ref{ln:useOwnLabel}), say, because it became inactive (crashed), or it names a maximal label that is the $\preceq_{lb}$-greatest label pair among all the ones that the network ever delivers to $p_j$. \Argument~\ref{th:boundedAdopt} bounds the number of labels from $p_j$'s domain that any processor $p_i \in P$ adopts in $R$. 
%for a practically infinite period [[@@ We need to discuss this issue and how to say  it right. @@]]

%[[@@ Lemma 4.3: problem with Lemma numbering from this point and on @@ Io. What seems to be the problem? Should theorems lemmas and corollaries have same counters?]]

\begin{argument}
\label{th:boundedAdopt}
Let $p_i,p_j \in P$, be two processors. Suppose that $p_j$ has stopped adding labels to the system configuration (the else part of line~\ref{ln:useOwnLabel}), and sending (line~\ref{ln:transmit}) these labels during $R$. Processor $p_i$ adopts (line~\ref{ln:adopt}) at most $(n+m)$ labels, $\ell_j : (\ell_j =_{lCreator} j)$, from $p_j$'s unknown domain ($\ell_j \notin labels_i(\ell_j)$) where $m$ is the maximum number of label pairs that can be in transit in the system. 
%[[@@ Is this the right definition, as Yiannis wrote, or just the ones in transit, as I actually use it in the proofs.  @@ where $m$ is the maximal number of labels that may exists in any configuration (either in transit over the communication links or at the processor states).]]
\end{argument}

%[[@@ Check the addition about ''unknown to it, $\ell_j \in labels_i(\ell_j)$.'' @@]]
\begin{proof}
Let $p_k \in P$. At any time (after the first step in $R$) processor $p_k$'s state encodes at most one legitimate label, $\ell_j$, for which $\ell_j =_{lCreator} j$ (Corollary~\ref{th:noStaleInfo}). 
Whenever $p_i$ adopts a new label $\ell_j$ from $p_j$'s domain (line~\ref{ln:adopt}) such that $\ell_j : (\ell_j =_{lCreator} j)$, this implies that $\ell_j$ is the only legitimate label pair in $storedLabels_i[j]$.
Since $\ell_j$ was not transmitted by $p_j$ before it was adopted, $\ell_j$ must come from $p_k$'s state delivered by a transmit event (line~\ref{ln:transmit}) or delivered via the network as part of the set of labels that existed in the initial arbitrary state.
%[[[Whenever processor $p_i$ adopts (line~\ref{ln:adopt}) a label, $\ell_j : (\ell_j =_{lCreator} j)$ from $p_j$'s domain, because $\ell_j$ is not added to the $p_i$'s state (the else part of line~\ref{ln:useOwnLabel}) and sent (line~\ref{ln:transmit}) during $R$ (as in this argument statement).
%Thus, $\ell_j$ must come from $p_k$'s state delivered by a transmit event (line~\ref{ln:transmit}) or delivered via the network as part of the set of labels that existed in the initial arbitrary state. ]]]
The bound holds since there are $n$ processors, such as $p_k$, and $m$ bounds the number of labels in transit.
Moreover, no other processor can create label pairs from the domain of $p_j$.
\end{proof}

\Paragraph{Maximum number of label creations}
\Argument~\ref{th:boundedCreatorSteps} shows a bound on the number of adoption steps that does not depend on  whether the labels are from the domain of an active or (eventually) inactive processor.
%[[@@ Ioannis: Note that the title refers to bounding the creations while the argument is bounding adoptions. We should connect the two at some point. @@]]

\begin{argument}
\label{th:boundedCreatorSteps}
Let $p_i \in P$ and $L_i=\ell_{i_0}, \ell_{i_1}, \ldots$ be the sequence of legitimate labels, $\ell_{i_k} =_{lCreator} i$, from $p_i$'s domain, which $p_i$ stores in $max_i[i]$ through the reception  (line~\ref{ln:uponReceive}) or creation of labels (line~\ref{ln:useOwnLabel}), where $k \in \mathbb{N}$. It holds that $|L_i| \leq n(n^2+m)$.
%Let $p_i \in P$ and $L_i=\ell_{i_0}, \ell_{i_1}, \ldots$ be the sequence of legitimate labels, $\ell_{i_k} =_{lCreator} i$, from $p_i$'s domain, which $p_i$ stores in $max_i[i]$ over time, where $k \in \mathbb{N}$. [[@@ what does over time mean here. @@]] It holds that $|L_i| \leq n(n^2+m)$.
\end{argument}

%[[@@ Recheck the proof arguments. Are we using L and C in the right way. @@ Io. I believe they are OK, added $^c$ to further disambiguate C and L elements.]]\\
\begin{proof}
Let $L_{i,j}=\ell_{i_0,j}, \ell_{i_1,j}, \ldots$ be the sequence of legitimate labels that $p_i$ stores in $max_i[j]$ during $R$ and $C_{i,j}=\ell^c_{i_0,j}, \ell^c_{i_1,j}, \ldots$ be the sequence of legitimate labels that $p_i$ receives from processor $p_j$'s domain. We consider the following cases in which $p_i$ stores $L$'s values in $max_i[i]$.

\noindent \textbf{(1) When $\mathbf{\ell_{i_k}=\ell_{j_0,{j'}}}$, where $\mathbf{p_j,p_{j'} \in P}$ and $\mathbf{k \in \mathbb{N}}$.}~~~~This case considers the situation in which $max_i[i]$ stores a label that appeared in $max_j[{j'}]$ at the (arbitrary) starting configuration, (i.e. $\ell_{j_0,{j'}} \in L_{j,j'}$). There are at most $n(n-1)$ such legitimate label values from $p_i$'s domain, namely $n-1$ arrays $max_j[]$ of size $n$.

\noindent \textbf{ (2) When $\mathbf{\ell_{i_k}=\ell_{j_{k'},{j'}} = \ell^c_{j_0,{j'}}}$, where $\mathbf{p_j,p_{j'} \in P}$, $\mathbf{k,k' \in \mathbb{N}}$ and $\mathbf{\ell_{j_{k'},{j'}} \neq \ell_{j_{k'},j}}$.} %L_{j_{k'},j}$.}
~~~~This case considers the situation in which $max_i[i]$ stores a label that appeared in the communication channel between $p_j$ and $p_{j'}$ at the (arbitrary) starting configuration, (i.e. $\ell^c_{j_0,{j'}} \in C_{j,j'}$) and appeared in $max_j[{j'}]$ before $p_j$ communicated this to $p_i$. 
There are at most $m$ such values, i.e., as many as the capacity of the communication links in labels, namely $|\mathcal{H}|$. 
%[[@@ Make sure that it is $m$ or $n^2m$ (in case we think about the capacity of every channel). @@]]
%[[@@ Elad: Maybe explain here more. @@ @@ Io: Addressed: I extended the last sentence.]]

\noindent \textbf{ (3) When $\mathbf{\ell_{i_k}}$ is the return value of $\mathbf{nextLabel()}$ (the else part of line~\ref{ln:useOwnLabel}).} ~~~~ Processor $p_i$ aims at adopting the $\preceq_{lb}$-greatest legitimate label pair that is stored in $max_i[]$, whenever such exists (line~\ref{ln:adopt}). Otherwise, $p_i$ uses a label from its domain; either one that is the $\preceq_{lb}$-greatest legit label pair among the ones in $storedLabels_i[i]$, whenever such exists, or the returned value of  $nextLabel()$ (line~\ref{ln:useOwnLabel}).

The latter case (the else part of line~\ref{ln:useOwnLabel}) refers to labels, $\ell_{i_k}$, that $p_i$ stores in $max_i[i]$ only after checking that there are no legitimate labels stored in $max_i[]$ or $storedLabels_i[i]$. Note that every time $p_i$ executes the else part of line~\ref{ln:useOwnLabel}, $p_i$ stores the returned label, $\ell_{i_k}$, in $storedLabels_i[i]$. After that, there are only three events for $\ell_{i_k}$ not to be stored as a legitimate label in $storedLabels_i[i]$:\\ ($i$) execution of line~\ref{ln:clean}, ($ii$) the network delivers to $p_i$ a label, $\ell'$, that either cancels $\ell_{i_k}$ ($\ell'.cl \not \preceq_{lb} \ell_{i_k}.ml$), or for which $\ell'.ml \not \preceq_{lb} \ell_{i_k}.ml$, and ($iii$) $\ell_{i_k}$ overflows from $storedLabels_i[i]$ after exceeding the $(n(n^2+m)+1)$ limit which is the size of the queue.

Note that \Argument~\ref{th:onlyOneDoubleRecord} says that event ($i$) can occur only once (during $p_i$'s first step). Moreover, only $p_i$ can generate labels that are associated with its domain (in the else part of line~\ref{ln:useOwnLabel}). Each such label is $\preceq_{lb}$-greater-equal than all the ones in $storedLabels_i[i]$ (by the definition of $nextLabel()$ in Algorithm~\ref{alg:NB}).

\sloppy{Event $(ii)$ cannot occur after $p_i$ has learned all the labels $\ell \in remoteLabels_{i}$ for which $\ell \notin storedLabels_i[i]$, where $remoteLabels_{i}$ $=$ $(((\cup_{p_j \in P}$ $localLabels_{i,j})$ $\cup$ ${\cal H})$ $\setminus$ $storedLabels_{i}[i])$ and $localLabels_{i,j}$ $=$ $\{ \ell'$ $:$ $\ell' =_{lCreator} i,$ $\exists p_j$ $\in$ $P$ $:$ $((\ell'$ $\in$ $storedLabels[i])$ $\lor$ $(\exists$ $p_k$ $\in$ $P$ $:$ $\ell'$ $=$ $max_j[k].ml))\}$.} During this learning process, $p_i$ cancels or updates the cancellation labels in $storedLabels_{i}[i]$ before adding a new legitimate label. Thus, this learning process can be seen as moving labels from $remoteLabels_{i}$ to $storedLabels_i[i]$ and then keeping at most one legitimate label available in $storedLabels_i[i]$. 
Every time $storedLabels_i[i]$ accumulates a label $\ell$ that was unknown to $p_i$, the use of $nextLabel()$ allows it to create a label $\ell_{i_k}$ that is $\preceq_{lb}$-greater than any label pair in $storedLabels_i[i]$ and eventually from all the ones in $remoteLabels_{i}$.

%%%%%%%%%%%%%%%%%%%% END OF RESTRUCTURING OF CALCULATION %%%%%%%%%%%%%%%%%%%%

%[[@@ Elad: The bound calculation question. $(n-1)(2n+m)+m=2n^2+mn-2n-m+m=2n^2+mn-2n=n(2n+m-2)$. We need to justify this. @@]]
%[[@@ Io. The following paragraph provides a more accurate bound and a more elaborate explanation than before.]]
Note that $remoteLabels_{i}$'s labels must come from the (arbitrary) start of the system, because $p_i$ is the only one that can add a label to the system from its domain and therefore this set cannot increase in size.
These labels include those that are in transit in the system and all those that are unknown to $p_i$ but exist in the $max_j[\bullet]$ or $storedLabels_j[i]$ structures of some other processor $p_j$.
By Lemma~\ref{th:boundedAdopt} we know that $|storedLabels_j[i]|\leq n+m$ for $i \neq j$.
From the three cases of $L_i$ labels that we detailed at the beginning of this proof ((1)--(3)), we can bound the size of $remoteLabels_{i}$ as follows: for $p_j \in P : j \neq i$ we have that $|remoteLabels_{i}|$ $\leq$ $(n-1)(|max[]| + |storedLabels_j[i]|)+|{\cal H}|=(n-1)(n+(n+m))+m=mn+2n^2-2n$. 
Since $p_i$ may respond to each of these labels with a call to $nextLabel()$, we require that $storedLabels_i[i]$ has size $2|remoteLabels_i|+1$ label pairs in order to be able to accommodate all the labels from $|remoteLabels_i|$ and the ones created in response to these, plus the current greatest.
Thus, what is suggested by event $(ii)$ of $p_i$, i.e., receiving labels from $remoteLabels_{i}$, stops happening before overflows (event $(iii)$) occurs, since $storedLabels_i[i]$ has been chosen to have a size that can accommodate all the labels from $remoteLabels_{i}$ and those created by $p_i$ as a response to these.
\bl{This size is $2(mn+2n^2-2n)+1 = 2(n^3cap+2n^2-2n)+1$ (since $m=n^2cap$) which is $O(n^3)$.}
\end{proof}

\bl{From the end of the proof of Lemma~\ref{th:boundedCreatorSteps}, we get Corollary~\ref{thm:kSize}.}

\begin{corollary}
\label{thm:kSize}
\bl{The number $k$ of antistings needed by Algorithm~\ref{alg:NB} is $2\cdot(2(n^3cap+2n^2-2n)+1)$ (twice the queue size).}
\end{corollary}

\subsubsection{Pair diffusion}
The proof continues and shows that active processors can eventually stop adopting or creating labels. We are particularly interested in looking into cases in which there are canceled label pairs and incomparable ones. We show that they eventually disappear from the system (\Argument~\ref{th:riskEmpty}) and thus no new labels are being adopted or created (\Argument~\ref{th:emptyRisk}), which then implies the existence of a global maximal label (\Argument~\ref{th:boundedDiffusion}).  
 
\Arguments~\ref{th:riskEmpty} and~\ref{th:emptyRisk}, as well as \Argument~\ref{th:boundedDiffusion} and Theorem~\ref{th:oneGreatest4All} assume the existence of at least one processor, $p_{unknown} \in P$ whose identity is unknown, that takes practically infinite number of steps in $R$. Suppose that processor $p_i \in P$ takes a bounded number of steps in $R$ during a period in which $p_{unknown}$ takes a practically infinite number of steps. We say that $p_i$ has become inactive (crashed) during that period and assume that it does not resume to take steps at any later stage of $R$ (in the manner of fail-stop failures, as in Section~\ref{s:sys}).

Consider a processor $p_i \in P$ that takes any number of (bounded or practically infinite) steps in $R$ and two processors $p_j, p_k \in P$ that take a practically infinite number of steps in $R$. 
%Ioannis
Given that $p_j$ has a label pair $\ell$ as its local maximal, and there exists another label pair $\ell'$ such that $(\ell'.ml \not \preceq_{lb} \ell.ml) \lor \ell'.cl \not \preceq_{lb} \ell.ml$ and they have the same creator $p_i$.
Algorithm~\ref{alg:WFR} suggests only two possible routes for some label pair $\ell'$ to find its way in the system through $p_j$. Either by $p_j$ adopting $\ell'$ (line~\ref{ln:adopt}), or by creating it as a new label (the else part of line~\ref{ln:useOwnLabel}).
Note, however, that $p_j$ is not allowed to create a label in the name of $p_i$ and since $\ell' =_{lCreator} i$, the only way for $\ell'$ to disturb the system is if this is adopted by $p_j$ as in line~\ref{ln:adopt}.
%endIoannis
We use the following definitions for estimating whether there are such label pairs as $\ell$ and $\ell'$ in the system.
%We use the following definitions for estimating whether there are  label pairs, $\ell$ and $\ell'$, that have the potential to disturb the system by bringing $p_j$ to either add a label, $\ell_j =_{lCreator} i$, to the system configuration (the else part of line~\ref{ln:useOwnLabel}), or adopt labels (line~\ref{ln:adopt}), $\ell_j : (\ell_j =_{lCreator} i)$, from $p_i$'s
%%[[[$\ell_j : (\ell_j =_{lCreator} j)$, from $p_j$'s]]] AMMENDED ABOVE
%unknown domain {{\color{blue}($\ell_j \notin storedLabels_j[i])$} 
%%[[[($\ell_j \notin labels_j(\ell_j)$)]]].

There is a \emph{risk} for two label pairs from $p_i$'s domain, $\ell_j$ and $\ell_k$, %from $p_i$'s domain
to cause such a disturbance when either they cancel one another or when it can be found that one is not greater than the other.
Thus, we use the predicate $risk_{i,j,k}(\ell_j,\ell_k)=(\ell_j =_{i} \ell_k) \land legit(\ell_j) \land (notGreater(\ell_j, \ell_k) \lor canceled(\ell_j, \ell_k))$ to estimate whether $p_j$'s state encodes a label pair, $\ell_j =_{lCreator} i$, from $p_i$'s domain that may disturb the system due to another label, $\ell_k$, from $p_i$'s domain that $p_k$'s state encodes, where $canceled(\ell_j, \ell_k)=(legit(\ell_j) \land \neg legit(\ell_k) \land \ell_j =_{ml} \ell_k)$ refers to a case in which label $\ell_j$  is canceled by label $\ell_k$, $notGreater(\ell_j, \ell_k)=(legit(\ell_j) \land legit(\ell_k) \land \ell_k.ml \not \preceq_{lb} \ell_j.ml)$ that refers to a case in which label $\ell_k$ is not $\preceq_{lb}$-greater than $\ell_j$ and $(\ell_j =_i \ell_k) \equiv (\ell_j =_{lCreator} \ell_k =_{lCreator} i)$.

\begin{table}[t]

\begin{footnotesize}

\centering
\begin{tabular}[t]{|c|l| m{3.5cm}|}

\hline
\textbf{Notation} & \textbf{Definition}  & \textbf{Remark} \\
\hline
$hName_{i,j,k}$ & $\{ (\ell_j, \ell_k) : \ell_j = max_j[j] \land (\exists \langle \ell_k, \bullet \rangle \in {\cal H}_{k,j}) \}$ & In transit from $p_k$ to $p_j$ as $sentMax$ feedback about $max_k[k]$ \\ 
\hline 
$hAck_{i,j,k}$ & $\{ (\ell_j, \ell_k) : \ell_j = max_j[k] \land (\exists \langle \bullet, \ell_k \rangle \in {\cal H}_{k,j}) \}$ & In transit from $p_k$ to $p_j$ as $lastSent$ feedback about $max_k[j]$ \\ 
\hline  
$max_{i,j,k}$ & $\{ ( max_j[j], max_k[k]) \}$ & Local maximal labels of $p_j$ and $p_k$ \\ 
\hline 
$ack_{i,j,k}$ & $\{ ( max_j[j], max_k[j]) \}$ & $\ell_j$ is $p_j$'s local maximal label and $\ell_k = max_k[j]$ \\ 
\hline 
$stored_{i,j,k}$ & $\{\{ max_j[j] \} \times storedLabels_k[i]\}$ & A label $\ell_k$ in $storedLabels_k[i]$ that can cancel $\ell_j = max_j[j]$ \\ 
\hline 

\end{tabular}

\end{footnotesize}

\caption{\hangindent=4em The notation used to identify the possible positions of label pairs $\ell_j$ and $\ell_k$ that can cause canceling as used in \Arguments~\ref{th:riskEmpty} to \ref{th:boundedDiffusion} and in Theorem~\ref{th:oneGreatest4All}.}
\label{tab:LabelThmNotation}
\end{table}

These two label pairs, $\ell_j$ and $\ell_k$, can be the ones that processors $p_j$ and $p_k$ name as their local maximal label, as in $max_{i,j,k} = \{ ( max_j[j], max_k[k]) \}$, or recently received from one another, as in $ack_{i,j,k} = \{ ( max_j[j], max_k[j]) \}$. These two cases also appear when considering the communication channel (or buffers) from $p_k$ to $p_j$, as in 
%[[@@ Note that use of the hider notation here. We need another way to refer to the label pairs that are on transit. @@ 
$hName_{i,j,k}  = \{ (\ell_j, \ell_k) : \ell_j = max_j[j] \land (\exists \langle \ell_k, \bullet \rangle \in {\cal H}_{k,j}) \}$ and 
$hAck_{i,j,k}  = \{ (\ell_j, \ell_k) : \ell_j = max_j[k] \land (\exists \langle \bullet, \ell_k \rangle \in {\cal H}_{k,j}) \}$.
%]] 
We also note the case in which $p_k$ stores a label pair that might disturb the one that $p_j$ names as its (local) maximal, as in $stored_{i,j,k} = \{\{ max_j[j] \} \times storedLabels_k[i]\}$
We define the union of these cases to be the set 
$risk = \{ (\ell_j, \ell_k) \in max_{i,j,k} \cup ack_{i,j,k} \cup hName_{i,j,k} \cup hAck_{i,j,k} \cup stored_{i,j,k} : \exists p_i, p_j, p_k \in P \land stopped_j \land stopped_k \land risk_{i,j,k}(\ell_j,\ell_k) \}$, where $stopped_i = true$ when processor $p_i$ is inactive (crashed) and $false$ otherwise.
The above notation can also be found in Table~\ref{tab:LabelThmNotation}.

%, where $max_{i,j,k} = \{ ( max_j[j], max_k[k]) \}$, $ack_{i,j,k} = \{ ( max_j[j], max_k[j]) \}$, $hName_{i,j,k}  = \{ (\ell_j, \ell_k) : \ell_j = max_j[j] \land (\exists \langle \ell_k, \bullet \rangle \in {\cal H}_{k,j}) \}$, $hAck_{i,j,k}  = \{ (\ell_j, \ell_k) : \ell_j = max_j[j] \land (\exists \langle \bullet, \ell_k \rangle \in {\cal H}_{k,j}) \}$, $stored_{i,j,k} = \{  \{ max_j[j] \} \times storedLabels_k[i]  \}$, where $stopped_i = true$ when processor $p_i$ is inactive (crashed) and $false$ otherwise.

%We define $risk = \{ (\ell_j, \ell_k) \in max_{i,j,k} \cup ack_{i,j,k} \cup hName_{i,j,k} \cup hAck_{i,j,k} \cup stored_{i,j,k} : \exists p_i, p_j, p_k \in P \land stopped_j \land stopped_k \land risk_{i,j,k}(\ell_j,\ell_k) \}$, where $max_{i,j,k} = \{ ( max_j[j], max_k[k]) \}$, $ack_{i,j,k} = \{ ( max_j[j], max_k[j]) \}$, $hName_{i,j,k}  = \{ (\ell_j, \ell_k) : \ell_j = max_j[j] \land (\exists \langle \ell_k, \bullet \rangle \in {\cal H}_{k,j}) \}$, $hAck_{i,j,k}  = \{ (\ell_j, \ell_k) : \ell_j = max_j[j] \land (\exists \langle \bullet, \ell_k \rangle \in {\cal H}_{k,j}) \}$, $stored_{i,j,k} = \{  \{ max_j[j] \} \times storedLabels_k[i]  \}$, where $stopped_i = true$ when processor $p_i$ is inactive (crashed) and $false$ otherwise.

\begin{argument}
\label{th:riskEmpty}
Suppose that there exists at least one processor, $p_{unknown} \in P$ whose identity is unknown, that takes practically infinite number of steps in $R$ during a period where $p_j$ never adopts labels (line~\ref{ln:adopt}), $\ell_j : (\ell_j =_{lCreator} i)$, from $p_i$'s unknown domain ($\ell_j \notin labels_j(\ell_j)$).
Then eventually $risk = \emptyset$ .
%Suppose that there exists at least one processor, $p_{unknown} \in P$ whose identity is unknown, that takes practically infinite number of steps in $R$ during a period whilst neither $p_j$ adds a label, 
%$\ell_j =_{lCreator} i$, to the system (the else part of line~\ref{ln:useOwnLabel}), nor $p_j$ adopts labels (line~\ref{ln:adopt}), $\ell_j : (\ell_j =_{lCreator} i)$, from $p_i$'s unknown domain ($\ell_j \notin labels_j(\ell_j)$).
%Then $risk = \emptyset$ eventually.
\end{argument}

\begin{proof}
%[[@@ Ioannis: A more refined version of the above. @@]]\\
Suppose this \Argument~ is false, i.e., the assumptions of this \Argument~hold and yet in any configuration $c \in R$, it holds that $(\ell_j, \ell_k) \in risk \neq \emptyset$. We use $risk$'s definition to study the different cases. By the definition of $risk$, we can assume, without the loss of generality, that $p_j$ and $p_k$ are alive throughout $R$. 
\vspace{.5em}

\setlength{\leftskip}{.7cm}

\noindent{\bf Claim:} If $p_j$ and $p_k$ are alive throughout $R$, i.e. $stopped_j = stopped_k = \emph{\texttt{\em False}}$, then $risk \neq \emptyset \iff risk_{i,j,k} = \emph{\texttt{\em  True}}$. 
This means that there exist two label pairs $(\ell_j, \ell_k)$ where  $\ell_k$ can force a cancellation to occur. Then the only way for this two labels to force $risk \neq \emptyset$ is if, throughout the execution, $\ell_k$ never reaches $p_j$.
\vspace{.5em}

\setlength{\leftskip}{0pt}

The above claim is verified by a simple observation of the algorithm. 
If $\ell_k$ reaches $p_j$ then lines~\ref{ln:lastSentCancel}, \ref{ln:receivedCanceled}  and~\ref{ln:cancelMax} guarantee a canceling and lines~\ref{ln:add} and~\ref{ln:cancelLabels} ensure that these labels are kept canceled inside $storedLabels_j[]$. 
The latter is also ensured by the bounds on the labels given in \Arguments~\ref{th:boundedAdopt} and~\ref{th:boundedCreatorSteps} that do not allow queue overflows. 
Thus to include these two labels to $risk$, is to keep $\ell_k$  hidden from $p_j$ throughout $R$.
We perform a case-by-case analysis to show that it is impossible for label $\ell_k$ to be ``hidden'' from $p_j$ for an infinite number of steps in $R$.

\noindent {\bf The case of $(\ell_j, \ell_k) \in hName_{i,j,k}$.}~~~~~~This is the case where $\ell_j = max_j[j]$ and $\ell_k$ is a label in $\mathcal{H}_{k,j}$ that appears to be $max_k[k]$. 
This may also contain such labels from the corrupt state.
We note that $p_j$ and $p_k$ are alive throughout $R$. 
The stabilizing implementation of the data-link ensures that a message cannot reside in the communication channel during an infinite number of $transmit()$ -- $receive()$ events of the two ends. Thus $\ell_k$, which may well have only a single instance in the link coming from the initial corrupt state, will either eventually reach $p_j$ or it become lost.
In the both cases (the first by the Claim for the second trivially) the two clashing labels are removed from $risk$ and the result follows.

\noindent {\bf The case of $(\ell_j, \ell_k) \in hAck_{i,j,k}$.}~~~~~~
This is the case where $\ell_j = max_j[j]$ and $\ell_k$ is a label in $\mathcal{H}_{k,j}$ that appears to be $max_k[j]$.
The proof line is exactly the same as the previous case.

This case follows by the same arguments to the case of $(\ell_j, \ell_k) \in ack_{i,j,k}$.

\noindent {\bf The case of $(\ell_j, \ell_k) \in max_{i,j,k}$.}~~~~~
Here the label pairs $\ell_j$ and $\ell_k$ are named by $p_j$ and $p_k$ as their local maximal label.  
We note that $p_j$ and $p_k$ are alive throughout $R$. By our self-stabilizing data-links and by the assumption on the communication that a message sent infinitely often is received infinitely often, then $p_k$ transmits its $max_k[k]$ label infinitely often when executing line~\ref{ln:transmit}.
This implies that $p_j$ receives $\ell_k$ infinitely often. 
By the Claim the canceling takes place, and the two labels are eventually removed from the global observer's $risk$ set, giving a contradiction.

\noindent \textbf{ The case of $\mathbf{(\ell_j, \ell_k) \in ack_{i,j,k}}$.}~~~~~~ 
This is the case  case where the labels $(\ell_j, \ell_k)$ belong to $\{ ( max_j[j], max_k[j]) \}$.
Since processor $p_k$ continuously transmits its label pair in $max_k[j]$ (line~\ref{ln:transmit}) the proof is almost identical to the previous case. 

\noindent \textbf{ The case of $\mathbf{(\ell_j, \ell_k) \in stored_{i,j,k}}$.}~~~~
%$stored_{i,j,k} = \{\{ max_j[j] \} \times storedLabels_k[i]\} $
%
This case's proof, follows by similar arguments to the case of $(\ell_j, \ell_k) \in max_{i,j,k}$. Namely, $p_k$ eventually receives the label pair $\ell_j = max_j[j]$. The assumption that $risk_{i,j,k}(\ell_j,\ell_k)$ holds implies that one of the tests in lines~\ref{ln:cancelLabels} and \ref{ln:cancelMax} will either update $storedLabels_k[i]$, and respectively, $max_k[j]$ with canceling values. We note that for the latter case we argue that $p_j$ eventually received the canceled label pair in $max_k[j]$, because we assume that $p_j$ does not change the value of $max_j[j]$ throughout $R$. 

By careful and exhaustive examination of all the cases, we have proved that there is no way to to keep $\ell_k$  hidden from $p_j$ throughout $R$.
This is a contradiction to our initial assumption, and thus eventually $risk = \emptyset$.   
\end{proof}

These two label pairs, $\ell_j$ and $\ell_k$, can be the ones that processors $p_j$ and $p_k$ name as their local maximal label, as in $max_{i,j,k} = \{ ( max_j[j], max_k[k]) \}$, or recently received from one another, as in $ack_{i,j,k} = \{ ( max_j[j], max_k[j]) \}$. These two cases also appear when considering the communication channel (or buffers) from $p_k$ to $p_j$, as in $hName_{i,j,k}  = \{ (\ell_j, \ell_k) : \ell_j = max_j[j] \land (\exists \langle \ell_k, \bullet \rangle \in {\cal H}_{k,j}) \}$ and 
$hAck_{i,j,k}  = \{ (\ell_j, \ell_k) : \ell_j = max_j[j] \land (\exists \langle \bullet, \ell_k \rangle \in {\cal H}_{k,j}) \}$. We also note the case in which $p_k$ stores a  label pair that might disturb the one that $p_j$ names as its (local) maximal, as in $stored_{i,j,k} = \{  \{ max_j[j] \} \times storedLabels_k[i]  \}$. %, where $stopped_i = true$ when processor $p_i$ is inactive (crashed) and $false$ otherwise.

\begin{argument}
\label{th:emptyRisk}
Suppose that $risk = \emptyset$ in every configuration throughout $R$ and that there exists at least one processor, $p_{unknown} \in P$ whose identity is unknown, that takes practically infinite number of steps in $R$. Then $p_j$ \emph{never} adopts labels (line~\ref{ln:adopt}), $\ell_j : (\ell_j =_{lCreator} i)$, from $p_i$'s unknown domain ($\ell_j \notin labels_j(\ell_j)$).
%Suppose that $risk = \emptyset$ in every configuration throughout $R$ and that there exists at least one processor, $p_{unknown} \in P$ whose identity is unknown, that takes practically infinite number of steps in $R$. Neither $p_j$ adds a label, $\ell_j =_{lCreator} i$, to the system (the else part of line~\ref{ln:useOwnLabel}), nor $p_j$ adopts labels (line~\ref{ln:adopt}), $\ell_j : (\ell_j =_{lCreator} i)$, from $p_i$'s unknown domain ($\ell_j \notin labels_j(\ell_j)$).
\end{argument}

\begin{proof}
%[[@@ Suppose this argument statement is false, i.e., the assumptions of this argument hold and yet either $p_j$ adds a label, $\ell_j =_{lCreator} i$, to the game state (the else part of line~\ref{ln:useOwnLabel}), or $p_j$ adopts labels (line~\ref{ln:adopt}), $\ell_j : (\ell_j =_{lCreator} i)$, from $p_i$'s unknown domain ($\ell_j \notin labels_j(\ell_j)$).
%
%\noindent {\bf The case of $p_j$ adds a label, $\ell_j =_{lCreator} i$, [[@@ $p_j$ cannot add to the system configuration labels from $p_i$'s domain. @@]] to the game state (the else part of line~\ref{ln:useOwnLabel}).}~~~~~~
%
%\noindent {\bf The case of $p_j$ adopts labels from $p_i$'s unknown domain.}~~~~~~
%
%
%However, that contradicts the assumption that this argument statement is false in the first line of this proof.@@]]
%
Note that the definition of $risk$ considers almost every possible combination of two label pairs $\ell_j$ and $\ell_k$ from $p_i$'s domain that are stored by processor $p_j$, and respectively, $p_k$ (or in the channels to them). The only combination that is not considered is $(\ell_j, \ell_k) \in storedLabels_j[i] \times storedLabels_k[i]$. 
However, this combination can indeed reside in the system during a legal execution and it cannot lead to a disruption for the case of $risk = \emptyset$ in every configuration throughout $R$ because before that could happen, either $p_j$ or $p_k$ would have to adopt $\ell_j$, and respectively, $\ell_k$, which means a contradiction with the assumption that $risk = \emptyset$. 

%% Ioannis modifications
The only way that a label in $storedLabels[]$ can cause a change of the local maximum label and be communicated to also disrupt the system, is to find its way to $max[]$.
Note that $p_j$ cannot create a label under $p_i$'s domain (line~\ref{ln:useOwnLabel}) since the algorithm does not allow this, nor can it adopt a label from $storedLabels_j[i]$ (by the definition of $legitLabels()$, line~\ref{ln:legitLabels}).
So there is no way for $\ell_j$ to be added to $max_j[j]$ and thus make $risk \neq \emptyset$ through creation or adoption.

{\sloppy 
On the other hand, we note that there is only one case where $p_k$ extracts a label  from $storedLabels_k[i]: i \neq k$ and adds it to $max_k[j]$. 
This is when it finds a legit label $\ell_j \in max_k[j]$ that can be canceled by some other label $\ell_k$ in $storedLabels_k[i])$, line~\ref{ln:cancelMax}.
But this is the case of having the label pair $(\ell_j, \ell_k)$ in $stored_{i,j,k}$. 
Our assumption that $risk = \emptyset$ implies that $stored_{i,j,k} = \emptyset$. This is a contradiction.
Thus a label $\ell_k$ cannot reach $max_k[]$ in order for it to be communicated to $p_j$.

}

In the same way we can argue for the case of two messages in transit, ${\cal H}_{j,k} \times {\cal H}_{k,j}$ and that $risk = \emptyset$ throughout R.   
\end{proof}
%
%% END OF Ioannis modifications

\begin{argument}
\label{th:boundedDiffusion}
Suppose that $risk = \emptyset$ in every configuration throughout $R$ and that there exists at least one processor, $p_{unknown} \in P$ whose identity is unknown, that takes practically infinite number of steps in $R$.
There is a legitimate label $\ell_{\max}$, such that for any processor $p_i \in P$ (that takes a practically infinite number of steps in $R$), it holds that $max_i[i]=\ell_{\max}$.
Moreover, for any processor $p_j \in P$ (that takes a practically infinite number of steps in $R$), it holds that $((max_i[j].ml \preceq_{lb} \ell_{\max}.ml) \land ((\forall \ell \in storedLabels_i[j]: legit(\ell)) \Rightarrow$ $(\ell.ml \preceq_{lb} \ell_{\max}.ml)))$.
\end{argument}

\begin{proof}
We initially note that the two processors $p_i, p_j$ that take an infinite number of steps in $R$ will exchange their local maximal label $max_i[i]$ and $max_j[j]$ an infinite number of times.
By the assumption that  $risk = \emptyset$, there are no two label pairs in the system that can cause canceling to each other that are unknown to $p_i$ or $p_j$ and are still part of $max_i[i]$ or $max_i[j]$.
Hence, any differences in the local maximal label of the processors must be due to the labels' $lCreator$ difference.

Since $max_i[i]$ and $max_j[j]$ are continuously exchanged and received, assuming $max_i[i].ml \prec_{lb} max_j[j].ml$ where the labels are of different label creators, then $p_i$ will be led to a $receive()$ event of $\langle sentMax_j, lastSent_j \rangle$ where $max_i[i].ml \prec_{lb} sentMax_j.ml$.
By line~\ref{ln:exposeStore}, $sentMax_j$ is added to $max_i[j]$ and since $risk = \emptyset$ no action from line~\ref{ln:lastSentCancel} to line~\ref{ln:cancelMax} takes place.
Line~\ref{ln:adopt} will then indicate that the greatest label in $max_i[\bullet]$ is that in $max_i[j]$ which is then adopted by $p_i$ as $max_i[i]$, i.e., $p_i$'s local maximal.
The above is true for every pair of processors taking an infinite number of steps in $R$ and so we reach to the conclusion that eventually all such processors converge to the same $\ell_{max}$ label, i.e., it holds that $((max_i[j].ml \preceq_{lb} \ell_{\max}.ml) \land ((\forall \ell \in storedLabels_i[j]: legit(\ell)) \Rightarrow$ $(\ell.ml \preceq_{lb} \ell_{\max}.ml)))$.
%
%% END OF Ioannis modifications
%
\end{proof}

\subsubsection{Convergence}
Theorem~\ref{th:oneGreatest4All} combines all the previous lemmas to demonstrate that when starting from an arbitrary starting configuration, the system eventually reaches a configuration in which there is a global maximal label.

\begin{theorem}
\label{th:oneGreatest4All}
Suppose that there exists at least one processor, $p_{unknown} \in P$ whose identity is unknown, that takes practically infinite number of steps in $R$. Within a bounded number of steps, there is a legitimate label pair $\ell_{\max}$, such that for any processor $p_i \in P$ (that takes a practically infinite number of steps in $R$), it holds that $p_i$ has $max_i[i]=\ell_{\max}$. %when naming its (local) maximal label, $max_i[i].ml$. 
Moreover, for any processor $p_j \in P$ (that takes a practically infinite number of steps in $R$), it holds that 
%$((max_i[j] \preceq_{lb} \ell_{\max})$ $\land$ $(\forall \ell \in storedLabels_i[j]: legit(\ell)) \Rightarrow (\ell \preceq_{lb} \ell_{\max})))$.
$((max_i[j].ml \preceq_{lb} \ell_{\max}.ml) \land ((\forall \ell \in storedLabels_i[j]: legit(\ell)) \Rightarrow$ $(\ell.ml \preceq_{lb} \ell_{\max}.ml)))$.
\end{theorem}

\begin{proof}
For any processor in the system, which may take any (bounded or practically infinite) number of steps in $R$, we know that there is a bounded number of label pairs, $L_i=\ell_{i_0}, \ell_{i_1}, \ldots$, that processor $p_i \in P$ adds to the system configuration (the else part of line~\ref{ln:useOwnLabel}), where $\ell_{i_k} =_{lCreator} i$ (\Argument~\ref{th:boundedCreatorSteps}). Thus, by the pigeonhole principle we know that, within a bounded number of steps in $R$, there is a period during which $p_{unknown}$ takes a practically infinite number of steps in $R$ whilst (all processors) $p_i$ do not add any label pair, $\ell_{i_k} =_{lCreator} i$, to the system configuration (the $else$ part of line~\ref{ln:useOwnLabel}). 

During this practically infinite period (with respect to $p_{unknown}$), in which no label pairs are added to the system configuration due to the $else$ part of line~\ref{ln:useOwnLabel}, we know that for any processor $p_j \in P$ that takes any number of (bounded or practically infinite) steps in $R$, and processor $p_k \in P$ that adopts labels in $R$ (line~\ref{ln:adopt}), $\ell_j : (\ell_j =_{lCreator} j)$, from $p_j$'s unknown domain %($\ell_j \notin labels_k(\ell_j)$), 
($\ell_j \notin storedLabels_k(j)$) it holds that $p_k$ adopts such labels (line~\ref{ln:adopt}) only a bounded number times in $R$ (\Argument~\ref{th:boundedAdopt}). Therefore, we can again follow the pigeonhole principle and say that there is a period during which $p_{unknown}$ takes a practically infinite number of steps in $R$ whilst neither $p_i$ adds a label, $\ell_{i_k} =_{lCreator} i$, to the system (the else part of line~\ref{ln:useOwnLabel}), nor $p_k$ adopts labels (line~\ref{ln:adopt}), $\ell_j : (\ell_j =_{lCreator} j)$, from $p_j$'s unknown domain ($\ell_j \notin labels_k(\ell_j)$). 

We deduce that, when the above is true, then we have reached a configuration in $R$ where $risk = \emptyset$ (\Argument~\ref{th:riskEmpty}) and remains so throughout $R$~(\Argument~\ref{th:emptyRisk}).
\Argument~\ref{th:boundedDiffusion} concludes by proving that, whilst $p_{unknown}$ takes a practically infinite number of steps, all processors (that take practically infinite number of steps in $R$) name the same $\preceq_{lb}$-greatest legitimate label pair which the theorem statement specifies.
Thus no label $\ell =_{lCreator} j$ in $max_i[\bullet]$ or in $storedLabels_i[j]$ may satisfy $\ell.ml \not \preceq_{lb} \ell_{max}.ml$.        
\end{proof}

\subsubsection{\bl{Algorithm complexity}}%
\bl{The required local memory of a processor comprises of a queue of size (in labels) $2(n^3cap+2n^2-2n)$ that hosts the labels with the processor as a creator (Corollary~\ref{thm:kSize}). The local state also includes $n-1$ queues of size $n+n^2cap$ to store labels by other processors, and a single label for the maximal label of every processor. 
We conclude that the \emph{space complexity} is of order $\bigO(n^3)$ in labels.
Given the number of possible labels in the system by the same processor is $\beta=n^3cap+2n^2-2n$, as shown in the proof of Lemma~\ref{th:boundedCreatorSteps}, we deduce that the size of a label in bits is $\bigO(\beta \log \beta)$.
}

%[[@@EMS@@{Why do we talk about the stabilization time and not the number of deviations from the abstract task. Also, why just creations and not consider adaptation. I am saying that because an adaptation is also a deviation from the abstract task, which requires all labels to refer to the same maximal.}@@]]
\bl{
By Theorem~\ref{th:oneGreatest4All} we can bound the \emph{stabilization time} based on the number of label creations. 
Namely, in an execution with $\bigO({n\cdot \beta})$ label creations (e.g., up to $n$ processors can create $O(\beta)$ labels), there is a practically infinite execution suffix (of size $2^\tau$ iterations) where the receipt of a label which starts an iteration never changes the maximal label of any processor in the system.
}

\subsection{Increment Counter Algorithm} %\vspace{-.2em}
\label{subsec:CounterI}
We adjust the labeling algorithm to work with counters, so that our counter increment algorithm is a stand-alone algorithm.
In this subsection, we explain how we can enhance the labeling scheme presented in the previous subsection to obtain a practically (infinite) self-stabilizing counter increment algorithm.

% % % % % FROM SECTION 2. I HAVE EXCLUDED THIS AS REDUNDANT. 
%\trnsfr{
%Note that when a processor $p_i$ establishes a new label $\ell$ as the global maximum, it sets the corresponding counter $cnt = \langle \ell, 0, i\rangle$; in this case, the label creator identifier and the sequence number writer identifier is $i$. When there is an already established maximal label $\ell$ in the system and processor $p_i$ wants to increment the counter, it increases the corresponding (to $\ell$) maximal sequence number found ($maxseqn$) by one, and sets the counter  $cnt = \langle \ell, maxseqn+1, i\rangle$; in this case, it is possible that the label creator identifier and the sequence number writer identifier are not the same, i.e., if $p_i$ was not the creator of label $\ell$. Also, note that some extra care is needed with respect to counter bookkeeping so as not to increase the size of the bounded histories used in the labeling algorithm. 
%%
%Having a counter increment algorithm, it is not difficult to obtain a practically self-stabilizing MWMR register implementation; counters are associated with values and the counter increment algorithm is run with this small amendment (more details in Sect.~\ref{subsec:CounterI}).
%}

\begin{figure}[t]
	\begin{\algorithmFontSize}
\begin{framed}
{\bf Variables:} 
A label $lbl$ is extended to the triple $\langle lbl, seqn, wid \rangle$ called a \textit{counter} where $seqn$, is the sequence number related to $lbl$, and $wid$ is the identifier of the creator of this $seqn$. %(as detailed in section~\ref{subsec:CounterI}). 
A counter pair $\langle mct, cct\rangle$ extends a label pair. $cct$ is a canceling counter for $mct$, such that $cct.lbl \not \prec_{lb} mct.lbl$ or $cct.lbl = \bot$. 
We rename structures $max[]$ and $storedLabels[]$ of Alg.~\ref{alg:WFR} to $maxC[]$ and $storedCnts[]$ that hold counter pairs instead of label pairs.
Variable $status \in \{{\sf MAX\_REQUEST, MAX\_WRITE, COMPLETE}\}$. \\ %\label{algCt:var}\\
{\bf Operators:}
$add(ctp)$ - places a counter pair $ctp$ at the front of a queue. If $ctp.mct.lbl$ already exists in the queue, it only maintains the instance with the greatest counter w.r.t. $\prec_{ct}$, placing it at the front of the queue. If one counter pair is canceled then the canceled copy is retained. % \label{algCt:operations}\\ 
We consider an array field as a single sized queue and use $add()$.
%{\bf Notation:} Let $y$ and $y'$ be two records that include the field $x$. Denote  $y=_{x} y' \equiv (y.x = y'.x)$.
\end{framed}\vspace{-1em}
\caption{Variables and Operators for Counter Increment; code for $p_i$.}
\label{fig:ctrVars}
\end{\algorithmFontSize}
\end{figure}

\subsubsection{From labels to counters and to a counter version of Algorithm~\ref{alg:WFR}}
\noindent \textbf{Counters.} 
% % % % FOLLOWING PARAGRAPH IS A MERGER WITH SECTION 2
%To do so, 
To achieve this task, we now need to work with practically unbounded {\em counters}. 
%As already mentioned in Section~\ref{sec:nutshell},  a 
A counter $cnt$ is a triplet $\langle lbl, seqn, wid\rangle$, where $lbl$ is an epoch label as defined
in the previous subsection, $seqn$ is a $\tau$-bit integer sequence number and $wid$ is the identifier of the processor that last incremented the counter's sequence number, i.e., $wid$ is the counter \emph{writer}. 
Then, given two counters $cnt_i, cnt_j$ we define the relation $cnt_i \prec_{ct} cnt_j$ $\equiv$ $(cnt_i.lbl \prec_{lb} cnt_j.lbl)$ 
$\lor$ $((cnt_i.lbl = cnt_j.lbl) \wedge (cnt_i.seqn < cnt_j.seqn))$ $\lor$ $((cnt_i.lbl = cnt_j.lbl) \wedge (cnt_i.seqn = cnt_j.seqn)
\wedge (cnt_i.wid < cnt_j.wid))$. 
Observe that when the labels of the two counters are incomparable, the counters are also incomparable. %.
%Note that 

%\trnsfr{
%Using our labeling scheme, we show how to implement a practically infinite counter supporting multiple writers.  
%The idea is to extend the labeling scheme to handle {\em counters}, where a counter \bl{is a triple consisting of} a $label$, as used in the labeling scheme; 
%an integer {\em sequence number}, \bl{$seqn$}, ranging from $0$ to $2^{b}$, where $b$ is large enough, say $b=64$; and a processor (writer) identifier, \bl{$wid$}. 
%Conceptually, if the system stabilizes to use a global maximal label, then the pair of the sequence number and the processor identifier (of this sequence number) can be used as an unbounded counter, as used, for example, in MWMR register implementations~\cite{LSFTC97,RAMBO}. 
%Specifically, we say that counter $cnt_1=\langle \ell_1, seqn_1, wid_1\rangle$ is {\em smaller} than counter $cnt_2=\langle \ell_2, seqn_2, wid_2\rangle$ if ($\ell_1 \prec \ell_2$) or (($\ell_1 = \ell_2$) and ($seqn_1<seqn_2$)) or
%$((\ell_1 = \ell_2$) and ($seqn_1 = seqn_2$) and ($wid_1<wid_2))$. 
%Note that when processors have the same label, the above relation forms a total ordering and processors can increment a shared counter also when attempting to do so concurrently. We argue that starting from any initial configuration, eventually the counter algorithm supports such increments.
%}

The relation $\prec_{ct}$ defines a total order (as required by practically unbounded counters) for counters with the same label, thus, only when processors share a globally maximal label. %, i.e., the system runs within a ``stable" epoch. 
Conceptually, if the system stabilizes to use a global maximal label, then the pair of the sequence number and the processor identifier (of this sequence number) can be used as an unbounded counter, as used, for example, in MWMR register implementations~\cite{LSFTC97,RAMBO}. 
\remove{
As we have shown in Theorem~\ref{th:oneGreatest4All}, %the previous subsection, 
starting from an arbitrary configuration, we eventually reach a configuration where the active processors have adopted the same maximal label. 
In this case, %when processors have the same label, the above relation forms a total ordering and 
processors can increment a shared counter also when attempting to do so concurrently. 
In what follows we argue that starting from any initial configuration, eventually the counter algorithm supports such increments.
Essentially, the counter increment algorithm enhances the labeling algorithm to take care of the counter increment once such a maximal label exists in the system. %[[@@  elaborate on why the number of label creations does not increase. @@ Io. We explain this two Paragraphs below and formally in the proofs.]]
}

\begin{figure}[t]
\begin{algorithm}[H]

%\caption{Macros and Procedures for Algorithm~\ref{alg:counterOperationsMaintenance}; code for $p_i$}
%\RemoveAlgoNumber
%\begin{\algorithmFontSize}
\begin{\macroFontSize}
\tcp{Where macros coincide with Algorithm~\ref{alg:WFR} we do not restate them.}
{\bf Macros:} 
%	$exhausted(ctp)$ $=$ $(ctp.mct.seqn$ $\geq$ $2^{\tau})$\\
%	$legit(ctp)=(ctp.cct = \bot \rangle )$\\	
%	$staleCntrInfo() = staleInfo() \lor (\exists p_j \in P, x\in storedCnts[j]:exhausted(x)\land legit(x))$\\
%	$retCntrQ(ct) :$ {\bf return} $(storedCnts[ct.lbl.lCreator])$\\
%	$legitCnts()$ $=$ $\{maxC[j].mct: \exists p_j \in P \land legit(maxC[j])\} $\\
%	$cancelExhausted(ctp) :$ {$ctp.cct \gets ctp.mct$}\\  
%	$cancelExhaustedMaxC() :$ 	\lForEach{$p_j\in P,\ c \in maxC[j]: exhausted(c)$}{$cancelExhausted(maxC[j])$}  
%	 \label{algCt:cancExh}

	$exhausted(ctp)$ $=$ $(ctp.mct.seqn$ $\geq$ $2^{\tau})$\\
	$cancelExh(ctp) :$ {$ctp.cct \gets ctp.mct$}\\  
	$cancelExhMaxC() :$ 	\lForEach{$p_j\in P,\ c \in maxC[j]: exhausted(c)$}{$cancelExh(maxC[j])$}	 \label{algCt:cancExh}  
	$legit(ctp)=(ctp.cct = \bot \rangle )$ \label{cnt:legit}\\
	$staleCntrInfo() = staleInfo() \lor (\exists p_j \in P, x\in storedCnts[j]:exhausted(x)\land legit(x))$\\
	$retCntrQ(ct) :$ {\bf return} $(storedCnts[ct.lbl.lCreator])$\\
	$retMaxCnt(ct)$ $=$ {\bf return} $(max_{\prec_{ct}}(ct, ct'))$ \textbf{where}  $ct' \in retCntrQ(ct) \land (ct =_{lbl} ct')$\label{cnt:retMaxCnt}\\
	$legitCnts()$ $=$ $\{maxC[j].mct: \exists p_j \in P \land legit(maxC[j])\} $\\
	$useOwnCntr() =\mathbf{if~}(\exists cp \in storedCnts[i] : legit(lp))$ $\mathbf{then~}maxC[i]$ $\gets$ $cp$ $\mathbf{else~}storedCnts[i].add(maxC[i]$ $\gets$ $\langle \langle nextLabel(), 0, i\rangle, \bot \rangle)$~\label{cnt:useOwnLabelDef}
	\tcp{For every $cp \in storedCnts[i]$, we pass to $nextLabel()$ both $cp.mct.lbl$ and $cp.cct.lbl$.} 
	$getMaxSeq():$ {\bf return} $max_{wid} (\{max_{seqn}(\{ctp:ctp.mct \in legitCnts()$ $\land$ $maxC[i] =_{mct.lbl} ctp\})\})$\\
%	$initiateMaxRequest()= \{\textbf{let}$ $responseSet \gets  \emptyset$; $status \gets {\sf MAX\_REQUEST};\}$\\
%	$quorResponded(S)$ \textbf{return} $(\exists Q \in \mathbb{Q}: Q \subseteq \{ responseSet.pid\})$\\
%	$quorResponded(S)$ \textbf{return} $(\exists Q \in \mathbb{Q}: Q \subseteq \{responseSet\})$\\
	$initWrite =\{\langle maxC[i], responseSet, status \rangle \gets \langle maxC[i](), \emptyset, {\sf MAX\_WRITE} \rangle;\}$\\
	$increment()=$ $\{maxC[i] \gets \langle maxC[i].mct.lbl, maxC[i].mct.seqn+1, i\rangle;\}$\\
	$correctResponse(A,B)=$ \textbf{return} $((status = {\sf MAX\_REQUEST} \land (A, B \notin \{\bot \}))$ $\lor$ $((status$ $=$ ${\sf MAX\_WRITE})$ $\land$ $(\langle A,B\rangle = \langle \bot, maxC_i[i] \rangle)) $\\
%	$addCntr(\langle ctp_\rangle) = maintainCntr(\langle ctp, {\sf NULL}\rangle) \rangle$\\
%\BlankLine
%\tcp{The counter increment procedures.}

\end{\macroFontSize}
%\end{\algorithmFontSize}
\end{algorithm}\vspace{-0.4em}
\caption{Macros for Algorithm~\ref{alg:cntrIncrement}.}
\label{fig:ctrMacros}
\end{figure}

%\subsubsection{Enhancing the labeling algorithm to handle counters}
\noindent \textbf{Structures.}
We convert the label structures $max[]$ and $storedLabels[ ]$ of the labeling algorithm into the structures $maxC[ ]$ and $storedCnts[ ]$ that hold counters rather than labels (see Figure~\ref{fig:ctrVars}). %Recall that in the labeling algorithm each processor $p_i$ was maintaining two main structures of pairs of labels: array $max[]$ that stored the local maximal labels of each other processor (based on the message exchange) and $storedLabels[]$, an array of queues of label pairs that each processor maintains in an attempt to clean up obsolete labels created by itself or other processors. 
%These structures now need to contain counters instead of just labels and are renamed to $maxC[]$ and $storedCnts[]$ (see line~\ref{algCt:var} of Algorithm~\ref{alg:counterOperations}). 
Each label can yield many different counters with different $\langle seqn, wid \rangle$. 
Therefore, in order to avoid increasing the size of the queues of $storedCnts$ (with respect to the number of elements stored), we only keep  the highest sequence number observed for each label (breaking ties with the $wid$). 

This is encapsulated in the definition of the $add()$ operator (Figure~\ref{fig:ctrVars} -- Operators).
In particular, we define the operator $add(ctp)$ (Fig.~\ref{fig:ctrVars}) to enqueue a counter pair $ctp$ to a queue of $storedCnts[n]$, where in case a counter with the same label already exists,  the following two rules apply:
(1) if at least one of the two counters is canceled we keep a canceled instance, and (2) if both  counters are legitimate, we keep the greatest counter with respect to $\langle seqn, wid\rangle$.
The counter is placed at the front of the queue. 
In this way we allow for labels for which the counters have not been exhausted to be reused.
We denote a counter pair by $\langle mct, cct\rangle$, with this being the extension of a label pair $\langle ml, cl\rangle$, where $cct$ is a canceling counter for $mct$, such that either $cct.lbl \not \prec_{lb} mct.lbl$ (i.e., the counter is canceled), or $cct.lbl = \bot$.

%Also, note that if there are counters in the system that are corrupt (being in the initial arbitrary configuration), then they can only force a change of label if their sequence number is {\em exhausted} (i.e., $seqn \geq 2^\tau$). 
%Exhausted counters are treated by the counter algorithm in a way similar to the canceled labels in the labeling algorithm; an exhausted counter $cnt_i$ in a counter pair $\langle cnt_i, cnt_j\rangle$ is canceled, by setting $cnt_j.lbl = cnt_i.lbl$ (i.e., the counter's own label cancels it) and hence making the counter non-legit (thus it cannot be used as a local maximal counter in  $maxC_i[i]$).
\noindent \textbf{Exhausted counters}. These are the ones satisfying  $seqn \geq 2^\tau$, and they are treated %by the counter algorithm 
in a way similar to the canceled labels in the labeling algorithm; an exhausted counter $mct$ in a counter pair $\langle mct, cct\rangle$ is canceled, by setting $mct.lbl = cct.lbl$ (i.e., the counter's own label cancels it) and hence %making the counter non-legit (thus it 
cannot be used as a local maximal counter in  $maxC_i[i]$. %).
%Exhausted counters are treated by the algorithm in a similar way as canceled labels in the labeling algorithm; an exhausted counter $cnt$ is canceled, by setting $cnt.lbl.cl = cnt.lbl.ml$ (i.e., the counter's own label cancels it) and hence making the counter non-legit (thus it cannot be used as maxC). 
This cannot increase the number of labels that are created, since the initial set of corrupt counters remains the same as the one for labels, for which we have already produced a proof in Section~\ref{subsec:Labels}. %, as shown in the correctness proof that  follows. %as the total capacity of the links in the system still corresponds to $m$.

\begin{figure}[t!]
\begin{algorithm}[H]

%\caption{Counter Maintenance; code for $p_i$}
%
%\label{alg:cntrMaintenance}
%\begin{\algorithmFontSize}
\begin{\macroFontSize}
%\textbf{Macros:}\\
%\tcp{Where macros coincide with Algorithm~\ref{alg:WFR} we do not restate them.}
%	$exhausted(ctp)$ $=$ $(ctp.mct.seqn$ $\geq$ $2^{\tau})$\\
%	$cancelExh(ctp) :$ {$ctp.cct \gets ctp.mct$}\\  
%	$cancelExhMaxC() :$ 	\lForEach{$p_j\in P,\ c \in maxC[j]: exhausted(c)$}{$cancelExh(maxC[j])$}	 \label{algCt:cancExh}  
%	$legit(ctp)=(ctp.cct = \bot \rangle )$\\
%	$staleCntrInfo() = staleInfo() \lor (\exists p_j \in P, x\in storedCnts[j]:exhausted(x)\land legit(x))$\\
%	$retCntrQ(ct) :$ {\bf return} $(storedCnts[ct.lbl.lCreator])$\\
%	$legitCnts()$ $=$ $\{maxC[j].mct: \exists p_j \in P \land legit(maxC[j])\} $\\
%	$useOwnCntr() =\mathbf{if~}(\exists cp \in storedCnts[i] : legit(lp))$ $\mathbf{then~}maxC[i]$ $\gets$ $cp$ $\mathbf{else~}storedCnts[i].add(maxC[i]$ $\gets$ $\langle \langle nextLabel(), 0, i\rangle, \bot \rangle)$~\label{cnt:useOwnLabelDef}
%	\tcp{For every $cp \in storedCnts[i]$, we pass to $nextLabel()$ both $cp.mct.lbl$ and $cp.cct.lbl$.} 
	
\tcp{Lines~\ref{algCt:transmit} and~\ref{algCt:receive} run in the background.}
\vspace{.2em}
{\bf upon} $transmitReady(p_j \in P \setminus \lbrace p_i \rbrace)$ {\bf do transmit}($\langle maxC[i], maxC[j]\rangle$)\label{algCt:transmit}\; 
\vspace{.1em}
{\bf upon} $receive(\langle sentMax, lastSent \rangle)$ {\bf from} $p_k$ \textbf{do} $processCntr(sentMax, lastSent, j)$ \label{algCt:receive} %$maintainCtr(sentMax, lastSent, j)$;
%\Begin{
%%$\process(sentMax, lastSent, j)$; \label{algCt:receive} 
%\lForEach{$p_j\in P, ctp \in storedCnts[j]: legit(ctp)\land exhausted(ctp)$}{$cancelExh(ctp)$}\label{algCt:exhstoredCnts}
%\lIf{$(\exists ctp' \in \langle sentMax, lastSent\rangle: exhausted(ctp'))$}{$cancelExh(ctp')$}\label{algCt:exhInputCnts}
%%	\lIf{$\neg legit(lastSent)$ $\land$ $maxC[i] =_{mct.lbl} lastSent$}{$maxC[i] \gets lastSent$\label{ln:lastSentCancel}}
%%	$maxC[j]$ $\gets$ $sentMax$\; \label{cnt:exposeStore}
%%	\lForEach{$cp \in \{sentMax, lastSent\}$}{$placeCntr(cp)$}
%	$cancelExhMaxC()$;\label{algCt:exhMaxC}
%	$\process(sentMax, lastSent)$\;
%	}\label{algCt:diffuseEnd}
%\BlankLine

\vspace{.4em}
\textbf{procedure} $processCntr($\textsl{counter\ pair} $sentMax$,  \textsl{counter pair} $lastSent)$,  \textsl{int} $k$) \label{cnt:processCntr}
\Begin{
%$\process(sentMax, lastSent, j)$; \label{algCt:diffuse} 
\lForEach{$p_j\in P, ctp \in storedCnts[j]: legit(ctp)\land exhausted(ctp)$}{$cancelExh(ctp)$}\label{algCt:exhstoredCnts}
\lIf{$(\exists ctp' \in \langle sentMax, lastSent\rangle: exhausted(ctp'))$}{$cancelExh(ctp')$}\label{algCt:exhInputCnts}
%	\lIf{$\neg legit(lastSent)$ $\land$ $maxC[i] =_{mct.lbl} lastSent$}{$maxC[i] \gets lastSent$\label{ln:lastSentCancel}}
%	$maxC[j]$ $\gets$ $sentMax$\; \label{cnt:exposeStore}
%	\lForEach{$cp \in \{sentMax, lastSent\}$}{$placeCntr(cp)$}
	$cancelExhMaxC()$;\label{algCt:exhMaxC}
	$\process(sentMax, lastSent)$\;
	}\label{algCt:receiveEnd}
%\textbf{procedure} $placeCntr($\textsl{counter pair} $ctp,$ \textsl{integer} $j)$ \Begin{
%	\lIf{$exhausted(ctp')$}{$cancelExhausted(ctp')$}\label{algCt:exhInputCnts}
%	\lIf{$\neg legit(lastSent)$ $\land$ $maxC[i] =_{mct.lbl} lastSent$}{$maxC[i] \gets lastSent$\label{ln:lastSentCancel}}
%	$maxC[j]$ $\gets$ $sentMax$\; \label{cnt:exposeStore}
%} 

\vspace{0.4em}
\textbf{operator} $\process($\textsl{counter\ pair} $sentMax$,  \textsl{counter pair} $lastSent)$,  \textsl{int} $k$) \label{cnt:process}

\Begin{
%\lForEach{$ctp\in\{sentMax, lastSent\}$}
\lIf{$sentMax\neq {\sf NULL}$}{$maxC[k]$ $\gets$ $sentMax$\label{cnt:exposeStore}}
%{\textbf{if} {$exhausted(ctp)$} \textbf{then} {$cancelExh(ctp)$}\label{algCt:exhInputCnts}}
\lIf{$lastSent \neq {\sf NULL} \land \neg legit(lastSent)$ $\land$ $maxC[i] =_{mct.lbl} lastSent$}{$maxC[i].add(lastSent)$\label{cnt:lastSentCancel}}
%$cancelExhaustedMaxC()$;\label{algCt:exhMaxC}\\
\lIf{$staleCntrInfo()$}{$storedCnts.emptyAllQueues()$} \label{cnt:clean}
\lForEach{$p_j \in P : recordDoesntExist(j)$}{$retCntrQ(maxC[j]).add(maxC[j])$} \label{cnt:add}

\ForEach{$p_j \in P, cp \in storedCnts[j] : (legit(cp) \land (notgeq(j,cp)\neq \bot))$}{$cp.cct \gets notgeq(j,cp)$} \label{cnt:cancelLabels}

\lForEach{$p_j \in P: ((\neg legit(maxC[j])\lor (cp <_{mct.seqn} maxC[j])) \land (maxC[j] =_{ml}$ $cp)$ $\land$ $legit(cp)$ \textbf{\em where} $ cp \in retCntrQ(maxC[j]) $}{$cp \gets maxC[j]$} \label{cnt:receivedCanceled}

\lForEach{$p_j \in P, cp \in storedCnts[j] : double(j, cp)$}{$cp.remove()$} \label{cnt:remove}

\ForEach{$p_j \in P : (legit(maxC[j]) \land (canceled(maxC[j])\neq \langle \bot, \bot \rangle))$}{$maxC[j] \gets canceled(maxC[j])$} \label{cnt:cancelMax}

\lForEach{$p_j \in P, cp \in )$}{$maxC[j] \gets getMaxCnt(maxC[j])$} \label{cnt:retMaxCntUsed}

%\lForEach{$p_j \in P: (legit(maxC[j])\lor (cp <_{mct.seqn} maxC[j])) \land (maxC[j] =_{ml}$ $cp)$ $\land$ $legit(cp)$ \textbf{\em where} $ cp \in retCntrQ(maxC[j]) $}{$cp \gets maxC[j]$} \label{cnt:maxSeqn2MaxC}
\lIf{$legitCnts() \neq \emptyset$}{$maxC[i] \gets \langle \max_{\prec_{ct}}(legitCnts()), \bot \rangle$} \label{cnt:adopt}
%\langle \max_{\prec_{lb}}(legitLabels()), \bot \rangle$} \label{cnt:adopt}

\lElse{$useOwnCntr()$}  \label{cnt:useOwnLabel}
}

\end{\macroFontSize}
%\end{\algorithmFontSize}

\end{algorithm}\vspace{-0.4em}
\caption{The $maintainCntrs()$ operator (code for $p_i$).}

\label{alg:cntrMaintenance}
\end{figure}
%-----------------------------------------------------------------

%\subsubsection{Counter Maintenance Algorithm}

\noindent \textbf{The enhanced labeling algorithm.} Figure~\ref{alg:cntrMaintenance} presents a standalone version of the labeling algorithm adjusted for counters.
Each processor $p_i$ uses the token-based communication to transmit to every other processor $p_j$ its own maximal counter and the one it currently holds for $p_j$ in $maxC_i[j]$ (line~\ref{algCt:transmit}).
Upon receipt of such an update from $p_j$, $p_i$ first performs canceling of any exhausted counters in $storedCnts[]$ (line~\ref{algCt:exhstoredCnts}), in $maxC[]$ (line~\ref{algCt:exhMaxC}) and in the received couple of counter pairs (line~\ref{algCt:exhInputCnts}).
Having catered for exhaustion, it then calls $\process(\langle \bullet, \bullet\rangle)$ with the received two counter pairs as arguments.
This is essentially a counter version of Algorithm~\ref{alg:WFR}.
Macros that require some minor adjustments to handle counters are seen in Figure~\ref{fig:ctrMacros} lines~\ref{cnt:legit} to~\ref{cnt:useOwnLabelDef}.
We also address the need to update counters of $maxC[]$ w.r.t. $seqn$ and $wid$ based on counters from the $storedCnts[]$ structure and vice versa in lines~\ref{cnt:retMaxCntUsed} and~\ref{cnt:lastSentCancel}.

%The algorithm uses the enhanced counter structures $maxC[n]$ and $storedCnts[n]$ which are maintained in the same way as in the labeling algorithm with some additional operations.
We define the operator $add(ctp)$ (Fig.~\ref{fig:ctrVars}) to enqueue a counter pair $ctp$ to a queue of $storedCnts[n]$, where in case a counter with the same label already exists  the following two rules apply:
%if a corresponding counter with the same $lbl$ doesn't exist, or to keep only one of the two instances if it exists. 
%There are two enqueuing rules: 
(1) if at least one of the two counters is cancelled we keep a canceled instance, and (2) if both  counters are legitimate, we keep the greatest counter with respect to $\langle seqn, wid\rangle$.
The counter is placed at the front of the queue. 

% % % % % % % END OF REMOVE % % % % % % % % % % %

\subsubsection{Counter Increment Algorithm}
Algorithm~\ref{alg:cntrIncrement} shows a self-stabilizing counter increment algorithm where multiple processors can increment the counter. 
We start with some useful definitions and proceed to describe the algorithm.
\begin{algorithm*}[t]
\caption{Increment Counter; code for $p_i$}

\label{alg:cntrIncrement}
\begin{\algorithmFontSize}
%\begin{\algorithmFontSize}
\textbf{interface function} $incrementCounter()$ \Begin{
\textbf{let} $\langle responseSet, status\rangle \gets  \langle \emptyset, {\sf MAX\_REQUEST}\rangle;$\label{algCt:maxreceive_{j'}eq} \\
\Repeat{$status={\sf COMPLETE}$}{
\If{$status = {\sf MAX\_REQUEST} \land (\exists Q \in \mathbb{Q}: Q \subseteq \{responseSet\})$\label{algCt:maxreceive_{j'}eqACKED}}{$initWrite()$;
$increment()$\label{algCt:max_write}}
\ElseIf{$status = {\sf MAX\_WRITE} \land (\exists Q \in \mathbb{Q}: Q \subseteq \{responseSet\})$\label{algCt:max_writeACKED}}{$\langle status \gets \sf COMPLETE \rangle$\label{algCt:complete}}
\lForEach{$p_j \in P$}{\textbf{send} $\langle status, maxC[i], maxC[j]\rangle$\label{algCt:send}}
}
\textbf{return} $maxC[i]$
}

{\bf upon receive of} $m=\langle subj, sentMax, lastSent \rangle$ {\bf from} $p_j$ \Begin{
\lIf{$(m.subj = {\sf MAX\_REQUEST})$}{%$maxC_i[i] \gets maxNonExhaustCntr()$; 
\textbf{send} $\langle {\sf ACK}, maxC_i[i], maxC_i[j] \rangle$ \textbf{to} $p_j$\label{algCt:sendMeYourMax}}
\ElseIf{$(m.subj = {\sf MAX\_WRITE})$}{
%	$max_{ct}(\{m.cp.mct, maxC_i[j].mct\})$\;
%	\lIf{$m.cp.mct=_{lbl.lCreator}i$}{%
%		$storedCnts_i[i].enqueue(maxC_i[i])$}\label{algCt:qWriteKeepMaxCnt} 
%	\lIf{$exhausted(maxC_i[j])$}{$cancelExhausted(maxC_i[j])$}\label{algCt:qWriteExhTest} 
%	{\bf send $ACK$ to} $p_j$\label{algCt:qWriteAck}\;
	$processCntr(sentMax, lastSent, j)$;
	{\bf send $\langle \sf ACK, \bot, lastSent \rangle$ to} $p_j$\label{algCt:qWriteAck};
}
\ElseIf{$(m.subj = {\sf ACK} \land correctResponse(sentMax,lastSent))$\label{algCt:correctACK}}{$processCntr(sentMax, lastSent, j);$ $responseSet\gets j$}
} 

%\textbf{procedure} $maxNonExhaustCntr()$ \label{algCt:maxnonExhaustCntr}\Begin{
%	\Repeat{$legit(maxC[i]) \land \neg exhausted(maxC[i])$}{
%%		$cancelExhaustedMaxC()$;\label{algCt:cancelExhFindMax}  
%		$processCntr(sentMax, lastSent, j)$;\\
%		%$\process({\sf NULL}, {\sf NULL})$;\label{algCt:tidyQ}
%		$maxC[i] \gets getMaxSeq()$;\label{algCt:assignNewMaxCntr}
%	}
%}

\end{\algorithmFontSize}
%\end{\algorithmFontSize}
\end{algorithm*}

\Paragraph{Quorums}
We define a \emph{quorum set} $\mathbb{Q}$ based on processors in $P$, as a set of processor subsets of $P$ (named \emph{quorums}), that ensure a non-empty intersection of every pair of quorums.
Namely, for all quorum pairs $Q_{i}, Q_j \in \mathbb{Q}$ such that $Q_i,Q_j \subset P$, it must hold that $Q_{i} \cap Q_j \neq \emptyset$.
%We define a \emph{quorum} as a subset of $P$, such that every  pair of quorums $Q_i, Q_j$ in the set of quorums $\mathbb{Q}$ of $P$ satisfies $Q_i \cap Q_j \neq \emptyset$. 
%[[@@ Quorums: set of quorums $\mathbb{Q}$ is not previously defined + 1st two sentences need rewriting. @@ Io. Is this better? @@]]
%I.e., every quorum intersects every other quorum of the system, and there is at least one processor in the intersection.
This \emph{intersection property} is useful to propagate information among servers and exploiting the common intersection without having to write a value $v$ to all the servers in a system, but only to a single quorum, say $Q$.
If one wants to retrieve this value, then a call to \emph{any} of the quorums (not necessarily $Q$), is expected to return $v$ because there is least one processor in every  quorum that also belongs to $Q$. 
In the counter algorithm we exploit the intersection property to retrieve the currently greatest counter in the system, increment it, and write it back to the system, i.e., to a quorum therein.
Note that majorities form a special case of a quorum system.

%The counter increment algorithm uses the same structures and procedures as the labeling algorithm, but now with counters instead of labels. 
\Paragraph{Algorithm description} To increment the counter, a processor $p_i$ enters status $\sf MAX\_REQUEST$ (line~\ref{algCt:maxreceive_{j'}eq}) and starts sending a request to all other processors, waiting for their maximal counter (via line~\ref{algCt:send}). 
Processors receiving this request respond with their current maximal counter and the last sent by $p_i$ (line~\ref{algCt:sendMeYourMax}).
When such a response is received (line~\ref{algCt:correctACK}), $p_i$ adds this to the local counter structures via the counter bookkeeping algorithm of Figure~\ref{alg:cntrMaintenance}.
Once a quorum of responses (line~\ref{algCt:maxreceive_{j'}eqACKED}) have been processed, $maxC[i]$ holds the maximal counter that has come to the knowledge of $p_i$ about the system's maximal counter. 
%Using a similar procedure as the labeling algorithm it (eventually) finds the maximal epoch label and the maximal sequence number it knows for this label. In other words, it collects counters and finds the counter(s) with the largest global label; there can be more than one such counter, in which case it returns the one with the highest sequence number, breaking symmetry with the sequence number processor identifiers. 
%Having received answers from a quorum it finds the maximal counter with respect to $\prec_{ct}$ (using procedure $maxNonExhaustCntr()$ -- see Fig.~\ref{fig:ctrMacros} line~\ref{algCt:maxnonExhaustCntr}).
%It verifies whether this counter is exhausted and if not it increments it ((Alg.~\ref{alg:cntrIncrement} -- line~\ref{algCt:maxreceive_{j'}eq}).
This counter is then incremented locally and $p_i$ enters status $\sf MAX\_WRITE$ by initiating the propagation of the incremented counter (line~\ref{algCt:max_write}), and waiting to gather acknowledgments from a quorum (the condition of line~\ref{algCt:max_writeACKED}).
When the latter condition is satisfied, the function returns the new counter.
This is, in spirit, similar to the two-phase write operation of MWMR register implementations, focusing on the sequence number rather than on an associated value.

\subsubsection{Proof of correctness}
%\paragraph{Proof of correctness}
\paragraph{Proof outline} Initially we prove, by extending the proof of the labeling algorithm, that starting from an arbitrary configuration the system eventually reaches to a global maximal label (as given in Theorem~\ref{th:oneGreatest4All}), even in the presence of exhausted counters (Lemma~\ref{th:incrCntrMaxLabel}). 
By using the intersection property of quorums we establish that a counter that was written is known by at least one processor in every quorum (Lemma~\ref{th:quorumKnowsMax}. 
We then combine the two previous lemmas to prove that counters increment monotonically.

\begin{lemma}
\label{th:incrCntrMaxLabel}
%
%Consider two processors $p_i$ taking a practically infinite number of steps and a setting as described by Theorem~\ref{th:oneGreatest4All}, adjusted for labels rather than counters as described above. 
In a bounded number of steps of Algorithm~\ref{alg:cntrMaintenance} %guarantees that, within a bounded number of steps, 
every processor $p_i$ has counter $maxC_i[i]=ct$ with $ct.lbl = \ell_{max}$ the globally maximal non-exhausted label.
%Moreover, $\forall Q \in \mathbb{Q}, \exists p_j \in Q: maxC_j[j] = ct \land (ct' \prec_{ct} ct)$, where  $ct' \in \{storedCnts_k[k']\cup maxC_k[k']: ct' =_{lbl} ct\}_{p_k,p_{k'} \in P}\setminus \{ct\} $, i.e., $ct'$ is every counter in the system with identical label but less than $ct$ w.r.t. $seqn$ or $wid$.
\end{lemma}

\begin{proof}
For this lemma we refer to the enhanced labeling algorithm for counters (Figure~\ref{alg:cntrMaintenance}).
%[[I propose to substitute the proof that follows with this version]]\\
The lemma proof can be mapped on the arguments proving lemmas \Argument~\ref{th:noStaleInfo} to Lemma~\ref{th:boundedCreatorSteps} of Algorithm~\ref{alg:WFR}. 
Specifically, consider a processor $p_i$ that has performed a full execution of $processCntr()$ (Fig.~\ref{alg:cntrMaintenance} line~\ref{cnt:processCntr}) at least once due to a receive event. %a receive event (or has performed a counter increment), i.e., has executed all lines~\ref{algCt:receive} to \ref{algCt:receiveEnd} at least once. %(or has called $maxNonExhaustCntr()$).
This implies a call to $maintainCntrs$ and thus to $staleCntrInfo()$ (Fig.~\ref{alg:cntrMaintenance} line~\ref{cnt:clean}) which will empty all queues if exhausted non-canceled counters exist.
Also there is a call to $cancelExhaustedMaxC()$ which cancels all counters that are exhausted in $maxC[]$.
By observation of the code, after a single iteration, there is no local exhausted counter that is not canceled.

Since every counter that is received and is exhausted becomes canceled, and since the arbitrary counters in transit are bounded, we know that there is no differentiation between exhausted labels that may cause a counter's label to be canceled.
Namely, the size of the queues of $storedCnts[]$ remain the same while at the same time provide the guarantees provided by the proof of the labeling algorithm.
It follows from the labeling algorithm correctness and by our cancellation policy on the exhausted counters, that Theorem~\ref{th:oneGreatest4All} can be extended to also include the use of counters without any need to locally keep more counters than there are labels.

We proceed to deduce that, eventually, any processor taking practically infinite number of steps in $R$ obtains a counter with globally maximal label $\ell_{max}$.

\end{proof}

\noindent For the rest of the proof we refer to line numbers in Algorithm~\ref{alg:cntrIncrement}.

\begin{lemma}
\label{th:quorumKnowsMax}
%
%Consider two processors $p_i$ taking a practically infinite number of steps and a setting as described by Theorem~\ref{th:oneGreatest4All}, adjusted for labels rather than counters as described above. 
In an execution where Lemma~\ref{th:incrCntrMaxLabel} holds, it also holds that $\forall Q \in \mathbb{Q}, \exists p_j \in Q: maxC_j[j] = ct \land (ct' \prec_{ct} ct)$, where  $ct' \in \{storedCnts_k[k']\cup maxC_k[k']: ct' =_{lbl} ct\}_{p_k,p_{k'} \in P}\setminus \{ct\} $, i.e., $ct'$ is every counter in the system with identical label but less than $ct$ w.r.t. $seqn$ or $wid$ and $ct$ is the last counter increment.
\end{lemma}

\begin{proof}
Observe that upon a quorum write, the new incremented counter $ct$ with the maximal label $lb$ is propagated (lines~\ref{algCt:max_writeACKED} and~\ref{algCt:send}) until a quorum of acknowledgments have been received.
Upon receiving such a counter by $p_i$, a processor $p_j$ will first add $ct$ to its structures via $processCntr()$ and will then acknowledge the write.
If this is the maximal counter that it has received (there could be concurrent ones) then the call to $processCntr()$ will also have the following effects: (i) the counter's $seqn$ and $wid$ will be updated in the $storedCnts_j[]$ structure in the queue of the creator of $lb$, (ii) $maxC_j[j] \gets ct$.

Since $p_i$ waits for responses by a quorum before it returns, it follows that by the intersection property of the quorums, the lemma must hold when $p_i$ reaches status $\sf COMPLETE$. 
\end{proof}

\begin{theorem}
\label{th:countersIncrMonotonic}
Given an execution $R$ of the counter increment algorithm in which at least a majority of processors take a practically infinite number of steps, the algorithm ensures  that counters eventually increment monotonically.
\end{theorem}

\begin{proof}
Consider a configuration $c \in R'$ where $R'$ is a suffix of $R$ in which Lemma~\ref{th:incrCntrMaxLabel} holds, and in which $ct_{max}$ is the counter which is maximal with respect to $\prec_{ct}$.
There are two cases that the counter may be incremented.

%In  a legal execution the counter is only incremented by a call to $incrementCounter()$, in which case Alg.~\ref{alg:cntrIncrement} line~\ref{algCt:maxreceive_{j'}eqACKED} and Lemma~\ref{th:quorumKnowsMax} ensures that the maximal counter value was received (from a quorum) and line~\ref{algCt:max_writeACKED} ensures that the incremented counter is written to a quorum. 
%By the intersection property of the quorums, any call to $incrementCounter()$ will increment the globally maximal counter monotonically. 
In the first case, a legal execution, the counter $ct_{max}$ is only incremented by a call to $incrementCounter()$,
By Lemma~\ref{th:quorumKnowsMax} any call to $incrementCounter()$ will return the last written maximal counter (namely  $ct_{max}$). 
When this is incremented giving $ct'_{max}$ then $ct'_{max}.seqn = ct_{max}.seqn+1$ which is monotonically greater than $ct'_{max}$ and in case of concurrent writes the $wid$ is unique and can break symmetry enforcing the monotonicity.

The second case arises when $ct_{max}$ comes from the arbitrary initial state, is not known by a quorum, and resides in either a local state or is in transit.
When $ct_{max}$ eventually reaches a processor, it becomes the local maximal and it is propagated either via counter maintenance or in the first stage of a counter increment when the maximal counters are requested by the writer.
In this case the use of $ct_{max}$ is also a monotonic increment, and from this point onwards any increment in $R'$ proceeds monotonically from  $ct_{max}$, as described in the previous paragraph.
\end{proof}

\subsubsection{Algorithm Complexity}
The local memory of a processor implementing the counter increment is not different in order to the labeling algorithm's, since converting to the counter structures only adds an integer (the sequence number).
Hence the \emph{space complexity} of the algorithm is $\bigO(n^3)$ in counters.
The upper bound on \emph{stabilization time} in the number of counter increments that are required to reach a period of practically infinite counter increments can be deduced by Theorem~\ref{th:countersIncrMonotonic}. For some $t$ such that $0\leq t\leq 2^\tau$ in an execution with $\bigO(n\cdot\beta\cdot t)$ counter increments (recall that $\beta=n^3cap+2n^2-2n$), there is a practically infinite period of ($2^\tau$) monotonically increasing counter increments in which the label does not change.

%\vspace{.4em}

%\Paragraph{\bl{MWMR Register Emulation}}
\subsubsection{MWMR Register Emulation}
Having a practically-self-stabilizing counter increment algorithm, it is not hard to implement a \emph{practically-self-stabilizing MWMR register emulation}.
Each counter is associated with a value and the counter increment procedure essentially becomes a write operation: once the maximal counter is found, 
it is increased and associated with the new value to be written, which is then communicated to a majority of processors. The read operation 
is similar: a processor first queries all processors about the maximum counter they are aware of. It collects responses from a majority
and if there is no maximal counter, it returns $\bot$ so the processor needs to attempt to read again (i.e., the system hasn't converged to a
maximal label yet). If a maximal counter exists,
it sends this together with the associated value to all the processors, and once it collects a majority of responses, it returns the counter with the associated value (the second phase is a %standard requirement for preserving 
required to preserve the consistency of the register (c.f.~\cite{ABD,RAMBO}).

\section{Virtually Synchronous Stabilizing Replicated State Machine} %\vspace{-.6em}
\label{sec:VS} 
% 
% % % EDITED BY YIANNIS % % % %
%We now present a self-stabilizing reliable multicast algorithm that provides fault-tolerance, with respect to processor crashes and communication asynchrony, by considering the {\em (current group) view} of the changing processor set at the end of the group communication endpoint. We propose a self-stabilizing algorithm that guarantees the VS property. Namely, any two processors that are members of the same view, ought to deliver identical message sets to their SMRs as long as they continue to share the same view, which indeed may change. 
%This way, SMR algorithms can use the multicast service to synchronize their state transitions, i.e., the group members multicast their current automaton state, and the last received input that had led to that state.  
\bl{Group communication systems (GCSs) that guarantee the virtual synchrony property, essentially suggest that processes that remain together in consecutive groups (called \emph{views}) will deliver the same messages in the desired order~\cite{DBLP:conf/replication/Birman10}. 
This is particularly suited to maintain a replicated state machine service, where replicas need to remain consistent, by applying the same changes suggested by the environment's requests.}
\trnsfr{
A key advantage of multicast services (with virtual synchrony) is the ability to reuse the same view during many multicast rounds, and which allows every automaton step to require just a single multicast round, as compared to other more expensive solutions. 
%The aim of the proposed algorithm is to demonstrate, in a self-stabilizing manner, the most important ways to cut down the number of times in which the service needs to agree on a new view, and when it does, to perform it swiftly. 
%Similar to~\cite{birmanbook}, we assume that the service works in the network's primary partition (see Definition~\ref{def:Primary}) and require that a majority of processors are present in every view set. 
%We do not however require all (local) failure detectors to agree on the set of recently alive and connected processors.
}

\bl{
GCSs provide the VS property by implementing two main services: a reliable multicast service, and a membership service to provide the membership set of the view, whilst they also assume access to unbounded counters to use as unique view identifiers.
We combine existing self-stabilizing versions of the two services (with adaptations where needed), and we use the counter from the previous section to build the first (to our knowledge) practically self-stabilizing virtually synchronous state machine replication.
While the ideas appear simple, combining the services is not always intuitive, so we first proceed to a high-level description of the algorithm and the services, and then follow the algorithm with a more technical description and the correctness proof.
}

%\paragraph{Overview} 
\subsection{Preliminaries}
%As already mentioned, group communication systems providing the VS property implement two main services: 
%a membership service and a reliable multicast service, whilst they assume there is access to an unbounded counter to use as unique view identifiers.
%We provide these services in a coordinator-based solution, considering a {\em primary-group} implementation~\cite{birmanbook}. 
\bl{
The algorithm progresses in state replication by performing multicast rounds, when a \emph{view}, a tuple composed of a members $set$ taken from $P$, and of a unique identifier ($ID$) that is a counter as defined in the previous section, is installed. 
This view must include a \emph{primary component} (defined formally in Definition~\ref{def:Primary}), namely it must contain a majority of the processors in $P$, i.e., $n/2+1$.
In our version, a processor, the \emph{coordinator}, is responsible: (1) to progress the multicast service which we detail later, (2) to change the view when its failure detector suggests changes to the composition of the view membership.
Therefore, the output of the coordinator's failure detector defines the set of view members; this helps to maintain a consistent membership among the group members, despite inaccuracies between the various failure detectors. 
}

\bl{
On the other hand, the counter increment algorithm that runs in the background allows the coordinator to draw a counter for use as a view identifier
%To assign view identifiers, we use our counter increment algorithm. 
%Specifically, a counter defines a view identifier, and the counter's writer identifier is that of the view's coordinator.  
and in this case, the counter's writer identifier ($wid$) is that of the view's coordinator.
This defines a simple interface with the counter algorithm, which provides an identical output. 
%We show that this does not break the VS property, as long as the majority-based failure detector property  is preserved.
Pairing the coordinator's member set with a counter as view identity we obtain a $view$.
%The coordinator is also responsible for the consistency of the multicast mechanism within the group.
Of course as we will describe later, reaching to a unique coordinator may require issuing several such view proposals, of which one will prevail.
We first suggest a possible implementation of a failure detector (to provide membership) and of a reliable multicast service over the self-stabilizing FIFO data link given in Section~\ref{s:sys}, and then proceed to an algorithm overview.
}

\begin{definition}
\label{def:Primary}
We say that the output of the (local) failure detectors in execution $R$ includes a {\em primary partition} when it includes a supporting majority of processors $P_{\majority} : P_{\majority} \subseteq P$, that (mutually) never suspect at least one processor, i.e., $\exists p_\ell \in P$ for which $|P_{\majority}|$ $>$ $\lfloor n/2\rfloor$ and $(p_i \in (P_{\majority}$ $\cap$ $FD_\ell))$ $\Longleftrightarrow$ $(p_\ell \in (P_{\majority}$ $\cap$ $FD_i))$ in every $c \in R$, where $FD_x$ returns the set of processors that according to $p_x$'s failure detector are active. 
\end{definition}

\Paragraph{\bl{Failure detector}}
\bl{We employ the self-stabilizing failure detector of~\cite{DBLP:conf/netys/BlanchardDBD14} which is implemented as follows. 
Every processor $p$ uses the token-based mechanism to implement a heartbeat (see Section~\ref{s:sys}) with every other processor, and maintain a heartbeat integer counter for every other processor $q$ in the system. Whenever processor $p$ receives the token from processor $q$ over their data link, processor $p$ resets the counter's value to zero and increments all the integer counters associated with the other processors by one, up to a predefined threshold value $W$.
Once the heartbeat counter value of a processor $q$ reaches $W\!,$~the failure detector of processor $p$ considers $q$ as inactive. In other words, the failure detector at processor $p$ considers processor $q$ to be active, if and only if the heartbeat associated with $q$ is strictly less than $W.$ 
}

\bl{
As an example, consider a processor $p$ which holds an array of heartbeat counters for processors $p_i, p_j, p_k$ such that their corresponding values are $\langle 2, 5, W-1\rangle$.
If $p_j$ sends its heartbeat, then $p$'s array will be changed to $\langle 3, 0, W\rangle$.
In this case, $p_k$ will be suspected as crashed, and the failure detector reading will return the set $p_i, p_j$ as the set of processors considered correct by $p$.
}

\bl{
Note that our virtual synchrony algorithm, employs the same implementation but has weaker requirements than~\cite{DBLP:conf/netys/BlanchardDBD14} that solve consensus, and thus they resort to a failure detector at least as strong as $\Omega$~\cite{DBLP:journals/jacm/ChandraHT96}. 
Specifically, in Definition~\ref{def:Primary} we pose the assumption that \emph{just a majority of the processors} do not suspect at least one processor of $P$ for sufficiently long time, in order to be able to obtain a long-lived coordinator. 
This is different, as we said before, to an eventually perfect failure detector that ensures that after a certain time, no active processor suspects any other active processor.} 

\bl{
Our requirements, on the other hand, are stronger than the weakest failure detector required to implement atomic registers (when more than a majority of failures are assumed), namely the $\Sigma$ failure detector~\cite{DBLP:conf/podc/Delporte-GalletFGHKT04}, since virtual synchrony is a more difficult task.
In particular, whilst the $\Sigma$ failure detector eventually outputs a set of only correct processors to correct processors, we require that this set in at least half of the processors, will contain at least one common processor.
In this perspective our failure detector seems to implement a \emph{self-stabilizing} version of a slightly stronger failure detector than $\Sigma$. 
It would certainly be of interest for someone to study what is the weakest failure detector required to achieve practically-self-stabilizing virtually synchronous state replication, and whether this coincides with our suggestion. 
}

%\trnsfr{
%Multicast services that provide %virtual synchrony~\cite{Birman91Lightweight} 
%VS
%often leverage on the system's ability to preserve (when possible) the coordinator during view transitions rather than electing a new coordinator. %upon view transition. 
%The motivation here is that the coordinator has the most recent automaton state and holds a copy of the set of {\em unstable} messages, which are the ones that were delivered to at least one view member, but the (alive and connected) view members have yet to receive a delivery acknowledgement for these. 
%Our solution naturally follows this approach since it often helps the service to abstain from electing a leader upon every view change, as well as to avoid view transitions that require the coordinator to first investigate about all unstable messages (and the most recent automaton state) among all view members that continue to the next view. 
%This is done so that the service can provide the virtual synchrony property. Thus, we consider the notion of coordinators that a majority of processors never suspects and we show that, in the existence of such processors, one of these coordinator will be eventually used in all subsequent views (Definition~\ref{def:Primary}).
%%As explained in Section~\ref{sec:nutshell}, the algorithm, uses the counter increment algorithm, as well as a reliable multicast and a failure detector built over a self-stabilizing FIFO data link.
%}

\Paragraph{\trnsfr{Reliable multicast implementation}} 
\trnsfr{ %As we will see next, we use a coordinator \bl{(whose choice we explain later,)} 
The coordinator of the view controls the %to exchange messages 
exchange of messages (by multicasting) within  the view. 
The coordinator requests, collects and combines input from the group members, and then it multicasts the updated information. 
Specifically, when the coordinator decides to collect inputs, it waits for the token (see Section~\ref{s:sys}) to arrive from each group participant. 
Whenever a token arrives from a participant, the coordinator uses the token to send the request for input to that participant, and waits the token to return with some input (possibly $\bot$, when the participant does not have a new input). 
%Of course, the tokens continue to move back and forth, but once the coordinator receives an input from a certain participant with respect to this multicast invocation, the corresponding token will not carry any new request to receive input from the same participant. 
Once the coordinator receives an input from a certain participant with respect to this multicast invocation, the corresponding token will not carry any new requests to receive input from the same participant; of course, the tokens continue to move back and forth to sustain the heartbeat-based failure detector.
When all inputs are received, the processor combines them and again uses the token to carry the updated information. Once this is done, the coordinator can proceed to the next input collection. %, when needed. %\vspace{.3em} 
}

\bl{We provide the pseudocode for the practically-stabilizing replicate state machine implementation as Algorithm~\ref{alg:multVirSyn}, a high level description, and proceed to a line-by-line description and correctness analysis.}

\begin{algorithm*}[p!]

   \caption{A \bl{practically-}self-stabilizing automaton replication using virtual synchrony, %using the ASPP approach (Definition~\ref{def:Primary}), 
   code for processor $p_i$}
\label{alg:multVirSyn}
\begin{scriptsize}

\noindent {\bf Constants:} \nllabel{lnVS:const}
$PCE$ (periodic consistency enforcement) number of rounds between global state check\;

\noindent {\bf Interfaces:}
$fetch()$ next multicast message, $apply(state, msg)$ applies the step $msg$ to $state$ (while producing side effects), $synchState(replica)$ returns a replica consolidated state, $synchMsgs(replica)$ returns a consolidated array of last delivered messages, $failureDetector()$ returns a vector of processor pairs $\langle pid, crdID\rangle$, $inc()$ returns a counter from the increment counter algorithm\; \nllabel{ln:inter}

\noindent {\bf Variables:} \nllabel{ln:var}
$rep[n]=\langle view$ $=$ $\langle ID$, $set \rangle$, $status$ $\in$ $\{{\sf Multicast}$, ${\sf Propose}$, ${\sf Install}\}$, $(multicast$ $round$ $number)$ $rnd$, $(replica)$ $state$, $(last$ $delivered$ $messages)$ $msg[n]$ $(to$ $the$ $state$ $machine)$, $(last$ $fetched)$ $input$ $(to$ $the$ $state$ $machine)$, $propV$ $=$ $\langle ID$, $set \rangle$, $(no$ $coordinator$ $alive)$ $noCrd$, $(recently$ $live$ $and$ $connected$ $component)$ $FD \rangle$ : an array of state replica of the state machine, where $rep[i]$ refers to the one that processor $p_i$ maintains.
A local variable $FDin$ stores the $failureDetector()$ output. $FD$ is an alias for $\{FDin.pid \}$, i.e. the set of processors that the failure detector considers as active. Let $crd(j)=\{FDin.crdID: FDin.pid=j\}$, i.e. the id of $p_j$'s local coordinator, or $\bot$ if none.

\noindent {\bf Do forever} \Begin{

{\bf let} $FDin = failureDetector()$\; \nllabel{ln:FD}

{\bf let} $\seemCrd$ $=$ $\{ p_\ell$ $=$ $rep[\ell].propV.ID.wid$ $\in$ $FD$ $:$ $(|rep[\ell].propV.set|$  $>$ $\lfloor n/2\rfloor)$ $\land$ $(|rep[\ell].FD|$ $>$ $\lfloor n/2\rfloor)$ $\land$ $( p_\ell$ $\in$ $rep[\ell].propV.set)$ $\land$ $(p_k$ $\in$ $rep[\ell].propV.set$ $\leftrightarrow$ $p_\ell$ $\in$ $rep[k].FD )$ $\land$ $((rep[\ell].status$ $=$ ${\sf Multicast})$ $\rightarrow$ $(rep[\ell].(view$ $=$ $propV) \land crd(\ell)=\ell)) \land ((rep[\ell].status = {\sf Install}) \rightarrow crd(\ell)=\ell)\}$\; \nllabel{ln:seemCrd}

{\bf let} $\valCrd$ $=$  $\{ p_\ell$ $\in$ $seemC‏rd$ $:$ $(\forall p_k$ $\in$ $seemC‏rd$ $:$ $rep[k].propV.ID$ $\preceq_{ct}$ $rep[\ell].propV.ID) \}$\; \nllabel{ln:valCrd}

$noCrd$ $\gets$ $(|\valCrd|$ $\neq$ $1)$; $crdID \gets \valCrd$\; \nllabel{ln:noCrd}

\lIf{$(|FD|>\lfloor n/2\rfloor)$ 
$\land$ $(((|\valCrd|$ $\neq$ $1)$ $\land$ $(|\{ p_k \in FD $ $:$ $p_i$ $\in$ $rep[k].FD$ $\land$ $rep[k].noCrd \}|$ $>$ $\lfloor n/2\rfloor))$ 
$\lor$ $((\valCrd$ $=$ $\{p_i\})$ $\land$ $(FD \neq propV.set) \land (|\{p_k\in FD:$ $rep[k].propV = propV\}|$ $> \lfloor n/2 \rfloor)))$}
{$(status, propV)$ $\gets$ $({\sf Propose}$, $\langle inc()$, $FD\rangle)$} \nllabel{ln:incrCntr}

\ElseIf{$(\valCrd$ $=$ $\{p_i\})$ $\land$ $(\forall$ $p_j$ $\in$ $view.set$ $:$ $rep[j].(view$, $status$, $rnd)$ $=$ $(view$, $status$, $rnd))$ $\lor$ $((status$ $\neq$ ${\sf Multicast})$ $\land$ $(\forall$ $p_j$ $\in$ $propV.set$ $:$ $rep[j].(propV,status)=(propV,{\sf Propose}))$ \nllabel{ln:switch}}{
\If{$status={\sf Multicast}$\nllabel{ln:mulC}}{
$apply(state, msg)$;
$input\gets fetch()$;\\ \nllabel{ln:fetchCrd}
\lForEach{$p_j \in P$}{{\bf if} $p_j \in view.set$~{\bf then} $msg[j]\gets rep[j].input$ {\bf else} $msg[j]\gets \bot$\nllabel{ln:collect}}
$rnd \gets rnd+1$\;\nllabel{ln:rndIncr}
}

\lElseIf{$status={\sf Propose}$\nllabel{ln:proC}}{$(state,status, msg)\gets (synchState(rep), {\sf Install}, synchMsgs(rep))$}
\lElseIf{$status={\sf Install}$\nllabel{ln:insC}}{$(view, status, rnd)\gets (propV,{\sf Multicast}$, 0)}}%}

\ElseIf{$\valCrd=\{p_\ell\} \land \ell\neq i \land ((rep[\ell].rnd = 0 \lor rnd < rep[\ell].rnd \lor rep[\ell].(view\neq propV))$\nllabel{ln:repF}}{

%\If{$(rep[\ell].status={\sf Multicast}) \lor (rep[\ell].status={\sf Install} \land state \neq rep[\ell].state \land p_i \in view.set \cap propV.set)$\nllabel{ln:optCond}}{
\If{$rep[\ell].status={\sf Multicast}$\nllabel{ln:optCond}}{
\lIf{$rep[\ell].state = \bot$ }{$rep[\ell].state \gets state$ \nllabel{ln:optR} $/*$ PCE optimization, line~\ref{ln:opt} $*/$} 
$rep[i] \gets rep[\ell]$; \nllabel{ln:replicate}
$apply(state,rep[\ell].msg)$; $/*$ for the sake of side-effects $*/$ \DontPrintSemicolon\\
$input \gets fetch()$; \nllabel{ln:fetchFol}
}
\lElseIf{$rep[\ell].status = {\sf Install} $\nllabel{ln:adoptRep}}{$rep[i] \gets rep[\ell]$}
\lElseIf{$rep[\ell].status = {\sf Propose}$}{$(status, propV) \gets rep[\ell].(status, propV)$}
\nllabel{ln:adoptProp}
}

{\bf let} $m = rep[i]$ $/*$ sending messages: all to coordinator and coordinator to all $*/$ \;\nllabel{ln:mPrep} 

\lIf{$status={\sf Multicast} \land rnd (\bmod~PCE) \neq 0$ }{$m.state$ $\gets$ $\bot$ $/*$ PCE optimization, line~\ref{ln:optR} $*/$ } \nllabel{ln:opt}

{\bf let} $sendSet$ $=$ $(seemC‏rd$ $\cup$ $\{ p_k$ $\in$ $propV.set$ $:$ $\valCrd$ $=$ $\{ p_i \} \}$ $\cup$ $\{ p_k$ $\in$ $FD$ $:$ $noCrd$ $\lor$ $(status$ $=$ ${\sf Propose}) \})$ \nllabel{ln:sendSet}

\lForEach{$p_j$ $\in$ $sendSet$}{$send(m)$} \nllabel{ln:send}
%\lIf{$status \neq {\sf Propose}$}{\lForEach{$p_j$ $\in$ $FD$}{$FD$} \nllabel{ln:sendRep}\DontPrintSemicolon}
} 

\noindent {\bf Upon message arrival} $m$ {\bf from} $p_j$ {\bf do} $rep[j]\gets m$\; \nllabel{ln:receive}
%\noindent {\bf Upon message arrival} $FD_j$ {\bf from} $p_j$ {\bf do} $rep[j].FD\gets FD_j$\; \nllabel{ln:receiveRep}

\end{scriptsize}
\end{algorithm*}
%%%%%%%%%%%%%%%%%%%%%%%%%%%%% PROOFS %%%%%%%%%%%%%%%%%%%%%%%%%%%%%%%%
%\setlength{\textfloatsep}{10pt}

%\subsection{Detailed Description of Algorithm~\ref{alg:multVirSyn}}
\subsection{Virtual Synchrony Algorithm}
\label{subsec:VSalgDetails}
% 

%\Paragraph{Self-stabilizing virtual synchrony implementation}   
\subsubsection{\trnsfr{High-level algorithm description}}

\trnsfr{
%\paragraph{Practically self-stabilizing replicate state machine implementation}  
Each participant maintains a replica of the state machine and the last processed (composite) message. Note that we bound the memory used to store the history of the replicated state machine 
by deciding to have the (encapsulated influence of the history represented by the) current state of the replicated state machine. In addition, each participant maintains the
last delivered (composite) message to ensure common reliable multicast, as the coordinator may stop being active prior to ensuring that all members received a copy of the last multicast message. 
Whenever a new coordinator is installed, the coordinator inquires all members (forming a majority) for the most updated state and delivered message.
Since at least one of the members, say $p_i$, participated in the view in which the last completed state machine transition took place, $p_i$'s information will be recognized as associated with
the largest counter, adopted by the coordinator that will in turn assign the most updated state and available delivered message to all the current group members, in essence satisfying the virtual synchrony property.}

\trnsfr{
After this, the coordinator, as part of the multicast procedure, will collect inputs received from the environment before ensuring that all group members apply these inputs to the replica state machine. 
Note that the received multicast message consists of input (possibly $\bot$) from each of the processors, thus, the processors need to apply one input at a time, 
the processors may apply them in an agreed upon sequential order, say from the input of the first processor to the last. Alternatively, the coordinator may request one input at a time in a round-robin fashion and multicast it. Finally, to ensure that the system stabilizes when started in an arbitrary
configuration, every so often, the coordinator assigns the state of its replica to the other members. %\vspace{.5em}
}

\trnsfr{
%It may happen that the system reaches a configuration with no coordinator.  For example, this could be the case in the arbitrary configuration that the system starts in, or in the case that the coordinator of an installed view is no longer active. 
If the system reaches a configuration with no coordinator, e.g., due to an arbitrary configuration that the system starts in, or due to the coordinator's crash. 
Each participant detecting the absence if a coordinator, seeks for potential candidates based on the exchanged information.
A processor $p$ regards a processor $q$ as a candidate, if $q$ is active according to $p$'s failure detector, and there is a majority of processors that also think so (all these are based on $p$'s knowledge, which due to asynchrony might not be up to date). 
When there is more than one such candidate, processor $p$ checks whether there is a candidate that has proposed a view with a highest identifier (i.e., counter) among the candidates. 
If there is one, then $p$ considers this to be the coordinator and waits to hear from it (or learn that it is not active). 
}

\trnsfr{
If there is none such, and if based on its local knowledge there is a majority of processors that also do not have a coordinator, then processor $p$ acquires a counter from the counter increment algorithm and proposes a new view, with view ID, the counter, and group membership,  the set of processors that appear active according to its failure detector. 
%id the counter and group member the set of active processors returned by its failure detector. 
As we show, if $p$ receives an ``accept" message from {\em all} the processors in the view, then it proceeds to install the view, unless another processor who has obtained a higher counter does so. 
In a transition from one view to the next, there can be several processors attempting to become the coordinator (namely, those who according to their knowledge have a supporting majority). 
Still, by exploiting the intersection property of the supporting majorities we prove that each of these processors will propose a view at most once, and out of these, one view will be installed (i.e., we do not have never-ending attempts for new views to be installed). 
}

\bl{
As an aside, we note that GCSs that provide VS often leverage on the system's ability to preserve (when possible) the coordinator during view transitions rather than venturing to install a new one, a certainly more expensive procedure.
Our solution naturally follows this approach through our assumption of a supportive majority (Definition~\ref{def:Primary}), where coordinators enjoy the support of a majority of processors by never being supported throughout a very long period. 
During such a period, our algorithm persists on using the same coordinator, even when views change.
}

\subsubsection{\trnsfr{Detailed algorithm description}} The existence of coordinator $p_\ell$ is in the heart of Algorithm~\ref{alg:multVirSyn}. 
Processors that belong to and accept $p_\ell$'s view proposal are called the \emph{followers} of $p_\ell$.
The algorithm determines the availability of a coordinator and acts towards the election of a new one when no valid such exists (lines~\ref{ln:FD} to~\ref{ln:incrCntr}). The pseudocode details the coordinator-side (lines~\ref{ln:switch} to~\ref{ln:insC}) and the follower-side (lines~\ref{ln:repF} to~\ref{ln:adoptRep}) actions. At the end of each iteration the algorithm, defines how $p_\ell$ and its followers exchange messages (lines~\ref{ln:opt} to~\ref{ln:receive}).  

%\noindent {\em Remark:} it is possible that your proposed value is not the one written (because of concurrency/asynchrony).
%This is not a problem, is a typical issue in the MWMR register implementations. [[[[REMEMBER TO NOTE THIS IN THE VIrtual Synch, algorithm, when %a process proposes a counter]]]]

%%%%%%%%%%%%%%%%%%%%%%%%%%%%%%%%%%%%%%%%%%%%%%%%%%%%%%%%%%%%%%%%%%%%%%%%%%%%%%%%%
%\begin{figure*}[t!]
%\begin{framed}
%\begin{scriptsize}

\Paragraph{The processor state and interfaces} 
The state of each processor includes its current $view$, and $status = \{{\sf Propose},{\sf Install},{\sf Multicast}\}$, which refers to usual message multicast operation when in ${\sf Multicast}$, or view establishment rounds in which the coordinator can ${\sf Propose}$ a new view and proceed to ${\sf Install}$ it once all preparations are done (line~\ref{ln:var}). During multicast rounds, $rnd$ denotes the round number, $state$ stores the replica, $msg[n]$ is an array that includes the last delivered messages to the state machine, which is the $input$ fetched by each group member and then aggregated by the coordinator during the previous multicast round. During multicast rounds, it holds that $propV=view$. However, whenever $propV\neq view$ we consider $propV$ as the newly proposed view and $view$ as the last installed one. Each processor also uses $noCrd$ and $FD$ to indicate whether it is aware of the absence of a recently active and connected valid coordinator, and respectively, of the set of processor present in the connected component, as indicated by its local failure detector. The processors exchange their state via message passing and store the arriving messages in the replica's array, $rep[n]$ (line~\ref{ln:receive}), where $rep[i].(view$, $\ldots$, $noCrd)$ 
is an alias to the aforementioned variables and $rep[j]$ refers to the last arriving message from processor $p_j$ containing $p_j$'s $rep[j]$. Our presentation also uses subscript $_k$ to refer to the content of a variable at processor $p_k$, e.g., $rep_k[j].view$, when referring to the last installed view that processor $p_k$ last received from $p_j$. 

Algorithm~\ref{alg:multVirSyn} assumes access to the application's message queue via $fetch()$, which returns the next multicast message, or $\bot$ when no such message is available (line~\ref{ln:inter}). 
It also assumes the availability of the automaton state transition function, $apply(state, msg)$, which applies the aggregated input array, $msg$, to the replica's $state$ and produces the local side effects. 
The algorithm also collects the followers' replica states and uses $synchState(replica)$ to return the new state. 
The function $failureDetector()$ provides access to $p_i$'s failure detector, and the function $inc()$ (counter increment) fetches a new and unique (view) identifier, $ID$, that can be totally ordered by $\preceq_{ct}$ and $ID.wid$ is the identity of the processor that incremented the counter, resulting to the counter value $ID$ (hence view $ID$s are counters as defined in Section~\ref{subsec:CounterI}). 
Note that when two processors attempt to concurrently increment the counter, due to symmetry breaking, one of the two counters is the largest. 
Each processor will continue to propose a new view based on the counter written, but then (as described below) the one will the highest counter will succeed (line~\ref{ln:valCrd}). 

\Paragraph{Determining coordinator availability} 

\trnsfr{
The algorithm takes an agile approach to multicasting with atomic delivery guarantees. \sloppy{Namely, a new view is installed whenever the coordinator sees a change to its local failure detector, $failureDetector()$, which $p_i$ stores in $FD_i$ (line~\ref{ln:FD}).} 
Nevertheless, we might reach a configuration without a  view coordinator as a result of an arbitrary initial configuration, or of a coordinator becoming inactive. 
Using the failure detector heartbeat exchange, processors can detect such initially corrupted states.
Each participant that detects that it has no coordinator, seeks for potential candidates based on the exchanged information.
}

\trnsfr{
Processor $p_i$ can see the set of processors, $seemCrd_i$, that each \emph{seems} to be the view coordinator, because $p_i$ stored a message from $p_\ell \in FD_i$ in which $p_\ell$ $=$ $rep[\ell].propV.ID.wid$. Note that $p_i$ cannot consider $p_\ell$ as a (seemly) coordinator unless the conditions in line~\ref{ln:seemCrd} hold.
Intuitively, such a processor must be active according to $p_i$'s failure detector, and there is a majority of processors that also think so. 
Note that all these are based on local knowledge, which due to asynchrony might not be up to date. 
%when $p_\ell$'s proposal view does not include a majority, $p_\ell$ is not a member of the view it claims to coordinate and, in the case of ${\sf Multicast}$ rounds, their view fields match their view proposal fields (line~\ref{ln:seemCrd}). 
%
% where a processor $i$ might consider processor $j$ to be its coordinator, but processor
%$j$ does not consider itself to be the coordinator. 
%
The next step is for $p_i$ to consider the processor in $seemCrd_i$ with the $\preceq_{ct}$-greatest view identifier (line~\ref{ln:valCrd}) as the valid coordinator. 
Here, set $valCrd_i$ is either a singleton or empty (line~\ref{ln:noCrd}). 
If $p_i$ considers some processor $p_\ell$ as a valid coordinator, it waits to hear from $p_\ell$ (or learn that it is not active). 
We call $p_i$ a \emph{follower} of $p_\ell$.
%If there is no such processor, $p_i$ will not propose a new view before its failure detector indicates that there exists a supportive majority of live and connected processors that also do not have a valid coordinator (line~\ref{ln:incrCntr}). 
If there is no such processor, $p_i$ will only propose a new view if its failure detector indicates that there exists a supportive majority of active processors that are also without a valid coordinator (line~\ref{ln:incrCntr}). 
If such a majority exists, $p_i$ acquires a counter from the counter increment algorithm and proposes a new view, with the counter as the view ID, and the set of processors that appear active according to its failure detector as the group membership. 
%Note that if two (or more) processors attempt to concurrently propose (i.e., increment the counter), due to symmetry breaking, one of the two counters is the largest. Each processor will continue to propose a new view based on the counter written, but only one with the highest counter will succeed (line~\ref{ln:valCrd}).
%In the case $p_i$ is a valid coordinator, it will propose a new view whenever the last proposed view does not match the set of processors that %were recently live and connected according to its (local) failure detector.
}

\trnsfr{
As we show, if $p_i$'s view is accepted from {\em all} the processors in the view, then it proceeds to install the view, unless another processor who has obtained a higher counter does so. 
In a transition from one view to the next, there can be several processors attempting to become the coordinator (namely, those who according to their knowledge have a supporting majority). 
Still, by exploiting the intersection property of the supporting majorities we prove that each of these processors will propose a view at most once, and out of these, one view will be installed (i.e., we do not have never-ending attempts for new views to be installed). 
To satisfy the VS property, no new multicast message is delivered to a new  view, before the coordinator of this new view has collected all the participants' last delivered messages (of their prior views) and has resent the messages appearing not to have been delivered uniformly. 
}

\Paragraph{The coordinator-side}  
Processor $p_i$ is aware of its valid coordinatorship when $(valCrd_i$ $=$ $\{p_i\})$ (line~\ref{ln:switch}). 
It takes action related to its role as a coordinator when it detects the round end, based on input from other processors. 
During a normal ${\sf Multicast}$ round, $p_i$ observes the round end once for every view member $p_j$ it holds that $(rep_i[j].(view$, $status$, $rnd)$ $=$ $(view_i$, $status_i$, $rnd_i))$ in line~\ref{ln:switch}. 
In the cases of ${\sf Propose}$ and ${\sf Install}$ rounds, the algorithm does not need to consider the round number, $rnd$.

Depending on its $status$, the coordinator $p_i$ proceeds once it observes the successful round conclusion. 
At the end of a normal ${\sf Multicast}$ round (line~\ref{ln:mulC}), the coordinator increments the round number (line~\ref{ln:rndIncr}) after applying the changes to its local replica (line~\ref{ln:fetchCrd}) and aggregating the followers' input (line~\ref{ln:collect}). 
The coordinator continues from the end of a ${\sf Propose}$ round to an ${\sf Install}$ round after using the most recently received replicas to install a synchronized state of the emulated automaton (line~\ref{ln:proC}). 
At the end of a successful ${\sf Install}$ round, the coordinator proceeds to a ${\sf Multicast}$ round after installing the proposed view and the first round number (line~\ref{ln:insC}). 
(Note that implicitly the coordinator creates a new view if it detects that the round number is exhausted ($rnd > 2^{\tau}$), or if there is another member of its view that has a greater round number than the one this coordinator has. This can only be due to corruption in the initial arbitrary state which affected $rnd$ part of the state.)
%{\color{BrickRed} [[@@ ``Implicitly'' means it does not appear in the code. Is this good? Moreover, shouldn't we at least refer to the stabilization issues of this somewhere in the proof or even here? The correctness is simpler than the counter algorithm since we are talking about a single writer. @@]]}

\Paragraph{The follower-side} 
Processor $p_i$ is aware of its coordinator's identity when $(\valCrd_i$ $=$ $\{p_\ell\})$ and $i \neq \ell$ (line~\ref{ln:repF}). 
%It has to act upon merely new messages, 
Being a follower, $p_i$ only enters this block of the pseudocode when it receives a new message, i.e., the first message round when installing a new view $(rep[\ell].rnd$ $=$ $0)$, the first time a message arrives $(rnd < rep[\ell].rnd)$ or a new view is proposed $(rep[\ell].(view\neq propV))$.

During normal ${\sf Multicast}$ rounds (lines~\ref{ln:optCond}--\ref{ln:fetchFol}) the follower $p_i$ adopts the coordinator's replica, applies the aggregated message of this round to its current automaton state so that it produces the needed side-effects, and then fetches new messages from the environment. 
Note that, in the case of a ${\sf Propose}$ round (line~\ref{ln:adoptProp}), the algorithm design stops $p_i$ from overwriting its round number, thus allowing the coordinator to know what was the last round number that it delivered during the last installed view.
This is only overwritten on upon the installation of the new view~(line~\ref{ln:adoptRep}). 

\Paragraph{The exchanging message and PCE optimization} 
Each processor periodically sends its current replica (line~\ref{ln:send}) and stores the received ones (line~\ref{ln:receive}). 
As an optimization, we propose to avoid sending the entire replica state in every 
${\sf Multicast}$ round. Instead, we consider a predefined constant, $PCE$ (periodic consistency enforcement), that determines the maximum number of ${\sf Multicast}$  rounds during which the followers do not transmit their replica state to the coordinator and the coordinator does not send its state to them (lines~\ref{ln:optR} and~\ref{ln:opt}). Note that the greater the $PCE$'s size, the longer it takes to recover from transient faults. Therefore, one has to take this into consideration when extending the approach of periodic consistency enforcement to other elements of replica, e.g., in $view$ and $propV$, one might want to reduce the communication costs that are associated with the $set$ field and the epoch part of the $ID$ field.

\subsection{Correctness Proof of Algorithm~\ref{alg:multVirSyn}}
\label{app:VS}

\bl{
The correctness proof shows that starting from an arbitrary state in an execution $R$ of Algorithm~\ref{alg:multVirSyn} and once the primary partition property (Definition~\ref{def:Primary}) holds throughout $R$, we reach a configuration $c \in R$ in which some processor with supporting majority $p_\ell$ will propose a view including its supporting majority. 
This view is either accepted by all its member processors or in the case where $p_\ell$ experiences a failure detection change, it can repropose a view. 
}

\bl{We conclude by proving that any execution suffix of $R$ that begins from such a configuration $c$ will preserve the virtual synchrony property and implement state machine replication. We begin with some definitions. 
Intuitively, the latter part of the proof  %We then proceed to show how the how the algorithm preserves Virtual Synchrony 
is deduced as follows: 
once a processor does not have a coordinator, it stops participating in group multicasting, and prior to delivering a new multicast message in a new  view, the algorithm assures that the coordinator of this new view has collected all the participants' last delivered messages (in their prior views) and resends the messages appearing not to have been delivered uniformly. 
To do so, each participant keeps the last delivered message and the view identifier that delivered this message. 
This, together with the intersection property of majorities,  provides the virtual synchrony property. We begin with some definitions.} 

%\paragraph*{Definitive suspicions} 
Once the system considers processor $p_\ell$ as the view coordinator (Definition~\ref{def:Primary}) its supporting majority can extend the support throughout $R$ and thus $p_\ell$ continues to emulate the automaton with them. 
Furthermore, there is no clear guarantee for a view coordinator to continue to coordinate for an unbounded period when it does not meet the criteria of Definition~\ref{def:Primary} throughout $R$. 
Therefore, for the sake of presentation simplicity, the proof considers any execution $R$ with only \emph{definitive suspicions}, i.e., once processor $p_i$ suspects processor $p_j$, it does not stop suspecting $p_j$ throughout $R$. 
The correctness proof implies that eventually, once all of $R$'s suspicions appear in the respective local failure detectors, the system elects a coordinator that has a supporting majority throughout~$R$.

Consider a configuration $c$ in an execution $R$ of Algorithm~\ref{alg:multVirSyn} and a processor $p_i \in P$.
We define the \emph{local (view) coordinator} of $p_i$, say $p_j$, to be the only processor that, based on $p_i$'s local information, has a proposed view satisfying the conditions of lines~\ref{ln:seemCrd} and~\ref{ln:valCrd} such that $valCrd=\{p_j\}$.
$p_j$ is also considered the \emph{global (view) coordinator} if 
%$\forall p_k$ $\in$ $P_{\majority}(j)$  
for all $p_k$ in $p_j$'s proposed view ($propV_j$), it holds that $valCrd_k=\{p_j\}$.
When $p_i$ has a (local) coordinator then $p_i$'s local variable $noCrd$ $=$ ${\sf False}$, whilst when it has no local coordinator, $noCrd$ $=$ ${\sf True}$. 
Moving to the proof, we consider the following useful remark on Definition~\ref{def:Primary} of page~\pageref{def:Primary}. 

\begin{remark}
Definition~\ref{def:Primary} suggests that we can have more than one processor that has supporting majority.
In this case, it is not necessary to have \emph{the same} supporting majority for all such processors. 
Thus for two such processors $p_i, p_j$ with respective supporting majorities $P_{\majority}(i)$ and $P_{\majority}(j)$ we do not require that $P_{\majority}(i) = P_{\majority}(j)$, but $P_{\majority}(i) \cap P_{\majority}(j) \neq \emptyset$ trivially holds.
\end{remark}

\begin{lemma}
\label{th:locCoord}
Let $R$ be an execution with an arbitrary initial configuration, of Algorithm~\ref{alg:multVirSyn} such that Definition~\ref{def:Primary} holds.
Consider a processor $p_i \in P_{\majority}$ which has a local coordinator ${p_k}$, such that $p_k$ is either inactive or it does not have a supporting majority throughout $R$.
There is a configuration $c \in R$, after which $p_i$ does not consider ${p_k}$ to be its local coordinator.
\end{lemma}

\begin{proof} There are the two possibilities regarding processor $p_k$.

\noindent \textbf{Case 1:} We first consider the case where ${p_k}$ is inactive throughout $R$. 
By the design of our failure detector, $p_i$ is informed of $p_k$'s inactivity such that line~\ref{ln:FD} will return an $FD_i$ to $p_i$ where ${p_k} \notin FD_i$. 
The threshold $W$ that we set for our failure detector determines how soon $p_k$ is suspected.
By the first condition of line~\ref{ln:seemCrd} we have that $p_k \notin FD_i \Rightarrow p_k \notin seemCrd \Rightarrow p_k \notin valCrd_i $, i.e., $p_i$ stops considering $p_k$ as its local coordinator.  
By definitive suspicions, that $p_i$ does not stop suspecting $p_k$ throughout~$R$. 

We now turn to the case where ${p_k}$ is active, however it does not have a supporting majority throughout $R$, but $p_i$ still considers $p_k$ as its local coordinator, i.e. $valCrd_i=\{p_k\}$. Two subcases exist:

\noindent \textbf{Case 2(a):} $p_k$ considers itself to have a supporting majority, and  $p_i \in propV_k$. 
Note that the latter assumption implies that $p_k$ is forced by lines \ref{ln:mPrep} - \ref{ln:send} to propagate $rep_k[k]$ to $p_i$ in every iteration.
%By the failure detector, there exists an iteration where $p_k$ will be informed that $|FD|<n/2$.
By the failure detector, there exists an iteration where $p_k$ will have $|FD_k = n/2 +1|$ and is informed that some $p_j \in propV_k$ has $p_k \notin FD_j$ and so the condition of line~\ref{ln:seemCrd} ($FD> \lfloor n/2 \rfloor$) fails for $p_k$, which stops being the coordinator of itself.
If $p_k$ does not find a new coordinator, hence $noCrd_k = {\sf True}$, then $p_k$ propagates its $rep_k[k]$ to $p_i$. 
But this implies that $p_i$ receives $rep_k[k]$ and stores it in $rep_i[k]$. 
Upon the next iteration of this reception, $p_i$ will remove $p_k$ from its $seemCrd$ set because $p_k$ does not satisfy the condition $|rep_i[k].FD|< \lfloor n/2 \rfloor$ of line~\ref{ln:seemCrd}. 
We conclude that $p_i$ stops considering $p_k$ as its local coordinator if $p_k$ does not find a new coordinator.
Nevertheless, $p_k$ may find a new coordinator before propagating $rep_k[k]$. If $p_k$ has a coordinator other than itself, then it only propagates $rep_k[k]$ to its coordinator and thus $p_i$ does not receive this information. 
We thus refer to the next case:

\noindent \textbf{Case 2(b):} $p_k$ has a different local coordinator than itself. 
This can occur either as described in Case 2(a) or as a result of an arbitrary initial state in which $p_i$ believes that $p_k$ is its local coordinator but $p_k$ has a different local coordinator. 
We note that the difficulty of this case is that $p_k$ only sends $rep_k[k]$ to its coordinator, and thus the proof of Case 2(a) is not useful here.
As explained in Algorithm~\ref{alg:multVirSyn}, the failure detector returns a set with the identities ($pid$) of all the processors it regards as active, as well as the identity of the local coordinator of each of these processors. 
As per the algorithm's notation, the coordinator of processor $p_k$ is given by $crd(k)$. 
Since $p_i$'s failure detector regards $p_k$ as active, then $crd(k)$ is indeed updated (remember that $p_i$ receives the token with $p_k$'s $crd(k)$ infinitely often from $p_k$), otherwise $p_k$ is removed from $FD$ and is not a valid coordinator for $p_i$.
But $p_k$ does not consider itself  as the coordinator (by the assumption of Case 2(b)), and thus it holds that $crd(k)\neq k$.
Therefore, in the first iteration after $p_i$ receives $crd(k)\neq k$, one of the last two conditions of line~\ref{ln:seemCrd} fails (depending on what is the view status that $p_i$ has in $rep_i[k]$) so $p_k \not \in seemCrd_i$ and thus $valCrd_i \neq \{p_k\}$. 
We conclude that any such $p_k$ stops being $p_i$'s coordinator and by the assumption of definitive suspicions we reach to the result. 
It is also important to note that $p_k$ never again satisfies all the conditions of line~\ref{ln:incrCntr} to create a new view.
\end{proof}

We now define the notion of ``propose'' more rigorously to be used in the sequel.
\begin{definition}
\label{def:propose}
Processor $p_{\ell} \in P$ with $status = {\sf Propose}$, is said to \emph{propose} a view $propV_{\ell}$, if in a complete iteration of Algorithm~\ref{alg:multVirSyn}, $p_{\ell}$ either satisfies $valCrd_{\ell} = \{p_{\ell}\}$ or satisfies all the conditions of line~\ref{ln:incrCntr} to create $propV_{\ell}$.
A proposal is completed when $propV_{\ell}$ is propagated through lines~\ref{ln:mPrep}--\ref{ln:send} to all the members of $FD_\ell$.
\end{definition}

\noindent The above definition does not imply that $p_{\ell}$ will continue proposing the view $propV$, since the replicas received from other processors may force $p_{\ell}$ to either exclude itself from $valCrd_\ell$ or create a new view (see Lemma~\ref{th:boundProp}).
If the view is installed, then the proposal procedure will stop, although $propV_\ell$ will still be sent as part of the replica propagation at the end of each iteration.
Also note that the origins of such a proposed view are not defined. 
Indeed it is possible for a view that was not created by $p_\ell$ but bears $p_\ell$'s creator identity to come from an arbitrary state and be proposed, as long as all the conditions of lines \ref{ln:seemCrd} and~\ref{ln:valCrd} are met.

\begin{lemma}
\label{th:VSstab}
If the conditions of Definition~\ref{def:Primary} hold throughout an execution $R$ of Algorithm~\ref{alg:multVirSyn}, then starting from an arbitrary configuration in which there is no global coordinator, the system reaches a configuration in which at least one processor with a supporting majority will propose a view (with ``propose'' defined as in Definition~\ref{def:propose}).
\end{lemma}

\begin{proof}
By Definition~\ref{def:Primary}, at least one processor with supporting majority exists.
Denote  one such processor as $p_\ell$. 
%If $p_\ell$ holds a view proposal from the initial arbitrary configuration that can be proposed to the system, then the lemma trivially holds. Let's look at the case where no such view proposal exists.
Assume for contradiction that throughout $R$, no processor $p_\ell$ with supporting majority proposes a view. %, nor does it hold a view appearing to have been created by itself and coming from the arbitrary configuration.
%i.e., $p_\ell$ does not hold or create a $propV$ that can satisfy the conditions in lines~\ref{ln:seemCrd} and~\ref{ln:valCrd} of the algorithm.
%This implies that $p_\ell$ either has a coordinator that is not global or does not have a coordinator, but also does not know of a majority of processors that do not have a coordinator and thus cannot propose by the third condition of line~\ref{ln:incrCntr}.
%Note that the first condition in line~\ref{ln:incrCntr} is satisfied by our assumption that $p_\ell$ is not suspected by a majority throughout $R$.
%
%If $p_\ell$ has a local coordinator, say $p_k$, then there are two sub-cases: 
%Either this coordinator has a supporting majority or it does not. 
%If this coordinator does not have a supporting majority then by Lemma~\ref{th:locCoord} the execution will reach a configuration in which $p_\ell$ does not have $p_k$ as its coordinator. 
%
$p_\ell$ either has a local coordinator (that is not global) or does not have a coordinator.\\ 
\textbf{Case 1: $\mathbf{p_{\ell}}$ does not have a coordinator ($\mathbf{noCrd_{\ell} = {\sf True}}$)}.
If $p_{\ell}$ does not propose a view (as per the ``propose'' Definition~\ref{def:propose}), this  is because it does not hold a proposal that is suitable and it does not satisfy some condition of line~\ref{ln:incrCntr} which would allow it to create a new view.
The first condition of line~\ref{ln:incrCntr}, $(|FD|>\lfloor n/2\rfloor)$ is always satisfied by our assumption that $p_\ell$ is not suspected by a majority throughout $R$.
In the second condition, both (i) $((|valCrd_\ell| \neq 1)$ $\land$ $(|\{ p_i \in FD_{\ell} $ $:$ $p_{\ell}$ $\in$ $rep_{\ell}[i].FD_{\ell}$ $\land$ $rep_{\ell}[i].noCrd \}|$ $>$ $\lfloor n/2\rfloor))$ and (ii) $((valCrd_\ell = \{p_\ell\})$ $\land$ $(FD_\ell \neq propV_\ell.set) \land (|\{p_i\in FD:$ $rep[i].propV = propV\}|$ $> \lfloor n/2 \rfloor))$ must fail due to our assumption that $p_\ell$ never proposes.
Indeed (ii) fails since $noCrd_\ell = {\sf True} \Rightarrow valCrd_\ell \neq \{p_\ell\}$.
If the first expression also fails, this implies that throughout $R$, $p_\ell$ does not know of a majority of processors with $noCrd = {\sf True}$ and so it cannot propose a new view.

Let's assume that only one processor $p_j \in P_{\majority}(\ell) \subseteq FD_\ell$ is required to switch from $noCrd_j= {\sf False}$ to $\sf True$ in order for $p_\ell$ to gain a majority of processors without a coordinator.
But if $noCrd_j= {\sf False}$ then $p_j$ must already have a coordinator, say $p_k$.
We have the following two subcases:\\
\emph{Case 1(a):} $p_k$ does not have a supporting majority.
Lemma~\ref{th:locCoord} guarantees that $p_j$ stops considering $p_k$ as its local coordinator.
Thus $p_j$ eventually goes to $noCrd= {\sf True}$ and by the propagation of its replica, $p_\ell$ receives the required majority to go into proposing a view.
But this contradicts our initial assumption, so we are lead to the following case.\\
\emph{Case 1(b):} $p_k$ has a supporting majority and a view proposal $propV_k$ from the initial arbitrary configuration but is not the global coordinator. 
But this implies that the Lemma trivially holds, and so the following case must be true.\\
\noindent\textbf{Case 2: $\mathbf{p_{\ell}}$ has a coordinator, say $\mathbf{p_{k'}}$.}
The two subcases of whether $p_{k'}$ has a supporting majority or not, are identical to the two subcases 1(a) and 1(b) concerning $p_k$ that we studied above.
Thus, it must be that either $p_\ell$ will eventually propose a label, or that $p_{k'}$ has a proposed view, thus contradicting our assumption and so the lemma follows.
~\end{proof}

\vspace{.4em}
Lemma~\ref{th:VSstab} establishes that at least one processor with supporting majority will propose a view in the absence of a valid coordinator. 
We now move to prove that such a processor will only propose one view, unless it experiences changes in its $FD$ that render the view proposal's membership obsolete.
The lemma also proves that any two processors with supporting majority will not create views in order to compete for the coordinatorship.

\begin{lemma}[\bl{Closure and Convergence}]
\label{th:boundProp}
If the conditions of Definition~\ref{def:Primary} hold throughout an execution $R$ of Algorithm~\ref{alg:multVirSyn}, then starting from an arbitrary configuration, the system reaches a configuration in which any processor $p_\ell$ with a supporting majority proposes a view $propv_\ell$, 
%may propose itself as the coordinator at most once {\color{BrickRed} with the same membership. 
and cannot create a new proposed view in $R$ unless $FD_\ell \neq propV_\ell.set$
and a majority of processors has adopted $propV_\ell$.
As a consequence, the system reaches a configuration in which one processor with supporting majority is the global coordinator until the end of the execution.
\end{lemma}

\begin{proof}
We distinguish the following cases:\\
\textbf{Case 1: Only one processor with supporting majority exists.} Assume there is only a single processor $p_\ell$ that has a supporting majority throughout $R$. 
According to Lemma~\ref{th:VSstab}, $p_\ell$ must eventually propose a view $propV_{\ell}$, based on the current $FD_{\ell}$ reading (line~\ref{ln:FD}) which becomes the $propV_{\ell}.set$. 
By Lemma~\ref{th:locCoord}, any other processor without a supporting majority will eventually stop being the local coordinator of any %$p_j \in P_{\majority}(\ell)$, 
$p_j \in propV_\ell.set$
and since such processors do not have a supporting majority, the first condition of line~\ref{ln:incrCntr} will prevent them from proposing. 
%Assume that $p_{\ell}$ proposes for a second time.
%Given that no other processor can propose, there exists some proposal that has a greatest view identifier than the one $p_{\ell}$ has proposed. 
%We note that the increment counter algorithm will return the greatest identifier of all the previous generated. In cases of concurrent calls, each processor proposes the view with its own view identifier without any guarantees on which is the greatest.
%Thus, there must have been a second processor that called $inc()$ either concurrently or after $p_\ell$. 
%But $p_\ell$ was the only processor that could propose (by being the only one having a supporting majority). 

%Processor $p_\ell$ waits for replies from all the members of $propV_{\ell}.set$, in order to move from status ${\sf Propose}$ to status ${\sf Install}$ and similarly from ${\sf Install}$ to the first round of ${\sf Multicast}$.
%In this process it continuously iterates through the algorithm, checking whether it is still the valid coordinator locally and then propagates its replica (including the proposal) to the members of its $FD$.
%Since $p_{\ell}$ has a supporting majority, $propV_\ell$ must satisfy all the conditions to locally classify $p_\ell$ as a valid coordinator.

Processor $p_\ell$ continuously proposes $propV_\ell$ until all processors in $propV_\ell.set$ have sent a replica showing that they have adopted $propV_\ell$ as their $propV$.
%After this, $p_\ell$ can install the view and become the global coordinator.
%In this process it continuously iterates through the algorithm, checking whether it is still the valid coordinator locally and then propagates its replica (including the proposal) to the members of its $FD$.
%Let's assume for contradiction that some $p_j \in propV_\ell$ does not adopt $propV_\ell$ in $R$, thus stopping $p_\ell$ from becoming the global coordinator.
%Our assumption implies that $p_j \in FD_\ell$, otherwise condition $(valCrd_\ell= \{p_\ell\}) \land (propV_\ell \neq FD_\ell)$ (line~\ref{ln:incrCntr}) would oblige $p_\ell$ to propose a new view excluding $p_j$.
%We thus proceed considering the case where $p_j \in FD_\ell$ is true throughout $R$.
Every processor that is alive throughout $R$ and in $FD_\ell$ should receive this replica through the self-stabilizing reliable communication.
The only condition that may prevent $p_j$ to adopt $propV_\ell$ is if for some $p_r \in rep_j[\ell].propV_\ell.set$ it holds that $p_\ell \not \in rep_j[r].FD$ (line~\ref{ln:seemCrd}). 
Plainly put, $p_j$ believes that $p_r$ suspects~$p_\ell$.

\textbf{Case 1(a):} If $p_j$'s information is correct about $p_r$, then $p_r \not \in P_{\majority}(\ell)$. 
Thus at some point $p_\ell$ will suspect $p_r$ and exclude $p_r$ from $FD_\ell$. 

\textbf{Case 1(b):} If $p_j$'s information is false --remnant of some arbitrary state-- then  $p_\ell \in FD_{r}$ and since $p_r$, by the last condition of line~\ref{ln:sendSet}, sends $rep_r[r]$ infinitely often to $p_j$, then $rep_j[r]$ will be corrected and $p_j$ will accept $propV_\ell$.

Since $p_\ell$ has a majority $P_{\majority}(\ell) \subseteq propV_{\ell}.set$, then at least a majority of processors have received $propV_\ell$ and eventually accept it.
If some processor $p_j \in propV_\ell$ does not adopt $p_\ell$'s proposal in $R$, it is eventually removed from $FD_{\ell}$ and thus does not belong to the supporting majority of $p_\ell$ (as detailed in Case 1(a) above).
%
%
%
%By Lemma~\ref{th:locCoord} and by definitive suspicions, $p_j$ stops having other processors as coordinators (since $p_\ell$ has supporting majority).
%%Surely, $propV_\ell$ satisfies all the conditions except for the ones that are based on $p_j$'s state.
%Since no processor other than $p_\ell$ has supporting majority, every processor $p_j$ that may not accept $propV_\ell$, must have some local information that results in $valCrd_j \neq \{p_\ell\}$.
%%Since $p_j \in propV_\ell.set \Rightarrow p_j \in FD_\ell$, then it must be that either $p_\ell \in FD_j$ or $p_j$ is removed from $FD_\ell$ and $p_\ell$ proposes a new view excluding $p_j$.
%%Clearly in the latter case $p_j \not \in P_{\majority} (\ell)$.
%%We must therefore have that $p_j \in P_{\majority} (\ell)$ or otherwise $p_j$ remains sufficiently long in $FD_\ell$ so that it receives $rep[\ell]$ with $propV_\ell$ and $crd_\ell=\ell$.
%By $p_j \in propV_\ell$ we know that $rep_\ell [\ell]$ is propagated an infinite number of times from $p_\ell$ to $p_j$ and must be received.
%Thus $p_j$ will receive $propV_\ell$ and $crd_\ell = \ell$ and these conditions must not exclude $p_\ell$ from $seemCrd_j$. 
%
By the above we note that $p_\ell$ is able to get at least the supporting majority $P_{\majority}(\ell)$ to accept its view if not all of the members in $propV_{\ell}.set$.
In the last case it can proceed to the installation of the view.
If there is any change in the failure detector of $p_\ell$ before it installs a view, $p_\ell$ can satisfy the second case of line~\ref{ln:incrCntr}, to create a new updated view.
Note that in the mean time no processor other than $p_\ell$ can satisfy the conditions of that line, and thus it is the only processor that can propose and become the coordinator.

%Note that if $propV_\ell$ is corrupt in the round number, line~\ref{ln:repF} will prevent adoption..........
\noindent Thus $p_\ell$ eventually becomes the coordinator if it is the single majority-supported processor.

\vspace{.4em}

%{\color{BrickRed} There are now two scenarios:
%
%\emph{Scenario 1:} $FD_\ell$ remains the same until all the members of $propV_{\ell}.set$ accept the proposal.
%Clearly this imposes $p_\ell$ as the coordinator of the members of its view, and any subsequent changes in $p_\ell$'s FD are dealt with a view change.
%No other processor can propose a view since the ones in the supporting majority of $p_{\ell}$ have a coordinator, and the others, even if they believe $p_{\ell}$ is inactive, they cannot form a majority of processors with no coordinator (they are a minority).	
%
%\emph{Scenario 2:} $FD_\ell$ changes before every member of $propV_{\ell}.set$ accepts this set.
%We must assume that this scenario does not happen continuously since there cannot be progress if FDs change very frequently.
%We should eventually reach a Scenario 1 execution to which our assumption of definitive suspicions takes us.
%I.e., the $propV_{\ell}.set$ should include all the members if $P_{\majority}$ but might also include other processors.
%Our definitive suspicions assumption will in the worst case give to $p_\ell$ a proposal set containing only its supporting majority, since all other processors are eventually suspected by $p_\ell$. 
%}
 
%
%
\noindent\textbf{Case 2: More than one processor with supporting majority.} Consider two processors $p_\ell, p_{\ell'}$ that have a supporting majority such that each creates a view (line~\ref{ln:incrCntr}).
By the correctness of our counter algorithm, $inc()$ returns two distinct and ordered counters to use as view identifiers. 
Without loss of generality, we assume that $propV_\ell$ proposed by $p_\ell$ has the greatest identifier of all the counters created by calls to $inc()$.
We identify the following four subcases:

\textbf{Case 2(a):} $p_\ell \in FD_{\ell'} \land p_{\ell'} \in FD_{\ell}$. 
In this case $p_{\ell'}$ will propose its view $propV_{\ell'}$ and wait for all $p_i \in propV_{\ell'}.set$ to adopt it (line~\ref{ln:switch}). 
Whenever $p_\ell$ receives $propV_{\ell'}$, it will store it but will not adopt it, since $propV_{\ell'}.ID \preceq_{ct} propV_\ell.ID$ (line~\ref{ln:valCrd}).
The proposal $propV_\ell$ is also propagated to every $p_i \in propV_\ell.set$. 
Since there is no greater proposed view identifier than $propV_\ell.ID$, this is adopted by all $p_i \in propV_\ell$ which also includes $p_{\ell'}$ as well. 
Thus any processor with supporting majority that belonged to the proposed set of $p_\ell$ will propose at most once, and 
$p_\ell$ will become the sole coordinator.
Note that if $p_{\ell'}$ is prevented from adopting $propV_\ell$ for some time, this is due to reasons detailed and solved in Case~1 of the previous lemma.
The case where the failure detection reading changes for $p_\ell$ is also tackled as in Case~1 of this lemma, by noticing that if $p_\ell$ manages to get a majority of processors of $propV.set$ then $p_{\ell}$ will change its proposed view without losing this majority.

\textbf{Case 2(b):} $p_\ell \not \in FD_{\ell'} \land p_{\ell'} \not \in FD_{\ell}$.
Since both processors were able to propose, this implies that a majority of processors that belonged to each of $p_\ell$'s and $p_{\ell'}$'s supporting majority had informed that they had no coordinator (line~\ref{ln:incrCntr}).
Each of $p_\ell$ and $p_{\ell'}$, proposes its view to its $propV.set$, and waits for acknowledgments from \emph{all} the processors in  $propV.set$ (line~\ref{ln:switch}), in order to install the view.
Since $p_{\ell} \not \in FD_{\ell'}$, $p_{\ell'}$ does not consider $propV_{\ell}$ a valid proposal (line~\ref{ln:seemCrd}) and retains its own proposal that it propagates. The same is done by $p_\ell$. 
Since $p_{\ell}$ has the greatest label, any %$p_i \in P_{\majority}(\ell) \cap P_{\majority}(\ell')$ 
$p_i \in propV_{\ell}.set \cap propV_{\ell'}.set$
might initially adopt $propV_{\ell'}$ but it will eventually choose the greatest $propV_{\ell}$.
If $p_{\ell'}$'s proposal was accepted by all members of $propV_{\ell'}$ then this means that $p_{\ell'}$ became the global coordinator but will then lose the coordinatorship  to $p_{\ell}$ because $propV_{\ell}$ has a greater view identifier. 

What is more crucial, is that $p_{\ell'}$ cannot make another proposal, since it will not have a majority of processors that do not have a coordinator. 
This is deduced from  the intersection property of the two majorities ($propV_{\ell}.set$ and $propV_{\ell'}.set$). 
Since any processor $p_k$ in the intersection %$P_{\majority}(\ell) \cap P_{\majority}(\ell')$ 
$propV_{\ell}.set \cap propV_{\ell'}.set$
has $p_\ell$ as its coordinator, $p_{\ell'}$ does not satisfy the condition $|\{ p_k \in FD_{\ell'} $ $:$ $p_{\ell'}$ $\in$ $rep[k].FD$ $\land$ $rep[k].noCrd \}|$ $>$ $\lfloor n/2\rfloor$ of line~\ref{ln:incrCntr}, and thus cannot propose a new view. 
Processor $p_\ell$ will install its view and remains the sole coordinator. 
Also, $p_\ell$ is the only one that can change its view due to failure detector change since it manages to get a majority of processors in $propV_{\ell}.set$ as opposed to $p_{\ell'}$.

%[[@@ THE ABOVE DOES NOT APPLY: $p_{\ell'}$ considers that $valCrd_{\ell'} = {\ell'}$.
%So on the next change to $FD_{\ell'}$, $p_{\ell'}$ will propose with a new higher view ID by the condition ($valCrd_{\ell'} = \ell' \land FD_\ell \neq propV_{\ell'}.set$) of line~\ref{ln:incrCntr}.
%It appears that in order to avoid this we need to have all the processors with supporting majority to stop having changing FDs, or that all processors with supporting majority that have managed to propose a legitimate proposal eventually enter the FD of each other.]]

\textbf{Case 2(c):} $p_\ell \in FD_{\ell'} \land p_{\ell'} \not \in FD_{\ell}$. Here we note that since $p_\ell$ has the greatest counter but has not included $p_{\ell'}$ to its $propV_{\ell}.set$, it should eventually be able to get all the processors in $propV_{\ell}.set$ to follow $propV_{\ell}$ by using the arguments of Case 2(a).
In the mean time $p_{\ell'}$ will, in vain, be waiting for a response from $p_\ell$ accepting $propV_{\ell'}$.
We note that $p_{\ell'}$  will not be able to initiate a new view once $propV_\ell$ is accepted, since it will not be able to gather a majority of processors with either $noCrd = {\sf True}$ or proposed view $propV_{\ell'}$. 

\textbf{Case 2(d):} $p_\ell \not \in FD_{\ell'} \land p_{\ell'}  \in FD_{\ell}$. This case is not symmetric to the above due to our assumption that $p_{\ell}$ is the one that has drawn the greatest view identifier from $inc()$.
Here $propV_{\ell}.set$ includes $p_{\ell'}$ so $p_{\ell}$ waits for a response from $p_{\ell'}$ to proceed to the installation of $propV_{\ell}$.
On the other hand, $p_{\ell'}$  will be waiting for responses from the processors in $propV_{\ell'}.set$.
Any %$p_i \in P_{\majority}(\ell) \cap P_{\majority}(\ell')$ 
$p_i \in propV_{\ell}.set \cap propV_{\ell'}.set$
cannot keep $propV_{\ell}$ (even if initially it has accepted it, since it does not satisfy condition $p_{\ell'}$ $\in$ $rep[\ell].propV.set$ $\Leftrightarrow$ $p_\ell$ $\in$ $rep[\ell'].FD $ of line~\ref{ln:seemCrd}.
Thus $p_i$ accepts  $propV_{\ell'}$ instead of $propV_{\ell}$, $p_\ell$ cannot propose a different view since it will not be able to get a majority of processors that have $propV_{\ell}$.

%[[@@ By permanent suspicions, $p_{\ell'}$ does not add $p_\ell$ to its $FD$ throughout $R$, but $p_\ell$ may suspect $p_{\ell'}$. In this case, given that $valCrd_\ell = \{p_\ell\} \land FD_\ell \neq propV_\ell.set$, $p_\ell$ will create a view and will become the coordinator. If all processors in $propV_{\ell'}$ adopt this view then $p_{\ell}$ will not be able to propose a different view, since it will not have a majority of processors with $noCrd = {\sf True}$, as detailed in the latter part of Case 2(b).@@]]

By the above exhaustive examination of cases, we reach to the result.
Note that the above proof guarantees both convergence and closure of the algorithm to a legal execution, since $p_\ell$ remains the coordinator as long as it has a supporting majority.
\end{proof}

\begin{theorem}
Starting from an arbitrary configuration, any execution $R$ of Algorithm~\ref{alg:multVirSyn} satisfying Definition~\ref{def:Primary}, simulates automaton replication preserving the virtual synchrony property.
\label{th:VSgood}
\end{theorem}

\begin{proof}
%We consider a finite prefix $R'$ of $R$ which has an arbitrary configuration $c$, and in which there exists a primary partition (as per Definition~\ref{def:Primary}). 
%Assume that this prefix is sufficiently long for Lemma~\ref{th:boundProp} to hold, i.e., to reach a configuration $c_{safe}$ in which there exists a global coordinator for a majority of processors.
%For this configuration we define a view $v$ that has a coordinator $p_\ell$ and that any processor $p_i \in v$ that is not the coordinator is a \emph{follower} of $p_\ell$.
%We define a \emph{multicast round} to be a sequence of ordered events: (i) $fetch()$ input and propagate to coordinator, (ii) coordinator disseminates messages to be delivered in this new round, (iii) messages delivered and (iv) side effects produced by all processors. 
%
Consider a finite prefix $R'$ of $R$. %which has an arbitrary configuration $c$, and in which there exists a primary partition (as per Definition~\ref{def:Primary}). 
Assume that in this prefix Lemma~\ref{th:boundProp} holds, i.e., we reach a configuration in which a processor $p_\ell$ has a supporting majority and is the global coordinator  with view $v$.
%For this configuration we define a view $v$ that has a coordinator $p_\ell$ and that any processor $p_i \in v$ that is not the coordinator is a \emph{follower} of $p_\ell$.
We define a \emph{multicast round} to be a sequence of ordered events: (i) $fetch()$ input and propagate to coordinator, (ii) coordinator propagates the collected messages of this round, (iii) messages are delivered and (iv) all view members $apply()$ side effects. 
The VS property is preserved between two consecutive rounds $r, r'$ that may belong to different views $v, v'$ (with possibly identical coordinators $p_\ell, p_{\ell'}$) respectively, if and only if $\forall p_i \in v.set \cap v'.set$ it holds that every $rep_i[i].input$ at round $r$ is in $rep_[{\ell'}].msg[i]$ of round $r'$.
%
%Our proof is broken into three steps that map the three possible transitions:\\
%\textbf{Step 1: Virtual synchrony is preserved between any two multicast rounds.}\\
%
Our proof is progressive: Claim~\ref{thVS:VSnoVchange} proves that VS is preserved between any two consecutive multicast rounds, Claim~\ref{thVS:VSwithVchange} that VS is preserved in two consecutive views with the same coordinator and Claim~\ref{thVS:VSchangeCrd} preservation in two consecutive view installations where the coordinator changes.

\begin{claim}
\label{thVS:VSnoVchange}
VS is preserved between $r$ and $r'$ where $v=v'$.
\end{claim}

\begin{proof}
Suppose that there exists an input and a related message $m$ in round $r$ that is not delivered within $r$.
We follow the multicast round $r$. First observe the following.

\noindent \textsl{Remark:} Within any multicast round, the coordinator executes lines~\ref{ln:fetchCrd} to \ref{ln:rndIncr} only once and a follower executes lines~\ref{ln:optCond} to~\ref{ln:fetchFol} only once, because the conditions are only satisfied the first time that the coordinator's local copy of the replica changes the round number.

By our Remark we notice that $fetch()$ is called only once per round to collect input from the environment.
This cannot be changed/overwritten since followers can never access $rep[i] \gets rep[\ell]$ of line~\ref{ln:replicate} that is the only line modifying the $input$ field, unless they receive a new round number greater than the one they currently hold. We notice that the followers have produced side effects for the previous round (using $apply()$) based on the messages and state of the previous round.
Similarly, the coordinator executes $fetch()$ exactly once and only before it populates the $msg$ array and after it has produced the side effects for the environment that were based on the previous messages (line~\ref{ln:fetchCrd}). 
Line~\ref{ln:collect} populates the $msg$ array with messages and including $m$.
The coordinator $p_{\ell}$then continuously propagates its current replica but cannot change it by the Remark and until condition  $(\forall$ $p_i$ $\in$ $v.set$ $:$ $rep_\ell[i].(view$, $status$, $rnd)$ $=$ $(view_\ell$, $status_\ell$, $rnd_\ell))$ (line~\ref{ln:switch}) holds again. 
This ensures that the coordinator will change its $msg$ array only when every follower has executed line~\ref{ln:replicate} which allows the aforementioned condition to hold. 

Any follower that keeps a previous round number does not allow the coordinator to move to the next round.
If the coordinator moves to a new round, it is implied that $rep[i] \gets rep[\ell]$ and thus message $m$ was received by any follower $p_i$, by our assumptions that the replica is propagated infinitely often and the data links are stabilizing.
Thus, by the assumptions, any message $m$ is certainly delivered within the view and round it was sent in, and thus the virtual synchrony property is preserved, whilst at the same time common state replication is achieved.
\end{proof}
%\noindent \textbf{Step 2: Virtual synchrony is preserved in two consecutive view installations where there is no change of coordinator.}\\

\begin{claim}
\label{thVS:VSwithVchange}
VS is preserved between $r$ and $r'$ where $v\neq v'$ and $p_\ell = p_{\ell'}$.
\end{claim}

\begin{proof}
We now turn to the case where from one configuration $c_{safe}$ we move to a new $c_{safe}'$ that has a different view $v'$ but has the same coordinator $p_{\ell}$.
Once $p_\ell$ is in an iteration where the condition $FD \neq propV.set$ of line~\ref{ln:incrCntr} holds, a view change is required.
Since $p_\ell$ is the global coordinator holds, no other processor can satisfy the condition $(|\{p_k\in FD_{\ell}:$ $rep[k]_{\ell}.propV = propV_{\ell}\}|$ $> \lfloor n/2 \rfloor)$ of line~\ref{ln:incrCntr}, and so only $p_{\ell}$. 
For more on why this holds one can prefer to Lemma~\ref{th:VSstab}. 
Processor $p_\ell$ creates a new $propV$ with a new view ID taken from the increment counter algorithm, which is greater than the previous established view ID in $v.ID$.
The last condition of line~\ref{ln:switch} guarantees that $p_\ell$ will not execute lines ~\ref{ln:fetchCrd} to \ref{ln:insC} and thus will not change its $rep.(state, input, msg)$ fields, until all the expected followers of the proposed view have sent their replicas.
Followers that receive the proposal will accept it, since none of the conditions that existed change and so the new view proposal enforces that  $valCrd=\{p_\ell\}$.
Moreover, the proposal satisfies the condition of line~\ref{ln:repF} and the followers of the view enter status {\sf Propose} leading to the installation of the view.
What is important is that virtual synchrony is preserved since no follower is changing $rep.(state, input, msg)$ during this procedure, and moreover each sends its replica to $p_\ell$ by line~\ref{ln:sendSet}.
Once the replicas of all the followers have been collected, the coordinator creates a consolidated $state$ and $msg$ array of all messages that were either delivered or pending.
%Interfaces $synchState()$ and $synchMsg$ return a consolidated state and message array, based on the message and the $view$ and round number $rnd$ it was sent in.
$p_\ell$'s new replica is communicated to the followers who adopt this state as their own (line~\ref{ln:adoptRep}).
Thus virtual synchrony is preserved and once all the processors have replicated the state of the coordinator, a new series of multicast rounds can begin by producing the side effects required by the input collected before the view change.\\
\end{proof}

\begin{claim}
\label{thVS:VSchangeCrd}
VS is preserved between $r$ and $r'$ where $v\neq v'$ and $p_\ell \neq p_{\ell'}$.
\end{claim}
%\textbf{Step 3: Virtual synchrony is preserved in two consecutive view installations where the coordinator changes.}\\
\begin{proof}
We assume that $p_\ell$ had a supporting majority throughout $R'$. 
We define a matching suffix $R''$ to prefix $R'$, such that $R''$ results from the loss of supporting majority by $p_\ell$.
Notice that since Definition~\ref{def:Primary} is required to hold, then some other processor with supporting majority $p_{\ell'}$,  will by Lemma~\ref{th:VSstab} propose the view $v'$ with the highest view ID.
We note that by the intersection property and the fact that a view set can only be formed by a majority set, $\exists p_i \in v \cap v'$. 
Thus, the ``knowledge'' of the system, $(state, input, msg)$ is retained within the majority.

As detailed in step 2, if a processor $p_i$ had $noCrd = {\sf True}$ for some time or was in status {\sf Propose} it did not incur any changes to its replica.
If it entered the {\sf Install} phase, then this implies that the proposing processor has created a consolidated state that $p_i$ has replicated.
What is noteworthy is that whether in status ${\sf Propose}$ or ${\sf Install}$, if the proposer collapses (becomes inactive or suspected), the virtual synchrony property is preserved.
It follows that, once status {\sf Multicast} is reached by all followers, the system can start a practically infinite number of multicast rounds.

Thus, by the self-stabilization property of all the components of the system (counter increment algorithm, the data links, the failure detector and multicast) a legal execution is reached in which the virtual synchrony property is guaranteed and common state replication is preserved.
content...
\end{proof}
\end{proof}

%\Paragraph{\bl{Algorithm Complexity}}
\subsection{\bl{Algorithm Complexity}}
\bl{The local memory for this algorithm consists of $n$ copies of two labels, of the encapsulated state (say of size $|S|$ bits) and of other lesser size variables. 
These give a \emph{space complexity} of order $\bigO(n\beta \log\beta + n|S|)$; recall that $\beta=n^3cap+2n^2-2n$.
\emph{Stabilization time} can be provided by a bound on view creations.
It is, therefore, implicit that stabilization is dependent upon the stabilization of the counter algorithm, i.e., $\bigO(n\cdot\beta\cdot t)$, before processors can issue views with identifiers that can be totally ordered.
When this is satisfied, then Lemma~\ref{th:boundProp} suggests that $O(n)$ view creations are required to acquire a coordinator, namely, in the worse case where every processor is a proposer.
Once a coordinator is established then Theorem~\ref{th:VSgood} guarantees that there can be practically infinite multicast rounds (0 to $2^\tau$). 
}

\section{Conclusion}
\label{s:concl} %\vspace{-1em}
State-machine replication (SMR) is a service that simulates finite automata by letting the participating processors to periodically exchange messages about their current state as well as the last input that has led to this shared state. Thus, the processors can verify that they are in sync with each other. A well-known way to emulate SMRs is to use reliable multicast algorithms that guarantee {\em virtual synchrony}~\cite{DBLP:conf/wdag/KhazanFL98,DBLP:journals/jsa/Bartoli04}. To this respect, we have presented the first \bl{practically-}self-stabilizing algorithm that guarantees virtual synchrony, and used it to obtain a \bl{practically-}self-stabilizing SMR emulation; within this emulation, the system progresses in more extreme asynchronous executions in contrast to consensus-based SMRs, like the one in~\cite{DBLP:conf/netys/BlanchardDBD14}. One of the key components of the virtual synchrony algorithm is a novel \bl{practically-}self-stabilizing counter algorithm, that establishes an efficient practically unbounded counter, which in turn can be directly used to implement a \bl{practically-}self-stabilizing MWMR register emulation; this extends the work in~\cite{Alon2014} that implements SWMR registers and can also be considered simpler and more communication efficient than the MWMR register implementation presented in~\cite{DBLP:conf/netys/BlanchardDBD14}.   
\label{s:dis}
%
%\Subsection{Summary}
%\Subsection{Future work}
%\Subsection{Conclusions}

\subsection*{Acknowledgements}
We thank Iosif Salem for providing comments to improve the readability of this paper. 
\bl{We also thank the anonymous reviewers whose constructive feedback has
helped us to significantly improve the presentation of our results and of the manuscript
in general.}

%\newpage

%\section*{References}

\bibliographystyle{plain}
\bibliography{vs,SelfStabil}

\begin{thebibliography}{10}

\bibitem{Alon2014}
Noga Alon, Hagit Attiya, Shlomi Dolev, Swan Dubois, Maria Potop{-}Butucaru, and
  S{\'{e}}bastien Tixeuil.
\newblock Practically stabilizing {SWMR} atomic memory in message-passing
  systems.
\newblock {\em J. Comput. Syst. Sci.}, 81(4):692--701, 2015.

\bibitem{DBLP:journals/jpdc/AroraKD06}
Anish Arora, Sandeep~S. Kulkarni, and Murat Demirbas.
\newblock Resettable vector clocks.
\newblock {\em J. Parallel Distrib. Comput.}, 66(2):221--237, 2006.

\bibitem{ABD}
Hagit Attiya, Amotz Bar-Noy, and Danny Dolev.
\newblock Sharing memory robustly in message-passing systems.
\newblock {\em J. ACM}, 42(1):124--142, January 1995.

\bibitem{DBLP:journals/jsa/Bartoli04}
Alberto Bartoli.
\newblock Implementing a replicated service with group communication.
\newblock {\em Journal of Systems Architecture}, 50(8):493--519, 2004.

\bibitem{DBLP:conf/replication/Birman10}
Ken Birman.
\newblock A history of the virtual synchrony replication model.
\newblock In {\em Replication: Theory and Practice}, volume 5959 of {\em
  Lecture Notes in Computer Science}, pages 91--120. Springer, 2010.

\bibitem{birmanbook}
Keneth Birman and Robbert~Van Renesse.
\newblock {\em Reliable distributed computing with the Isis toolkit}.
\newblock Wiley-IEEE Computer society press Los Alamitos, 1994.

\bibitem{DBLP:conf/netys/BlanchardDBD14}
Peva Blanchard, Shlomi Dolev, Joffroy Beauquier, and Sylvie Dela{\"{e}}t.
\newblock Practically self-stabilizing paxos replicated state-machine.
\newblock In {\em Networked Systems - Second International Conference, {NETYS}
  2014, Marrakech, Morocco, May 15-17, 2014. Revised Selected Papers}, pages
  99--121, 2014.

\bibitem{Blanchard2013SSPaxos}
Peva Blanchard, Shlomi Dolev, Joffroy Beauquier, and Sylvie Dela{\"{e}}t.
\newblock Practically self-stabilizing paxos replicated state-machine.
\newblock In {\em In Revised Selected Papers of the Second International
  Conference on Networked Systems, {NETYS} 2014}, volume 8593 of {\em Lecture
  Notes in Computer Science}, pages 99--121. Springer, 2014.

\bibitem{DBLP:journals/dc/BurnsGM93}
James~E. Burns, Mohamed~G. Gouda, and Raymond~E. Miller.
\newblock Stabilization and pseudo-stabilization.
\newblock {\em Distributed Computing}, 7(1):35--42, 1993.

\bibitem{DBLP:journals/jacm/ChandraHT96}
Tushar~Deepak Chandra, Vassos Hadzilacos, and Sam Toueg.
\newblock The weakest failure detector for solving consensus.
\newblock {\em J. {ACM}}, 43(4):685--722, 1996.

\bibitem{DBLP:conf/podc/Delporte-GalletFGHKT04}
Carole Delporte{-}Gallet, Hugues Fauconnier, Rachid Guerraoui, Vassos
  Hadzilacos, Petr Kouznetsov, and Sam Toueg.
\newblock The weakest failure detectors to solve certain fundamental problems
  in distributed computing.
\newblock In {\em Proceedings of the Twenty-Third Annual {ACM} Symposium on
  Principles of Distributed Computing, {PODC} 2004, St. John's, Newfoundland,
  Canada, July 25-28, 2004}, pages 338--346, 2004.

\bibitem{Dijkstra74}
Edsger~W Dijkstra.
\newblock Self-stabilizing systems in spite of distributed control.
\newblock {\em Communications of the ACM}, 17(11):643--644, 1974.

\bibitem{D2K}
Shlomi Dolev.
\newblock {\em Self-stabilization}.
\newblock The MIT press, 2000.

\bibitem{DBLP:journals/ipl/DolevDPT11}
Shlomi Dolev, Swan Dubois, Maria Potop{-}Butucaru, and S{\'{e}}bastien Tixeuil.
\newblock Stabilizing data-link over non-fifo channels with optimal
  fault-resilience.
\newblock {\em Inf. Process. Lett.}, 111(18):912--920, 2011.

\bibitem{DBLP:conf/sss/DolevHSS12}
Shlomi Dolev, Ariel Hanemann, Elad~Michael Schiller, and Shantanu Sharma.
\newblock Self-stabilizing end-to-end communication in (bounded capacity,
  omitting, duplicating and non-fifo) dynamic networks - (extended abstract).
\newblock In {\em Proceedings of the 14th International Symposium
  Stabilization, Safety, and Security of Distributed Systems, {SSS} 2012},
  pages 133--147, 2012.

\bibitem{DolevKS2010ConsMeetsSS}
Shlomi Dolev, Ronen~I. Kat, and Elad~M. Schiller.
\newblock When consensus meets self-stabilization.
\newblock {\em Journal of Computer and System Sciences}, 76(8):884 -- 900,
  2010.

\bibitem{DBLP:dblp_conf/sss/DolevLLN05}
Shlomi Dolev, Limor Lahiani, Nancy~A. Lynch, and Tina Nolte.
\newblock Self-stabilizing mobile node location management and message routing.
\newblock In {\em Self-Stabilizing Systems}, pages 96--112, 2005.

\bibitem{RAMBO}
Seth Gilbert, Nancy~A. Lynch, and Alexander~A. Shvartsman.
\newblock Rambo: a robust, reconfigurable atomic memory service for dynamic
  networks.
\newblock {\em Distributed Computing}, 23(4):225--272, 2010.

\bibitem{DBLP:conf/wdag/KhazanFL98}
Roger Khazan, Alan Fekete, and Nancy~A. Lynch.
\newblock Multicast group communication as a base for a load-balancing
  replicated data service.
\newblock In Shay Kutten, editor, {\em Distributed Computing, 12th
  International Symposium, {DISC} '98, Andros, Greece, September 24-26, 1998,
  Proceedings}, volume 1499 of {\em Lecture Notes in Computer Science}, pages
  258--272. Springer, 1998.

\bibitem{DBLP:dblp_journals/cacm/Lamport78}
Leslie Lamport.
\newblock Time, clocks, and the ordering of events in a distributed system.
\newblock {\em Commun. {ACM}}, 21(7):558--565, 1978.

\bibitem{LSFTC97}
Nancy~A. Lynch and Alexander~A. Shvartsman.
\newblock Robust emulation of shared memory using dynamic quorum-acknowledged
  broadcasts.
\newblock In {\em The 27th Annual International Symposium on Fault-Tolerant
  Computing (FTC 1997)}, pages 272--281, 1997.

\bibitem{SalemSchiller2017}
Iosif Salem and Elad~M. Schiller.
\newblock Practically-stabilizing vector clocks.
\newblock Technical Report~05, Dept. of Computer Science and Engineering,
  Chalmers Univ. of Technology, Rannvagen 6B, 412 96 (Goteborg) Sweden, 2017.
\newblock ISSN 1652-926X.

\bibitem{DBLP:journals/csur/Schneider90}
Fred~B. Schneider.
\newblock Implementing fault-tolerant services using the state machine
  approach: {A} tutorial.
\newblock {\em {ACM} Comput. Surv.}, 22(4):299--319, 1990.

\end{thebibliography}

\end{document}